\PassOptionsToPackage{backref,pagebackref}{hyperref}
\documentclass[11pt]{article}
\usepackage[top=30truemm,bottom=30truemm,left=25truemm,right=25truemm]{geometry}

\usepackage[utf8]{inputenc} % allow utf-8 input
\usepackage[T1]{fontenc}    % use 8-bit T1 fonts
\usepackage{textcomp}
\usepackage{bm}
\usepackage{bbm}  % \mathbbm

\usepackage{amsmath}
\usepackage{amssymb}
\usepackage{amsthm}
\usepackage{amsfonts}
\usepackage{mathrsfs}
\usepackage{nicefrac}
\usepackage{mathtools}

\usepackage{booktabs}
\usepackage{multirow}

\usepackage[authoryear,round]{natbib}

\usepackage[svgnames]{xcolor}
\usepackage{colortbl}  % for \rowcolor{blue!10}
\usepackage[hypertexnames=false]{hyperref}       % hyperlinks
\hypersetup{
  colorlinks=true,
  linkcolor=blue,
  citecolor=blue,
}

\usepackage{stmaryrd}
\usepackage{tgtermes} % times font
\usepackage{comment}

\usepackage[ruled,vlined,linesnumbered]{algorithm2e}
\SetArgSty{textnormal}
\usepackage{microtype}
\usepackage[labelfont=bf]{caption}
\usepackage{wrapfig}
\usepackage{tikz}

\usepackage{url}

\usepackage{enumitem}

\usepackage{thm-restate}

\usepackage[normalem]{ulem}
\usepackage{tcolorbox}
\usepackage{mdframed} 

\usepackage{threeparttable}

\usepackage[nameinlink]{cleveref}
\crefname{equation}{}{}
\Crefname{algocf}{Algorithm}{Algorithms}  % needed for algorithm

\theoremstyle{plain}
\newtheorem{theorem}{Theorem}
\newtheorem{proposition}[theorem]{Proposition}
\newtheorem{lemma}[theorem]{Lemma}
\newtheorem{corollary}[theorem]{Corollary}
\theoremstyle{definition}
\newtheorem{definition}{Definition}

\newtheorem{remark}{Remark}

\renewcommand{\hat}{\widehat}
\renewcommand{\tilde}{\widetilde}
\renewcommand{\epsilon}{\varepsilon}

\def\E{\mathbb{E}}
\def\P{\mathbb{P}}
\def\R{\mathbb{R}}

\def\N{\mathbb{N}}

\def\calA{\mathcal{A}}

\def\calK{\mathcal{K}}

\def\calM{\mathcal{M}}

\def\calS{\mathcal{S}}

\DeclareMathOperator*{\argmax}{arg\,max}
\DeclareMathOperator*{\argmin}{arg\,min}

\DeclareMathOperator*{\diag}{diag}
\DeclareMathOperator*{\tr}{tr}

\DeclareMathOperator{\interior}{int}

\DeclarePairedDelimiter{\abs}{\lvert}{\rvert} %
\DeclarePairedDelimiter{\brk}{[}{]}

\DeclarePairedDelimiter{\set}{\{}{\}}
\DeclarePairedDelimiter{\prn}{(}{)}
\DeclarePairedDelimiter{\nrm}{\|}{\|}

\DeclarePairedDelimiter{\inpr}{\langle}{\rangle}  % inner product

\newcommand{\zeros}{\mathbf{0}}
\newcommand{\ones}{\mathbf{1}}

\newcommand{\ie}{\textit{i.e.,}}
\newcommand{\eg}{\textit{e.g.,}}

\newcommand{\Reg}{\mathsf{Reg}}
\newcommand{\SwapReg}{\mathsf{SwapReg}}

\renewcommand{\t}{{(t)}}

\newcommand{\tp}{{(t+1)}}
\newcommand{\tm}{{(t-1)}}
\newcommand{\s}{{(s)}}

\newcommand{\sm}{{(s-1)}}

\newcommand{\logp}{\log_+}  % max{log(x), 4 Amax}
\newcommand{\ytil}{\tilde{y}}
\newcommand{\util}{\tilde{u}}
\newcommand{\ybm}{\bm{y}}
\newcommand{\hbm}{\bm{h}}
\newcommand{\zbm}{\bm{z}}

\newcommand{\xrm}{x}
\newcommand{\yrm}{y}

\newcommand{\mx}{m_\xrm}
\newcommand{\my}{m_\yrm}

\newcommand{\Qmat}{Q}

\newcommand{\ltil}{\widetilde{\ell}}
\newcommand{\gtil}{\widetilde{g}}

\newcommand{\lsto}{\bar{\ell}}

\newcommand{\xhat}{\widehat{x}}
\newcommand{\yhat}{\widehat{y}}

\newcommand{\Cx}{C_\xrm}
\newcommand{\Cy}{C_\yrm}
\newcommand{\Chatx}{\hat{C}_\xrm}
\newcommand{\Chaty}{\hat{C}_\yrm}
\newcommand{\Chati}{\hat{C}_i}
\newcommand{\Chatk}{\hat{C}_k}

\newcommand{\Ctilx}{\tilde{C}_\xrm}
\newcommand{\Ctily}{\tilde{C}_\yrm}

\newcommand{\mmax}{m}

\newcommand{\nn}{\nonumber\\}
\newcommand{\n}{\nonumber}
\newcommand{\per}{\,.}
\newcommand{\com}{\,,}

\newcommand{\sumT}{\sum_{t=1}^T}

\newcommand{\nsqrt}[1]{\sqrt{\smash[b]{#1}}\vphantom{#1}} % nice squared root

\title{Corrupted Learning Dynamics in Games}

\author{
  Taira Tsuchiya\footnote{
    The University of Tokyo and RIKEN; 
    \texttt{tsuchiya@mist.i.u-tokyo.ac.jp}.
  }
  \and
  Shinji Ito\footnote{
    The University of Tokyo and RIKEN; \texttt{shinji@mist.i.u-tokyo.ac.jp}.
  }
  \and
  Haipeng Luo\footnote{
    University of Southern California; \texttt{haipengl@usc.edu}.
  }
}

\begin{document}
\maketitle

\begin{abstract}
Learning in games refers to scenarios where multiple players interact in a shared environment, each aiming to minimize their regret. It is well known that an equilibrium can be computed at a fast rate of $O(1/T)$ when all players follow the optimistic follow-the-regularized-leader (OFTRL). However, this acceleration is limited to the \textit{honest regime}, in which all players fully adhere to a prescribed algorithm---a situation that may not be realistic in practice. To address this issue, we present \textit{corrupted learning dynamics} that adaptively find an equilibrium at a rate that depends on the extent to which each player deviates from the strategy suggested by the prescribed algorithm. First, in two-player zero-sum corrupted games, we provide learning dynamics for which the external regret of $x$-player (and similarly for $y$-player) is roughly bounded by $O(\log (m_x m_y) + \sqrt{\hat{C}_y} + \hat{C}_x)$, where $m_x$ and $m_y$ denote the number of actions of $x$- and $y$-players, respectively, and $\hat{C}_x$ and $\hat{C}_y$ represent their cumulative deviations. We then extend our approach to multi-player general-sum corrupted games, providing learning dynamics for which the swap regret of player $i$ is bounded by $O(\log T + \sqrt{\sum_{k} \hat{C}_k \log T} + \hat{C}_i)$ ignoring dependence on the number of players and actions, where $\hat{C}_i$ is the cumulative deviation of player $i$ from the prescribed algorithm. Our learning dynamics are agnostic to the levels of corruption. A key technical contribution is a new analysis that ensures the stability of a Markov chain under a new adaptive learning rate, thereby allowing us to achieve the desired bound in the corrupted regime while matching the best existing bound in the honest regime. Notably, our framework can be extended to address not only corruption in strategies but also corruption in the observed expected utilities, and we provide several matching lower bounds. 
\end{abstract}

\newpage

\section{Introduction}\label{sec:introduction}
Learning in games refers to settings where multiple players interact in a shared environment, each aiming to minimize their regret by iteratively adapting their strategies based on repeated interactions~\citep{freund99adaptive,hart00simple}. 
Each player can minimize their own regret using the framework of \emph{online learning}, which allows us to minimize the cumulative loss based on sequentially obtained past observations~\citep{cesabianchi06prediction}.
It is well known that an approximate equilibrium can be obtained when all players employ no-regret algorithms.
For instance, in two-player zero-sum games, if each player employs an online learning algorithm with regret $\Reg^T$, then an $O(\Reg^T / T)$-approximate Nash equilibrium is attained after $T$ rounds.
Since many online learning algorithms achieve regret bounds of $\Reg^T = O(\sqrt{T})$, this implies that an $O(1/\sqrt{T})$-approximate equilibrium is obtained after $T$ rounds. 
In online learning, 
the regret bound of $\Reg^T = O(\sqrt{T})$ is generally unimprovable even for stochastic losses~\citep{cesabianchi06prediction}.

However, in the context of learning in games, the observations of each player are determined by the strategies of their opponents. 
By leveraging this property, it has been shown that the $O(\sqrt{T})$ individual regret upper bounds, and the corresponding $O(1/\sqrt{T})$ convergence rates, can be significantly improved, a phenomenon first observed by~\citet{daskalakis11near}. 
A prominent approach to achieving fast rates is the use of optimistic prediction, which involves predicting the next observation based on past observations.
These include optimistic follow-the-regularized-leader (OFTRL) and optimistic online mirror descent~\citep{rakhlin13online,rakhlin13optimization,syrgkanis15fast}.
For example, in two-player zero-sum games, if both players employ a specific OFTRL algorithm with a constant learning rate, each player can achieve a regret bound of $O(\log (\mx \my))$, independent of $T$, where $\mx$ and $\my$ denote the number of actions of $x$- and $y$-players, respectively~\citep{syrgkanis15fast}.

However, such improvements can only be achieved in the \emph{honest regime}, where all players fully adhere to the prescribed algorithm. 
A player may deviate from the output of the prescribed algorithm, for example, to address their own constraints or to amplify the regret of other players.
When an opponent's sequence of actions deviates entirely from the algorithm's output---a scenario often referred to as the adversarial regime in the literature---approaches that rely on the assumption of perfect adherence to the prescribed algorithm may fail to achieve sublinear regret.

Several approaches are known to address such adversarial opponents. 
For example, in two-player zero-sum games, by appropriately switching algorithms, it is possible to simultaneously achieve regret bounds of $O(\log(\mx \my))$ in the honest regime and $\tilde{O}(\sqrt{T})$ against adversarial opponents~\citep{syrgkanis15fast}.
However, such theoretical guarantees in the adversarial scenario can be very pessimistic.
Even when a player's strategies deviate only slightly from the output of the prescribed algorithm, we can only ensure the $\tilde{O}(\sqrt{T})$ bound. 
For example, an opposing player might initially fail to follow the prescribed algorithm (in which case the algorithm in \citealt{syrgkanis15fast} switches to one designed to handle dishonest players), but later revert to honest behavior.
Ideally, it is desirable to design learning dynamics that achieve regret bounds that smoothly bridge the $\tilde{O}(1)$ bound in the honest regime and the $\tilde{O}(\sqrt{T})$ worst-case bound.

% \paragraph{Contributions of this paper}
\paragraph{Our contributions}
% ======================================
\begin{table*}[t]
      \caption{Comparison of individual (external) regret upper bounds of $x$-player in two-player zero-sum games with a payoff matrix $A \in [-1, 1]^{\mx \times \my}$ after $T$ rounds.
      The variables $\Chatx, \Chaty \in [0, 2T]$ are the amount of corruption in strategies for $x$- and $y$-players, respectively, $g^\t\in [-1 ,1]^{\mx}$ is a utility vector for $x$-player, and $P_\infty^T(g) = \sumT \nrm{g^\t - g^\tm}_\infty^2$.
    }
    \label{table:regret}
    % \centering
    % \footnotesize
    \small
    \begin{tabular}{@{}l@{\hspace{1ex}}l@{\hspace{1ex}}l@{}}
    % \begin{tabular}{lll}
        \toprule
      References & Honest & Corrupted (no corruption in observed utilities)
      \\
      \midrule
      \citet{rakhlin13optimization} & $\log (\mx \my T)$ & $\sqrt{P_\infty^T(g)} \log(\mx T) + \Chatx$ % \!\!\!\! & N/A % & OMD with adaptive lr 
      \\
      \citet{kangarshahi18lets} & $\log(\mx \my T)$ & $ \sqrt{T \log \mx} + \Chatx$ % & N/A % & FTRL, constant lr
      \\
      \citet{syrgkanis15fast}  & $\log(\mx \my)$ & $\log(\mx \my) + \sqrt{T \log \mx} + \Chatx$ % & N/A % & FTRL w/ const lr, \textbf{switch}
      \\  
      \textbf{This work (\Cref{thm:indiv_reg_corrupt})} & \!\!\! $\sqrt{\log(\mx \my) \log \mx} $ & 
      $\min \set[\Big]{ 
        \!
        \sqrt{ \!
          \prn{
            \log (\mx \my)
            \!+\!
            \Chatx 
            \!+\!
            \Chaty
          } \!
          \log \mx
        }
        , \!
        \sqrt{P_\infty^T(g) \! \log \mx}
      }
      \!+\!
      \Chatx
      $ 
      \\  
      \bottomrule
    \end{tabular}
  \end{table*}
% ====================================== 

% contributions of our paper
To achieve this goal, we present \emph{corrupted learning dynamics} that yield regret upper bounds which adapt to the degree of deviation from the prescribed algorithm. 
We introduce the \emph{corrupted regime} for $n$-player games,
characterized by corruption levels $\set{\Chati}_{i \in [n]}$.
In the corrupted regime, each player $i \in [n]$ accumulates a deviation of $\Chati \geq 0$ from the prescribed algorithm. 
When $\Chati = 0$ for all $i$, this regime corresponds to the honest regime.
See \Cref{sec:three_regimes} for details.

We propose learning dynamics that can effectively handle the corrupted regime for both two-player zero-sum games and multi-player general-sum games.
The proposed learning dynamics also achieve regret bounds that are comparable to or even better than the best existing dynamics in the honest regime. 
Our dynamics builds on OFTRL with adaptive learning rates. 
For each problem setup, we provide the following regret bounds and corresponding convergence rates to an equilibrium.

\paragraph{Two-player zero-sum games}
We start with two-player zero-sum games and provide learning dynamics with the following regret guarantees:
\begin{theorem}[Informal version of \Cref{thm:indiv_reg_corrupt}]\label{thm:indiv_reg_corrupt_informal}
  In two-player zero-sum games,
  there exists $(\Chatx, \Chaty)$-agnostic learning dynamics such that
  the (external) regret of $x$-player is bounded by
  $
  \sqrt{
    \log(\mx)
    \log(\mx \my)
  }
  $ in the honest regime 
  and 
  by
  $
  \min\set*{
    \sqrt{
      \log(\mx)
      \prn[\big]{
        \log(\mx \my)
        +
        \Chatx + \Chaty
      }
    }
    , 
    \sqrt{T \log \mx}
    }
    +
    \Chatx
  $
  in the corrupted regime,
  where $\Chatx$ and $\Chaty$ are the cumulative amount of corruption in strategies for $x$- and $y$-players, respectively.
  % The regret of $y$-player $\Reg_\yrm^T$ is bounded by a similar quantity.
  Similarly, the regret of $y$-player, $\Reg_\yrm^T$, is bounded by a similar quantity.
\end{theorem}
This is the first external regret upper bound for the corrupted regime.
Note that in the literature, the \emph{adversarial regime for $x$-player} typically refers to the case where $\Chatx = 0$ and $\Chaty$ is arbitrary,
whereas here we provide a bound for the more general setting with $\Chatx \geq 0$.
This regret guarantee incentivizes players to follow the prescribed dynamics: any deviation by an opponent from the algorithm's output incurs only a square-root penalty, whereas a deviation by a player from the prescribed algorithm incurs a linear penalty.
A comparison with existing bounds is provided in \Cref{table:regret}.
From the regret bound in \Cref{thm:indiv_reg_corrupt_informal} and the relationship between no-external-regret learning in online linear optimization and an (approximate) Nash equilibrium in \Cref{thm:reg2nasheq} (see next section),
it follows that the average play after $T$ rounds is an $\tilde{O}( \prn{\Chatx + \Chaty} / T)$-approximate Nash equilibrium in the corrupted regime, where we ignore the dependence on $\mx$ and $\my$.
\begin{comment}
% we can establish the following convergence rate:
\begin{corollary}
  There exists a learning dynamic such that, when used by all players, the average play after $T$ rounds is an $\tilde{O}( \prn{\Chatx + \Chaty} / T)$-approximate Nash equilibrium in the corrupted regime, where we ignore the dependence on $\mx$ and $\my$.
\end{corollary}
\end{comment}
% one can see that,
% there exists a learning dynamic such that, 
% when the above algorithm is used by all players, 

% ======================================
\begin{table*}[t]
  % \begin{threeparttable}[t!]
      \caption{Comparison of individual swap regret upper bounds of player $i$ in multi-player general-sum games with $n$-players and $m$-actions after $T$ rounds.
      The variable $\Chati \in [0, 2T]$ is the cumulative amount of corruption in strategies for player $i$, and $\hat{S} = \sum_{i \in [n]} \Chati$.
    }
    \label{table:regret_generalsum}
    \centering
    % \footnotesize
    \small
    \begin{tabular}{@{}l@{\hspace{1ex}}l@{\hspace{1ex}}l@{}}
    % \begin{tabular}{lll}
      \toprule
      References & Honest & Corrupted (no corruption in observed utilities)
      \\
      \midrule
      \citet{chen20hedging} &  $\sqrt{n} \prn{m \log m}^{3/4} T^{1/4}$ & $\sqrt{m T \log m} + \Chati$ % & N/A 
      \\
      \citet{anagnostides22near} &  $n m^4 \log m \log^4 T$ & N/A % & N/A  
      \\
      \citet{anagnostides22uncoupled} &  $n m^{5/2} \log T$ & $n m^{5/2} \log T + \sqrt{T m \log m} + \Chati$ % & N/A
      \\  
      \textbf{This work (\Cref{thm:indiv_swapreg})} & 
      $n m^{5/2} \log T$ & 
      $n m^{5/2} \log T 
      \!+\! 
      \min\set[\Big]{ \!
        \sqrt{\hat{S} \prn{n m^2 \!+\! m^{5/2}} \log T}
        ,
        m \sqrt{T \log T}
      }
      \!+\!
      \Chati
      $ 
      \\
      \bottomrule
    \end{tabular}
  \end{table*}
  % \end{threeparttable}
  % ======================================

\paragraph{Multi-player general-sum games}
Building on the analysis for two-player zero-sum games, we propose new corrupted learning dynamics for multi-player general-sum games. In multi-player general-sum games, to obtain a \emph{correlated equilibrium}~\citep{aumann74subjectivity,foster97calibrated,hart00simple}, we focus on minimizing the \emph{swap regret} to get the following guarantees:

\begin{theorem}[Informal version of \Cref{thm:indiv_swapreg}]\label{thm:indiv_swapreg_informal}
  In multi-player general-sum games with $n$-players and $m$-actions,
  there exists $\set{\Chati}_{i \in [n]}$-agnostic learning dynamics such that the swap regret of each player $i$
  % , $\SwapReg_i^T$, 
  is at most 
  $
  n m^{5/2} \log T 
  $
  in the honest regime 
  and
  $
  n m^{5/2} \log T 
  +
  \min\set[\Big]{
    \sqrt{
      \hat{S}
      \prn{
        n
        m^2
        +
        m^{5/2}
      }
      \log T
    }
    ,
    m \sqrt{T \log T}
  }
  +
  \Chati 
  $
  in the corrupted regime,
  where $\Chati$ is the cumulative amount of corruption in strategies of player $i$ and 
  $\hat{S} = \sum_{i \in [n]} \Chati$.
\end{theorem}
This is the first swap regret upper bound for the corrupted regime.
Note again that in the literature, the adversarial regime for player $i$ typically refers to the corrupted regime where $\Chati = 0$ and $\hat{C}_j$ is arbitrary for all $j \neq i$, whereas here we provide upper bounds for the more general setting with $\Chati \geq 0$.
Compared to the best swap regret bound by~\citet{anagnostides22uncoupled}, our algorithm achieves the same bound in the honest regime, a new adaptive bound in the corrupted regime in terms of $\hat{S}$, and a worst-case $\sqrt{T}$-type bound that is $\sqrt{m}$ time worse than theirs (which is due to the use of the adaptive learning rate).
A comparison with existing swap regret bounds is summarized in \Cref{table:regret_generalsum}.
Our algorithm is a variant of the algorithm by~\citet{anagnostides22uncoupled}, who use OFTRL with a constant learning rate~\citep{syrgkanis15fast} and the reduction of swap regret minimization to external regret minimization~\citep{blum07external}. 
From the swap regret upper bound in \Cref{thm:indiv_swapreg_informal} and the relationship between no-swap-regret learning in online linear optimization and a correlated equilibrium in \Cref{thm:swapreg2coreq} (see next section),
it follows that when the above algorithm is used by all players, 
the time-averaged history of joint play
after $T$ rounds is an $O\prn[\big]{\prn{\log T + \sqrt{\hat{S} \log T} + \max_{k \in [n]} \Chatk} / T}$-approximate correlated equilibrium in the corrupted regime, where we ignore the dependence on $n$ and $m$.
% 
\begin{comment}
 we can establish the following convergence rate:
 
\begin{corollary}
  There exists a learning dynamics such that, when used by all players, 
  the time-averaged history of joint play
  after $T$ rounds is an $O(\prn{\log T + \sqrt{\hat{S} \log T} + \max_{k \in [n]} \Chatk} / T)$-approximate correlated equilibrium in the corrupted regime, where we ignore the dependence on $n$ and $m$.
\end{corollary}
\end{comment}

Achieving the swap regret bound in the corrupted regime requires a new analysis of OFTRL and the stability of Markov chains.
The existing analysis in~\citet{anagnostides22uncoupled} relies heavily on the fact that OFTRL uses a constant learning rate, and thus we cannot directly employ it when we use an adaptive learning rate. 
By setting the learning rate sufficiently small and applying an analysis similar to that of two-player zero-sum games to prove \Cref{thm:indiv_reg_corrupt_informal}, as well as the analysis in \citet{wei18more},
we can show a swap regret bound of $O(n m^8 \log T)$ in the honest regime.
However, the dependence on $m$ of this upper bound is significantly worse than the $O(n m^{5/2} \log T)$ bound by \citet{anagnostides22uncoupled}.  

To address this issue, we leverage the fact that swap regret minimization can be achieved via multiple no-external-regret algorithms~\citep{blum07external}.
Specifically, we construct a transition probability matrix from the outputs of the external regret minimizer for each action and adopt the stationary distribution of the resulting Markov chain as the final strategy. 
We consider analyzing the stability of the transition probability matrix exploiting this property (\Cref{lem:diff_local_oftrl}), which leads to the significant improvement in the dependence on the number of actions $m$.

Our problem setting and analysis can be extended to account for corruption not only in the strategies but also in the observed expected utilities. 
This setup plays a crucial role in practical applications. For a detailed definition of the problem setting and an explanation of scenarios in which such corruption occurs, see \Cref{sec:three_regimes}.
Finally, at the end of this paper, we provide several matching lower bounds (\Cref{thm:lower_bounds}) that depend on the degree of corruption in strategies and utilities, offering a partial characterization of corrupted learning (with a complete characterization left for future work).

It is worth mentioning that our dynamics is \emph{strongly uncoupled}~\citep{hart00simple,daskalakis11near}, meaning that each player requires no prior information about the game and determines their strategies solely based on past observations gathered through game interactions. 
Naturally, players are unaware of other players' utilities, and there is no communication between players.

\section{Preliminaries}\label{sec:preliminaries}
\paragraph{Notation and conventions}
For a natural number $n \in \N$, we use $[n] = \set{1, \dots, n}$.
Let $\zeros$ and $\ones$ denote the zero vector and one vector with all entries equal to $0$ and $1$, respectively.
Given a vector~$x$, 
we use $x(a)$ to denote its $a$-th element
and use $\nrm{x}_p$ to denote its $\ell_p$-norm for $p \in [1, \infty]$.
For a matrix~$A$, we use $A(k,\cdot)$ to denote its $k$-th row.
% For matrices~$A$ and $B$, we use $\inpr{A, B} = \tr\prn{A^\top B}$ to denote the inner product of matrices.
We use $\Delta(\calK)$ to denote the set of all probability distributions over $\calK$,
and use $\Delta_d = \set{ x \in [0,1]^d \colon \nrm{x}_1 = 1 }$ to denote the $(d-1)$-dimensional probability simplex.
Let $D_\psi(x,y)$ denote the Bregman divergence between $x$ and $y$ induced by a differentiable convex function $\psi$, that is,
$
  D_{\psi}(x, y) 
  = 
  \psi(x) - \psi(y) - \inpr{\nabla \psi(y),x - y}
$.
We use $\nrm{h}_{x,f} = \sqrt{h^\top \nabla^2 f(x) h}$
and $\nrm{h}_{*,x,f} = \sqrt{h^\top \prn{\nabla^2 f(x)}^{-1} h}$ to denote the local norm and its dual norm of a vector $h$ at a point $x$ with respect to a convex function $f$, respectively.
To simplify the notation, we use $f \lesssim g$ to denote $f = O(g)$.
For a sequence $z = (z^{(1)}, \dots, z^{(T)})$ and $q \in [1, \infty]$, 
we define the (squared) path-length values in terms of the $\ell_q$-norm as
$
  P_q^T(z) = \sum_{t=2}^T \nrm{z^\t - z^\tm}_q^2.
$
Throughout this paper, we frequently use $i$ as the player index, $a$ as the action index, $d$ or $m$ as the dimension of a feasible set.

\subsection{Online linear optimization and external/swap regret}\label{subsec:olo}
Online linear optimization, a fundamental area of online learning, serves as a key framework for learning in games.
In online linear optimization, a learner is given a convex set $\calK \subseteq \R^d$ before the game starts.
Then at each round $t = 1, \dots, T$,
the learner first selects a point $x^\t \in \calK$ based on past observations,
and the environment then determines a loss vector $\ell^\t \in \R^d$ without knowledge of $x^\t$.
The learner then suffers a loss of $\inpr{x^\t, \ell^\t}\in \R$ and observes the loss vector $\ell^\t$.
The most common objective for the learner is to minimize the external regret $\Reg^T$, which is defined as
$
  \Reg^T
  =
  \max_{u \in \calK} 
  \Reg^T(u) 
$
  for
$
  \Reg^T(u)
  =
  \sumT \inpr{x^\t - u, \ell^\t}  
  .
$
We will see below that minimizing this external regret relates to the problem of finding an approximate Nash equilibrium.

Another important objective of the learner is to minimize the swap regret.
Here, we consider only the case where the feasible set is the probability simplex.
Let
$\calM_d
= 
\set{ M \in [0,1]^{d \times d} \colon M(k, \cdot) \in \Delta_d \mbox{ for } k \in [d]}
$ 
be the set of all $d \times d$ row stochastic matrices, which we also refer to as transition probability matrices.
Then, the swap regret is defined as
$
  \SwapReg^T
  =
  \max_{M \in \calM_{d}} 
  \SwapReg^T(M) 
$
  for
$
  \SwapReg^T(M)
  =
  \sumT \inpr{x^\t, \ell^\t - M \ell^\t}
  .
$
We will see below how minimizing the swap regret is related to the problem of finding a correlated equilibrium. 
It should be noted that the swap regret is nonnegative, a property leveraged to obtain a desirable regret in multi-player general-sum games. 

\subsection{Two-player zero-sum games}\label{subsec:two_player_preliminaries}
Here we introduce two-player zero-sum games with a payoff matrix $A \in [-1, 1]^{\mx \times \my}$,
where $\mx$ and $\my$ are the number of actions of $x$- and $y$-players, respectively.
The procedure of this game is as follows:
at each round $t = 1, \dots, T$, 
$x$-player chooses a \emph{mixed strategy} (or strategy for short) $x^\t \in \Delta_{\mx}$ and $y$-player chooses $y^\t \in \Delta_{\my}$.
Then $x$-player observes a utility vector $g^\t = A y^\t$ and gains a reward of $\inpr{x^\t, g^\t}$, and $y$-player observes a loss vector $\ell^\t = A^\top x^\t$ and incurs a loss of $\inpr{y^\t, \ell^\t}$.
The regret of $x$-player, $\Reg_{x,g}^T$ is given by
$
  \Reg_{x,g}^T 
  =
  \max_{x^* \in \Delta_{\mx}} 
  % \set[\big]
  {\Reg_{x,g}^T(x^*)}
$
  for 
  % \com \;
$
  \Reg_{x,g}^T(x^*)
  = 
  \sumT \inpr{x^* - x^\t, A y^\t}
  =
  \sumT \inpr{x^* - x^\t, g^\t}
$,
  % \com 
  % \\
and the regret of $y$-player is given by
$
  \Reg_{y,\ell}^T
  = 
  \max_{y^* \in \Delta_{\my}} 
  % \set[\big]
  {\Reg_{y,\ell}^T(y^*)}
  % \com \;
$
for 
$
  \Reg_{y,\ell}^T(y^*)
  = 
  \sumT \inpr{y^\t - y^*, A^\top x^\t}
  =
  \sumT \inpr{y^\t - y^*, \ell^\t}
  .
$
With this framework established, we now introduce the concept of a Nash equilibrium.
\begin{definition}[Nash equilibrium]
  In a two-player zero-sum game with a payoff matrix $A \in [-1, 1]^{\mx \times \my}$, a pair of probability distributions $\sigma = (x^*, y^*)$ over action sets $[\mx]$ and $[\my]$ is an \emph{$\epsilon$-approximate Nash equilibrium} for $\epsilon \geq 0$ if for any distributions $x \in \Delta_{\mx}$ and $y \in \Delta_{\my}$,
  % \begin{equation}
  $
    x^\top A y^* - \epsilon
    \leq
    {x^*}^\top A y^*
    \leq 
    {x^*}^\top A y + \epsilon
    .
  $
    % \per 
    % \n
  % \end{equation}
  The distribution $\sigma$ is a \emph{Nash equilibrium} if $\sigma$ is a $0$-approximate Nash equilibrium.
\end{definition}
The following lemma is well known in the literature and provides a connection between no-external-regret learning dynamics and learning in two-player zero-sum games.
\begin{theorem}[\citealt{freund99adaptive}]\label{thm:reg2nasheq}
  In two-player zero-sum games,
  suppose that the regrets of $x$- and $y$-players are $\Reg_{x,g}^T$ and $\Reg_{y,\ell}^T$, respectively.
  Then, the product distribution of the average play $\prn{\frac1T \sumT x^\t, \\ \frac1T \sumT y^\t}$ is a $(\prn{\Reg_{x,g}^T + \Reg_{y,\ell}^T} / T)$-approximate Nash equilibrium.
\end{theorem}

\subsection{Multi-player general-sum games}\label{subsec:multi_player_preliminaries}
We next introduce multi-player general-sum games.
Let $n \geq 2$ be the number of players and $[n] = \set{1,\dots,n}$ be the set of players.
Each player $i \in [n]$ has an action set $\calA_i$ with $\abs{\calA_i} = m_i$ and a utility function $u_i \colon \calA_1 \times \cdots \times \calA_n \to [-1,1]$.
The procedure of this game is as follows:
at each round $t = 1, \dots, T$, each player $i \in [n]$ chooses a mixed strategy $x_i^\t \in \Delta_{m_i}$ and observes a \emph{expected utility} $u_i^\t \in [-1, 1]^{m_i}$. 
Here, the $a_i$-th element of $u_i^\t$ is defined as 
% \begin{equation}\label{eq:def_uit}
  % \textstyle
$
  u_i^\t(a_i) = \E_{a_{-i} \sim x_{-i}^\t} \brk*{u_i(a_i, a_{-i})}
  ,
$
  % \com 
% \end{equation}
which is the expected reward when player $i$ selects $a_i$ and the other players choose their actions following a distribution $x_{-i}^\t = (x_1^\t, \dots, x_{i-1}^\t, x_{i+1}^\t, \dots, x_n^\t)$.
The swap regret of each player $i \in [n]$
% , $\SwapReg_i^T$, 
is given by
% \begin{equation}\label{eq:swapreg_i}
$
  \SwapReg_{x_i,u_i}^T
  =
  \max_{M \in \calM_{m_i}} 
  \SwapReg_{x_i,u_i}^T(M) 
$
for 
$
  \SwapReg_{x_i,u_i}^T(M)
  =
  \sumT \inpr{x_i^\t, M u_i^\t - u_i^\t}
  ,
$
where 
we recall that
$\calM_m
= 
\set{ M \in [0,1]^{m \times m} \colon M(k, \cdot) \in \Delta_m \mbox{ for } k \in [m]}
$ 
is the set of all $m \times m$ row stochastic matrices
and
note that here we consider the utility vector $u_i^\t$ instead of a loss vector.
Note also that two-player zero-sum games are special cases of multi-player general-sum games, and hereafter, corresponding notations may be used without clarification.
With this framework established, we can now introduce the concept of a correlated equilibrium.
\begin{definition}[Correlated equilibrium, \citealt{aumann74subjectivity}]
  A probability distribution $\sigma$ over action sets $\times_{i=1}^n \calA_i$ is an \emph{$\epsilon$-approximate correlated equilibrium} for $\epsilon \geq 0$ if 
  for any player $i \in [n]$ and any (swap) function $\phi_i \colon \calA_i \to \calA_i$ that swap action $a_i$ with $\phi_i(a_i)$,
  $
    \E_{a \sim \sigma} \brk*{ u_i(a) }
    \geq 
    \E_{a \sim \sigma} \brk*{ u_i(\phi_i(a_i), a_{-i}) }
    -
    \epsilon
    .
  $
  The distribution $\sigma$ is a \emph{correlated equilibrium} if $\sigma$ is a $0$-approximate correlated equilibrium.
\end{definition}
The following lemma is folklore and provides a connection between no-swap-regret learning dynamics and learning in multi-player general-sum games.
\begin{theorem}[\citealt{foster97calibrated}]\label{thm:swapreg2coreq}
  In multi-player general-sum games,
  suppose that the swap regret of each player $i$ is $\SwapReg_{x_i,u_i}^T$
  and let 
  $\sigma^\t = \otimes_{i \in [n]} x_i^\t \in \Delta(\times_{i=1}^n \calA_i)$ 
  given by
  $\sigma^\t(a_1, \dots, a_n) = \prod_{i\in[n]} x_i^\t(a_i)$ for each $a_i \in \calA_i$
  be the joint distribution at round $t$.
  Then, 
  its time-averaged distribution $\sigma = \frac{1}{T} \sumT \sigma^\t$
  is a $(\max_{i \in [n]} \SwapReg_{x_i,u_i}^T / T)$-approximate correlated equilibrium.
\end{theorem}

\subsection{Optimistic follow-the-regularized-leader}
Here we present the framework of OFTRL for online linear optimization introduced in \Cref{subsec:olo}, which we employ in both two-player zero-sum games and multi-player general-sum games.
At each round $t \in [T]$, OFTRL selects a point $x^\t \in \calK \subseteq \R^d$ such that
\begin{equation}
  x^\t 
  \in 
  \argmin_{x \in \calK} 
  \set*{
    \inpr*{x, m^\t + \sum_{s=1}^{t-1} \ell^\s} 
    +
    \psi^\t(x)
  }
  \n
\end{equation}
for a convex regularizer $\psi^\t$ over $\calK$ and an optimistic prediction $m^\t \in \R^d$ of the true loss vector $\ell^\t$.
The optimistic prediction $m^\t$ serves as the predicted value of $\ell^\t$ based on information available up to round $t-1$, and the closer this prediction is to the actual loss vector $\ell^\t$, the smaller the regret.
Intuitively, this property can be used to achieve fast convergence in learning in games because the loss vectors are determined by the strategies chosen by the other players. 
If all players follow a no-regret algorithm, their strategies can be predicted to some extent, enabling us to speed up convergence to an equilibrium. 
For this reason, $m^\t = \ell^\tm$ is often used as the optimistic prediction.
By employing OFTRL, we can obtain an RVU (Regret bounded by Variation in Utilities) bound, which was observed to be useful for learning in games by~\citet{syrgkanis15fast}.
See \Cref{lem:oftrl_shannon} and \Cref{lem:oftrl_logbarrier} in \Cref{app:proof_oftrl} for the RVU bounds for OFTRL with the negative Shannon entropy regularizer and with the log-barrier regularizer, respectively.

\section{Honest and Corrupted Regimes}\label{sec:three_regimes}
This section introduces the honest and corrupted regimes, 
focusing on multi-player general-sum games (with two-player zero-sum games being a special case, so all definitions apply).
Let $\xhat_i^\t \in \Delta_{m_i}$ denote the strategy suggested by the prescribed algorithm of player $i$ at round $t$, and let $x_i^\t \in \Delta_{m_i}$ denote the strategy actually chosen by player $i$ at round $t$.
A game is said to be in an \emph{honest regime} if all players fully adhere to the prescribed algorithm at every round; that is, $x_i^\t = \xhat_i^\t$ for all $i \in [n]$ and $t \in [T]$.

\paragraph{Corrupted regime}
In practice, the assumption that all players fully adhere to the prescribed algorithm may not always hold;
it is desirable for the individual regret to be bounded based on the degree of deviation from the prescribed algorithm.
Furthermore, not only the strategies but also the observed utilities may be subject to corruption for various reasons.
Motivated by the potential corruption in both strategies and utilities, we define the following corrupted regime:
\begin{definition}[Corrupted regime with corruption level 
  $\set{(\hat{C}_i, \tilde{C}_i)}_{i \in [n]}$
  % $\set{C_i}_{i \in [n]}$
  ]\label[definition]{def:corrupted_regime}
  A game is said to be in a \emph{corrupted regime with corruption level 
  % $\set{C_i}_{i \in [n]}$
  $\set{(\hat{C}_i, \tilde{C}_i)}_{i \in [n]}$
  } if the following two conditions hold:
  \begin{itemize}[topsep=0pt, itemsep=-3pt, partopsep=0pt, leftmargin=18pt] % [topsep=2pt, itemsep=0pt, partopsep=0pt, leftmargin=25pt]
    \item[(i)] 
    the strategies $\set{x_i^{(1)}, \dots, x_i^{(T)}}$ chosen by each player $i \in [n]$ deviate from outputs of the prescribed algorithm $\set{\xhat_i^{(1)}, \dots, \xhat_i^{(T)}}$ by at most $\hat{C}_i$, that is,
    % \begin{equation}\label{eq:def_ci}
    $
      \sumT
      \nrm{x_i^\t - \xhat_i^\t}_1
      \!\leq\!
      \hat{C}_i
    $
    for all $i \in [n]$;
    
    \item[(ii)]
    the utilities $\set{\tilde{u}_i^{(1)}, \dots, \tilde{u}_i^{(T)}}$ observed by each player $i \in [n]$ deviate from the expected utility $\set{u_i^{(1)}, \dots, u_i^{(T)}}$ by at most $\tilde{C}_i$, that is,
    $
      \sumT
      \nrm{u_i^\t - \tilde{u}_i^\t}_\infty
      \leq 
      \tilde{C}_i
    $
    for all $i \in [n]$.
  \end{itemize}
\end{definition}
Define $C_i = 2 \hat{C}_i + 2 \tilde{C}_i$.
% When the corruption level is clear from the context, it will be omitted. 
Then the corrupted regime with $C_i = 0$ for all $i$ corresponds to the honest regime, and the corrupted regime with arbitrary $\set{C_j}_{j \neq i}$ is the adversarial regime for player $i$.
In the following sections, we analyze the external and swap regrets in the corrupted regime for two-player zero-sum and multi-player general-sum games.
The corrupted procedure is summarized as follows:
\begin{mdframed}
At each round $t = 1, \dots, T$:
\begin{enumerate}[topsep=1pt, itemsep=-3pt, partopsep=0pt, leftmargin=25pt]
  \item A prescribed algorithm suggests a strategy
  $\xhat_i^\t \in \Delta_{m_i}$ 
  for each player $i \in [n]$;
  \item Each player $i \in [n]$ selects a strategy $x_i^\t \leftarrow \xhat_i^\t + \hat{c}_i^\t$; %  \hspace{\fill}  \# corruption of strategies
  \item Each player $i$ observes a corrupted utility vector $\tilde{u}_i^\t \leftarrow u_i^\t + \tilde{c}_i^\t$;
  \item Each player $i$ gains a reward of 
  $\inpr{x_i^\t, u_i^\t}$ in Setting (I) and $\inpr{x_i^\t, \util_i^\t}$ in Setting (II).
\end{enumerate}
\end{mdframed}
Here, $\hat{c}_i^\t$ and $\tilde{c}_i^\t$ are corruption vectors for strategies and utility, respectively, such that 
$\sumT \nrm{\hat{c}_i^\t}_1 = \sumT \nrm{ x_i^\t - \xhat_i^\t }_1 \leq \hat{C}_i$, 
$\sumT \nrm{\tilde{c}_i^\t}_\infty = \sumT \nrm{ u_i^\t - \tilde{u}_i^\t }_\infty \leq \tilde{C}_i$, and $C_i = 2 \hat{C}_i + 2 \tilde{C}_i$. 
Note that the corrupted feedback setting differs from the noisy feedback setting, where we observe feedback $u_i^\t + \xi_i^\t$ for a zero-mean noise $\xi_i^\t$~\citep{cohen17learning,hsieh22noregret,abe23last}.

\begin{remark}\label[remark]{rem:setting_AB}
As described above, 
there are two possible settings for the rewards obtained by each player: Setting~(I), where the reward is given by $\inpr{x_i^\t, u_i^\t}$, and Setting~(II), where the reward is given by $\inpr{x_i^\t, \util_i^\t}$.
Setting~(I) is natural in scenarios where communication channels or other external factors introduce noise before information from other players $j \neq i$ reaches player $i$.
In contrast, Setting~(II) is more appropriate when the true utility function $u_i$ (or the payoff matrix $A$) is time-varying.
As we will see below, the definition of regret should differ between these two settings. 
\end{remark}

In this corrupted setting, the following four types of external regret can be defined:
$\Reg_{x_i,u_i}^T(x^*) = \sumT \inpr{x^* - x_i^\t, u_i^\t}$,
$\Reg_{\xhat_i,u_i}^T(x^*) = \sumT \inpr{x^* - \xhat_i^\t, u_i^\t}$,
$\Reg_{x_i,\tilde{u}_i}^T(x^*) = \sumT \inpr{x^* - x_i^\t, \tilde{u}_i^\t}$, and
$\Reg_{\xhat_i,\tilde{u}_i}^T(x^*) = \sumT \inpr{x^* - \xhat_i^\t, \tilde{u}_i^\t}$.
As discussed in \Cref{rem:setting_AB}, depending on the cause of the corruption in the utility vector $u_i^\t$, it is natural to consider either $\Reg_{x_i,u_i}^T(x^*)$ or $\Reg_{x_i,\tilde{u}_i}^T(x^*)$ as the evaluation metric for the player.
Specifically, in Setting A, it is natural to use $\Reg_{x_i,u_i}^T(x^*)$, while in Setting B, it is natural to use $\Reg_{x_i,\util_i}^T(x^*)$.
For these external regrets, the following inequalities hold:
\begin{proposition}\label[proposition]{prop:reg_relation}
  For any $i \in [n]$ and $x^* \in \Delta_{m_i}$,
  $\abs{\Reg_{x_i,u_i}^T(x^*) - \Reg_{\xhat_i,u_i}^T(x^*)} \leq \hat{C}_i$,
  $\abs{\Reg_{x_i,\tilde{u}_i}^T(x^*) - \Reg_{\xhat_i,\tilde{u}_i}^T(x^*)} \leq \hat{C}_i$,
  $\abs{\Reg_{x_i,u_i}^T(x^*) - \Reg_{x_i,\tilde{u}_i}^T(x^*)} \leq 2 \tilde{C}_i$,
  and
  $\abs{\Reg_{\xhat_i,u_i}^T(x^*) - \Reg_{\xhat_i,\tilde{u}_i}^T(x^*)} \leq 2 \tilde{C}_i$.
\end{proposition}
The proof can be found in \Cref{app:proof_regimes}.
Following the above four regrets, define
$\SwapReg_{x_i,u_i}^T(M) = \sumT \inpr{x_i^\t, M u_i^\t - u_i^\t}$,
$\SwapReg_{\xhat_i,u_i}^T(M) = \sumT \inpr{\xhat_i^\t, M u_i^\t - u_i^\t}$,
$\SwapReg_{x_i,\util_i}^T(M) = \sumT \inpr{x_i^\t, M \util_i^\t - \util_i^\t}$, and 
$\SwapReg_{\xhat_i,\util_i}^T(M) = \sumT \inpr{\xhat_i^\t, M \util_i^\t - \util_i^\t}$.
For these swap regrets, a similar proposition as \Cref{prop:reg_relation} holds (see \Cref{app:proof_regimes} for the statement and proof).

\section{Corrupted Learning Dynamics in Two-player Zero-sum Games}\label{sec:two_player}
This section investigates learning dynamics of two-player zero-sum games in the corrupted regime (see \Cref{subsec:corrupted_game_two_player} for the procedure).
We will upper bound the external regrets of $x$- and $y$-players.

\subsection{Learning dynamics}\label{subsec:dynamics_two_player}
We use OFTRL with the Shannon entropy regularizer as the prescribed algorithm.
Let
$
\gtil^\t = g^\t + \tilde{c}_{\xrm}^\t
$
for
  % \com 
  % \quad
$
  g^\t = A y^\t
$
and 
$
\tilde{\ell}^\t = \ell^\t + \tilde{c}_{\yrm}^\t
$
for
$
\ell^\t = A^\top x^\t
$.
Then, 
% sequences of strategies 
$\set{\xhat^\t}_{t=1}^T$ and $\set{\yhat^\t}_{t=1}^T$ are given by
\begin{equation}
  \begin{split}
  \xhat^\t
  &\!=\! 
  \argmax_{x \in \Delta_{\mx}}
  \set*{
    \inpr*{x, \gtil^\tm \!+\! \sum_{s=1}^{t-1} \gtil^\s}
    \!+\!
    % \frac{1}{\eta^\t_\xrm} H(x)
    \frac{H(x)}{\eta^\t_\xrm} 
  }
  \com 
  \; 
  \eta_\xrm^\t
  \!=\!
  \sqrt{\frac{\logp(\mx) / 2}{\logp(\mx) \!+\! \sum_{s=1}^{t-1} \nrm{\gtil^\s - \gtil^\sm}_\infty^2 }}
  \com
  \\
  \yhat^\t 
  &\!=\! 
  \argmin_{y \in \Delta_{\my}}
  \set*{
    \inpr*{y, \tilde{\ell}^\tm \!+\! \sum_{s=1}^{t-1} \tilde{\ell}^\s}
    \!-\!
    % \frac{1}{\eta^\t_\yrm} H(y)
    \frac{H(y)}{\eta^\t_\yrm}
  }
  \com 
  % \quad 
  % \tilde{\ell}^\t = \ell^\t + \tilde{c}_{\yrm}^\t
  % \com 
  % \quad
  % \ell^\t = A^\top x^\t
  \;
  \eta_\yrm^\t
  \!=\!
  \sqrt{\frac{\logp(\my) / 2}{\logp(\my) \!+\! \sum_{s=1}^{t-1} \nrm{\ltil^\s - \ltil^\sm}_\infty^2 }}
  \com 
  \end{split}
  \n
\end{equation}
where $H(x) = \sum_k x(k) \log(1/x(k))$ is the Shannon entropy,
$\logp(z) = \max\set{\log z, 4}$,
and we let $g^{(0)} = \ell^{(0)} = \zeros$.
Note that the learning rates are adjusted to satisfy $\eta_\xrm^\t \leq 1/\sqrt{2}$ and $\eta_\yrm^\t \leq 1/\sqrt{2}$. % for all $t \in [T]$.

\subsection{External regret bounds and analysis}
The learning dynamics determined by the above algorithm guarantees the following bounds.
\begin{theorem}[External regret upper bounds]\label{thm:indiv_reg_corrupt}
Suppose that $x$- and $y$-players use the above algorithm to obtain strategies $\set{\xhat^\t}_{t=1}^T$ and $\set{\yhat^\t}_{t=1}^T$.
Then, in the corrupted regime, it holds that
\begin{align}
  % \textstyle
  % \Reg_\xrm^T 
  \Reg_{x,g}^T
  &
  \lesssim
  % \begin{cases}
  %   \sqrt{
  %     \log(\mx \my)
  %     \log \mx
  %   } & \!\! \mbox{honest regime} 
  %   \\ 
    \min\set[\Big]{
    \sqrt{
      \prn*{
        \log(\mx \my)
        +
        \Cx 
        +
        \Cy
      }
      \log \mx
    }
    ,
    \sqrt{
      \prn[\big]{
        % \sumT \nrm{g^\t \!-\! g^\tm}_\infty^2 
        P_\infty^T(\tilde{g})
        +
        \log \mx} \log \mx
    }
    }
    +
    \Cx
    \com 
    \label{eq:reg_xg_upper}
    \\
    % &
    % \!\!\mbox{corrupted regime}
  % \end{cases}
  % \n
% \end{equation}
% \begin{equation}
  % \textstyle
  % \Reg_\xrm^T 
  \Reg_{x,\gtil}^T
  &
  \lesssim
  % \begin{cases}
  %   \sqrt{
  %     \log(\mx \my)
  %     \log \mx
  %   } & \!\! \mbox{honest regime} 
  %   \\ 
    \min\set[\Big]{
    \sqrt{
      \prn*{
        \log(\mx \my)
        +
        \Cx 
        +
        \Cy
      }
      \log \mx
    }
    ,
    \sqrt{
      \prn[\big]{
        % \sumT \nrm{g^\t \!-\! g^\tm}_\infty^2 
        P_\infty^T(\tilde{g})
        +
        \log \mx} \log \mx
    }
    }
    +
    \hat{C}_{\xrm} 
    \per
    % &
    % \!\!\mbox{corrupted regime}
  % \end{cases}
  % \n
  \label{eq:reg_xgtil_upper}
\end{align}
For the upper bounds on $\Reg_{y,\ell}^T$ and $\Reg_{y,\ltil}^T$, the upper bounds replacing $\xrm$ with $\yrm$ and $\tilde{g}$ with $\tilde{\ell}$ in \Cref{eq:reg_xg_upper} and \Cref{eq:reg_xgtil_upper} holds, respectively. 
(We include these bounds in \Cref{app:proof_two_player} for completeness.)
\end{theorem}
A comparison with existing external regret upper bounds is provided in \Cref{table:regret}.
The above bounds are the first individual external regret bounds for the corrupted regime, and they are adaptive to the corruption levels $\Cx$ and $\Cy$.
Note that by setting $C_x = C_y = 0$ in the above bounds, we obtain $\Reg_{x,g}^T \lesssim \sqrt{\log(\mx \my) \log \mx}$ and $\SwapReg_{x,\gtil}^T \lesssim \sqrt{\log(\mx \my) \log \mx}$ in the honest regime.
It is worth noting that the regret guarantees incentivize players to follow the prescribed dynamics: any deviation by an opponent from the algorithm's output incurs only a square-root penalty, whereas a deviation by a player from the prescribed algorithm incurs a linear penalty.
The bound $O(\sqrt{\log(\mx \my) \log \mx})$ in the honest regime is slightly better than the classical bound of $O(\log(\mx \my))$.
$x$-player (and similarly for $y$-player) simultaneously achieves a path-length bound of $O\prn{\sqrt{P_\infty^T(g) \log \mx}}$ when $\Cx = 0$ against any sequence of strategies by $y$-player, which is $\prn{\log\prn{\mx T} / \sqrt{\log \mx}}$-times better than the bound of \citet{rakhlin13optimization} and also improves upon~\citet{syrgkanis15fast} as our bound is adaptive to the path-length of the $y$-player's strategies.
Finally, note that the upper bound on $\Reg_{x,g}^T$ depends linearly on the magnitude of $\hat{C}_{\xrm}$, whereas the upper bound on $\Reg_{x,\gtil}^T$ depends on $\hat{C}_{\xrm}$ via a square-root term.

We provide the proof of \Cref{thm:indiv_reg_corrupt} below.
We begin by providing the following lemma, which follows from \Cref{lem:oftrl_shannon} and the definitions of corruption levels.

\begin{lemma}\label[lemma]{lem:reg_xhat_bound}
  Suppose $x$- and $y$-players use the above algorithm to get $\set{\xhat^\t}_{t=1}^T$ and $\set{\yhat^\t}_{t=1}^T$. Then,
  \begin{equation}
    \begin{split}
      \Reg_{\xhat, \gtil}^T
      &\leq
      2 \sqrt{ 2 \logp(\mx) 
      \prn[\big]{
        % 1
        % + 
        \logp(\mx)
        +
        8 (\hat{C}_\yrm + \tilde{C}_\xrm)
        +
        % 4 \sum_{t=2}^T 
        % \nrm{\yhat^\t - \yhat^\tm}_1^2
        4 P_1^T(\yhat)
      }
      }
      - 
      \frac{1}{\sqrt{8}}
      % \sum_{t=2}^T \nrm{\xhat^\t - \xhat^\tm}_1^2
      P_1^T(\xhat)
      +
      2
      \com
      \\
      \Reg_{\yhat, \tilde{\ell}}^T
      &\leq
      2 \sqrt{ 2 \logp(\my) 
      \prn[\big]{
        % 1
        % + 
        \logp(\my)
        +
        8 \prn{ \hat{C}_\xrm + \tilde{C}_\yrm }
        +
        % 4 \sum_{t=2}^T 
        % \nrm{\xhat^\t - \xhat^\tm}_1^2
        4 P_1^T(\xhat)
      }
      }
      - 
      \frac{1}{\sqrt{8}} 
      % \sum_{t=2}^T \nrm{\yhat^\t - \yhat^\tm}_1^2
      P_1^T(\yhat)
      +
      2
      \com
    \end{split}
    \n
  \end{equation}  
  where we recall that $P_1^T(z) = \sum_{t=1}^T \nrm{z^\t - z^\tm}_1^2$ for $z = (z^{(1)}, \dots, z^{(T)})$.
\end{lemma}
Now we are ready to prove \Cref{thm:indiv_reg_corrupt}.
\begin{proof}[Proof sketch of \Cref{thm:indiv_reg_corrupt}]
Here we provide the regret upper bounds on $\Reg_{x,g}^T$ and $\Reg_{y,\ell}^T$, and the rest of the proof can be found in \Cref{app:proof_two_player}.
From \Cref{lem:reg_xhat_bound} and \Cref{prop:reg_relation},
we have % \red{（ここでは $\Reg_\xrm^T = \Reg_{x,g}^T$ になっている）}
\begin{align}% \label{eq:reg_x_bound_4}
  % \Reg_\xrm^T
  \Reg_{x,g}^T
  &\leq
  \Cx
  +
  \sqrt{ 8 \logp(\mx) 
  \prn[\big]{
    1
    + 
    \logp(\mx)
    +
    8 (\hat{C}_\yrm + \tilde{C}_\xrm)
    +
    4 P_1^T(\yhat)
  }
  }
  - 
  \frac{1}{\sqrt{8}} P_1^T(\xhat)
  +
  2
  \com
  \label{eq:reg_x_bound_4}
  \\
% \end{equation}
% \begin{equation}
  % \Reg_\yrm^T
  \Reg_{y,\ell}^T
  &\leq
  % \green{
  \Cy 
  + 
  % }
  \sqrt{ 8 \logp(\my) 
  \prn[\big]{
    1
    + 
    \logp(\my)
    +
    8 \prn{ \hat{C}_\xrm + \tilde{C}_\yrm }
    +
    4 P_1^T(\xhat)
  }
  }
  - 
  \frac{1}{\sqrt{8}} P_1^T(\yhat)
  +
  2
  \per 
  \label{eq:reg_y_bound_3}
\end{align}
Summing up the above two bounds, we can upper bound the social regret $\Reg_{x,g}^T + \Reg_{y,\ell}^T$ by
\begin{align}
  &
  \sqrt{ 8 \logp(\mx) 
  \prn*{
    1
    +
    \logp(\mx)
    + 
    8 \prn{ \hat{C}_\yrm + \tilde{C}_\xrm }
    +
    4 P_1^T(\yhat)
  }
  }
  - 
  \frac{1}{\sqrt{8}} P_1^T(\xhat)
  +
  \Cx
  +
  4
  \nn 
  &\quad+
  \sqrt{ 8 \logp(\my) 
  \prn*{
    1
    +
    \logp(\my)
    +
    8 \prn{ \hat{C}_\xrm + \tilde{C}_\yrm }
    +
    4
    P_1^T(\xhat)
  }
  }
  -
  \frac{1}{\sqrt{8}}
  P_1^T(\yhat)
  +
  \Cy
  \nn 
  &=
  O\prn*{
    \sqrt{\prn{\hat{C}_\yrm + \tilde{C}_\xrm} \log \mx}
    +
    \sqrt{\prn{\hat{C}_\xrm + \tilde{C}_\yrm} \log \my}
    +
    \log(\mx \my)
    +
    \Cx + \Cy
  }
  -
  \frac{1}{4\sqrt{2}} 
  \prn*{
    P_1^T(\xhat) + P_1^T(\yhat)
  }
  \com 
  \n
\end{align}
% where we considered the worst-case of $P_1^T(\xhat)$ and $P_1^T(\yhat)$.
where we used $b\sqrt{z} - az \leq b^2 / (4 a)$ for $a > 0$, $b \geq 0$ and $z \geq 0$.
Combining this with 
% $\Reg_\xrm^T + \Reg_\yrm^T \geq 0$, 
$\Reg_{x,g}^T + \Reg_{y,\ell}^T \geq 0$, 
which holds from the definition of the Nash equilibrium (see \Cref{lem:regx_regy_nn}),
we have 
\begin{align}
  % &
  P_1^T(\xhat) + P_1^T(\yhat)
  % =
  % O\prn*{
  % \nn
  &\lesssim
  \sqrt{\prn{\hat{C}_\yrm + \tilde{C}_\xrm}\log \mx}
  +
  \sqrt{\prn{\hat{C}_\xrm + \tilde{C}_\yrm} \log \my}
  +
  \log(\mx \my)
  +
  \Cx 
  +
  \Cy
  \nn
  &\lesssim
  \log(\mx \my)
  +
  \Cx 
  +
  \Cy
  % }
  \com 
  \label{eq:second_order_bounds}
\end{align}
where the last line follows from the AM--GM inequality and the definitions of $\Cx$ and $\Cy$.
Finally, plugging~\Cref{eq:second_order_bounds} in~\Cref{eq:reg_x_bound_4,eq:reg_y_bound_3} gives the desired bounds on $\Reg_{x,g}^T$ and $\Reg_{y,\ell}^T$.
The bound $\Reg_{x,g}^T \lesssim \sqrt{\prn{P_\infty^T(\gtil) + \log \mx} \log \mx} + \Cx$ follows from \Cref{eq:reg_x_bound_2} in the proof of \Cref{lem:reg_xhat_bound}.
\end{proof}

\section{Corrupted Learning Dynamics in Multi-player General-sum Games}\label{sec:multi_player}
This section presents learning dynamics for multi-player general-sum games in the corrupted regime.
Our algorithm for computing $\set{\xhat_i^\t}_{t=1}^T$ is a variant of the algorithm by~\citet{anagnostides22uncoupled}.

\paragraph{Reducing swap regret minimization to external regret minimization}
We begin by introducing a well-known result for swap regret minimization due to \citet{blum07external}, taking an algorithm for player $i$ for example.
They developed a method to reduce swap regret minimization to $m_i$ instances of external regret minimization. 
Specifically, they consider the following procedure.
First, for each $a \in \calA_i$, 
let $u_{i,a}^\t = \xhat_i^\t(a) u_i^\t \in [- \xhat_i^\t(a), \xhat_i^\t(a)]^{m_i}$ be the utility vector for the external regret minimizer $a \in \calA_i$, which we call expert $a$, at round $t$.
Then, let
$y_{i,a}^\t \in \Delta_{m_i}$ be the output of expert $a$.
% , whose utility at each round $t$ is $u_{i,a}^\t$. 
Then, the external regret for expert $a$ is given by
$
  \Reg_{i,a}^T 
  = 
  \max_{y \in \Delta_{m_i}} \sumT \inpr{y - y_{i,a}^\t, u_{i,a}^\t}
$.
Using $\set{y_{i,a}^\t}_{a \in \calA_i}$,
we construct a transition probability matrix $\Qmat_i^\t \in [0,1]^{m_i \times m_i}$ whose $a$-th row is $y_{i,a}^\t$ for each $a \in \calA_i$, \ie~$\Qmat_i^\t(a, \cdot) = \prn{y_{i,a}^\t}^\top$.
Let $\xhat_i^\t$ be a stationary distribution of the Markov chain associated with $\Qmat_i^\t$.
That is $\xhat_i^\t(j) = \sum_{k=1}^{m_i} Q_i^\t(k, j) \xhat_i^\t(k)$, which is equivalent to
$\prn{\Qmat_i^\t}^\top \xhat_i^\t = \xhat_i^\t$,
where $\xhat_i^\t$ is a column vector.
Finally, we use this $\xhat_i^\t$ as a final suggested strategy for player $i$.\footnote{\citet{blum07external} does not consider the corrupted regime and thus $\xhat_i^\t = x_i^\t$, $\util_i^\t = u_i^\t$, and $\util_{i,a}^\t = u_{i,a}^\t$.}

\citet{blum07external} showed that, with the above procedure, the swap regret is equal to the sum of the external regret of external regret minimizer $a \in \calA_i$.
In our corrupted scenario, in which there can be corruption in strategies and observed utilities,
% we observe a vector $\util_i^\t$ instead of $u_i^\t$, 
their result can be restated as follows:

\begin{lemma}\label[lemma]{lem:swap2base_regret}
  Define $\util_{i,a}^\t = \xhat_i^\t(a) \util_i^\t$
  and 
  $
  \widetilde{\Reg}_{i,a}^T(y^*)
  =
  \sumT \inpr{y^* - y_{i,a}^\t, \util_{i,a}^\t}
  $.
  Supose that 
  $\xhat_i^\t = \prn{Q_i^\t}^\top \xhat_i^\t$, where $Q_i^\t(a, \cdot) = y_{i,a}^\t$ for each $a \in \calA_i$
  and 
  recall that
  $
    \SwapReg_{\xhat_i,\tilde{u}_i}^T(M)
    =
    \sumT \inpr{\xhat_i^\t, M \util_i^\t - \util_i^\t}
    .
  $
  Then,
  it holds that
  $
    \SwapReg_{\xhat_i,\tilde{u}_i}^T(M)
    =
    \sum_{a \in \calA_i} \tilde{\Reg}_{i,a}^T (M(a, \cdot))
    \per
  $
\end{lemma}

We include the proof of this lemma for completeness in \Cref{app:proof_multi_player}.

% =========================
\LinesNumbered
\SetAlgoVlined  
\begin{algorithm}[t]
% \begin{algorithm2e}[t]

\For{$t = 1, 2, \dots, T$}{
  % \For{$a \in \calA_i$}{
  For each expert $a \in \calA_i$,
  compute $y_{i,a}^\t \in \Delta_{m_i}$ by OFTRL in \eqref{eq:oftrl_multiplayer}; \\  % with log-barrier regularizer $\phi(y) = - \sum_{k=1}^{m_i} \log(y(k))$ and adaptive learning rate $\eta_{i,a}^\t$ in \Cref{eq:lr_i_corrupt}; \\ 
  % }
  Let $\Qmat_i^\t \in [0,1]^{m_i \times m_i}$ be a matrix whose $a$-th row is $y_{i,a}^\t$, that is $\Qmat_i^\t(a, \cdot) = \prn{y_{i,a}^\t}^\top$; \\
  Let $\xhat_i^\t$ be a stationary distribution of Markov chain induced by $\Qmat_i^\t$, $\prn{\Qmat_i^\t}^\top \xhat_i^\t = \xhat_i^\t$; \\
  Play $\xhat_i^\t = x_i^\t + c_i^\t \in \Delta_{m_i}$ for some corruption $c_i^\t$; \\
  % \delg{Observe expected utility $u_i^\t \in [0,1]^{m_i}$ such that $u_i^\t(a_i) = \E_{a_{-i} \sim x_{-i}^\t} \brk*{u_i(a_i, a_{-i})}$ as in \Cref{eq:def_uit};} \\
  Observe corrupted utility $\tilde{u}_i^\t = u_i^\t + \tilde{c}_i^\t \in [0,1]^{m_i}$ for $u_i^\t(a_i) = \E_{a_{-i} \sim x_{-i}^\t} \brk*{u_i(a_i, a_{-i})}$; \\
  For each $a \in \calA_i$,
  let $\tilde{u}_{i,a}^\t = \xhat_i^\t(a) \tilde{u}_i^\t \in [0,1]^{m_i}$;
  % \For{$a \in \calA_i$}{
    % \gdel{let $u_{i,a}^\t = \xhat_i^\t(a) u_i^\t \in [0,1]^{m_i}$ a d}
  % }
}
\caption{
  No-swap-regret algorithm of player $i$ in multi-player general-sum games
}
\label{alg:multiple_player_swap}
\end{algorithm}
% \end{algorithm2e}
% =========================

\paragraph{No-external-regret algorithm}
From \Cref{lem:swap2base_regret}, it is sufficient to construct a no-external-regret algorithm that aims to minimize the external regret $\tilde{\Reg}_{i,a}^T(M(a,\cdot))$.
For this, we employ OFTRL with the log-barrier regularizer and an adaptive learning rate: we compute $y_{i,a}^\t \in \Delta_{m_i}$ by
\begin{equation}\label{eq:oftrl_multiplayer}
  y_{i,a}^\t 
  \!=\!
  \argmax_{y \in \Delta_{m_i}}
  \set*{
    \!
    \inpr*{y, \util_{i,a}^\tm \!+\! \sum_{s=1}^{t-1} \util_{i,a}^\s \!}
    \!-\! 
    % \frac{1}{\eta_{i,a}^\t} \phi(y)
    \frac{\phi(y)}{\eta_{i,a}^\t}
    \!
  }
  \com \;
  % \quad 
  \eta_{i,a}^\t
  \!=\!
  \min\set*{
    \!\!
    \sqrt{\frac{m_i \log T / 8}{4 \!+\! \sum_{s=1}^{t-1} \nrm{\util_{i,a}^\s \!-\! \util_{i,a}^\sm}_\infty^2 }}
    \com 
    \eta_{i,\max}\!
  }
  % \com
  % \quad 
  % \eta_{i,\max} = \frac{1}{256 n \sqrt{m_i}}
  % \quad 
  % \mbox{for each}
  % \;
  % a \in \calA_i
  % \com
\end{equation}
for each $a \in \calA_i$ and $\eta_{i,\max} = \frac{1}{256 n \sqrt{m_i}}$,
where
$\eta_{i,a}^\t$ is the learning rate for expert $a \in \calA_i$ of player $i$ at round $t$,
$\phi(x) = - \sum_k \log(x(k))$ is the logarithmic barrier function, and recall $\util_{i,a}^\t = \xhat_i^\t(a) \util_i^\t$.
Here, we let $\util_{i,a}^{(0)} = \zeros$ for simplicity.
We use the \emph{expert-wise} adaptive learning rate $\eta_{i,a}^\t$ for each player $i \in [n]$, 
while \cite{anagnostides22uncoupled} uses a constant common learning rate, $\eta_{i,a}^\t = \eta_i$.
% :
% \delg{
% \begin{equation}\label{eq:lr_i_corrupt}
%   \eta_i^\t
%   =
%   \min\set*{
%     \frac{m_i}{2}
%     \sqrt{\frac{\log T}{4 + \sum_{s=1}^{t-1} \nrm{u_i^\s - u_i^\sm}_\infty^2 }}
%     \com 
%     \frac{1}{256 n \sqrt{m_i}}
%   }
%   \per 
% \end{equation}
% }
\begin{comment}
\begin{equation}\label{eq:lr_i_corrupt}
  \eta_{i,a}^\t
  =
  \min\set*{
    \sqrt{\frac{m_i \log T / 8}{4 + \sum_{s=1}^{t-1} \nrm{\util_{i,a}^\s - \util_{i,a}^\sm}_\infty^2 }}
    \com 
    \eta_{i,\max}
  }
  \com
  \quad 
  \eta_{i,\max} = \frac{1}{256 n \sqrt{m_i}}
  \per 
  % \tag{\red{new}}
\end{equation}
\end{comment}
% \blue{(Note: using $\eta_{i,\max} = \frac{1}{ 256 } \sqrt{\frac{\mmax}{n \sum_{i \in [n]} m_i^2}}$ would yield slightly better bound)}
% Note that when $m_i = \mmax$ for all $i \in [n]$, $\eta_{i,\max} = 1/(256 n \sqrt{m})$.
% Note that the learning rate is the same for each expert $a \in \calA_i$, which outputs $y_{i,a}^\t \in \Delta(\calA_i)$.
The algorithm for each player $i \in [n]$ is summarized in \Cref{alg:multiple_player_swap}.

Let $\mmax = \max_{i \in [n]} m_i$ be the maximum number of actions and $\hat{S} = \sum_{i \in [n]} \hat{C}_i$, $\tilde{S} = \sum_{i \in [n]} \tilde{C}_i$, and $S = \sum_{i \in [n]} C_i$.
Then, \Cref{alg:multiple_player_swap} achieves the following individual swap regret upper bounds.
\begin{theorem}[Swap regret upper bounds]\label{thm:indiv_swapreg}
  % Consider multi-player general-sum games.
  In the corrupted regime, \Cref{alg:multiple_player_swap} achieves  
  \begin{equation}
    \begin{split}
    % \SwapReg_i^T 
    % =
    \textstyle
    \SwapReg_{x_i,u_i}^T 
    % &
    &\!\lesssim\!
    % \begin{cases}
    %     n \mmax^{5/2} \log T 
    %   & \mbox{honest regime}
    %   \\
        n \mmax^{5/2} \log T 
        \!+\!
        \min\set[\Big]{
          m
          \sqrt{
            \prn[\big]{
              \hat{S} \prn*{n + \sqrt{\mmax} }
              +
              \tilde{S} \sqrt{\mmax}
            }
            \log T
          }
          ,
          m_i \sqrt{T \log T}
        }
        +
        C_i 
      % & \mbox{corrupted regime} 
      \per
      % \label{eq:thm_swapreg_xu}
    % \end{cases}
    % \n
  % \end{equation}
  % \begin{equation}
  % \nn
  \\
    \textstyle
    \SwapReg_{x_i,\util_i}^T 
    % &
    &\!\lesssim\!
    % \begin{cases}
        % n \mmax^{5/2} \log T 
      % & \mbox{honest regime}
      % \\
        n \mmax^{5/2} \log T 
        \!+\!
        \min\set[\Big]{
          \mmax
          \sqrt{
            \!
            \prn[\big]{
              \hat{S}
              \prn*{
                n
                \!+\!
                \sqrt{\mmax}
              }
              \!+\!
              \tilde{C}_i
            } \!
            \log T
          }
          \!+\!
          \prn{\tilde{S} n m^6 \log T}^{\frac{1}{4}}
          ,
          m_i \sqrt{T \log T}
        }
        \!+\!
        \Chati 
      % & \mbox{corrupted regime} 
      \per
    % \end{cases}
    \n
    \end{split}
  \end{equation}
\end{theorem}

The proof can be found in \Cref{app:proof_multi_player},
which provides slightly better upper bounds.
A comparison against existing bounds can be found in \Cref{table:regret_generalsum}.
Note that setting $C_i = 0$ for all $i$ in the above bounds yields $\SwapReg_{x_i, u_i}^T \lesssim n \mmax^{5/2} \log T$ and $\SwapReg_{x_i,\util_i}^T \lesssim n \mmax^{5/2} \log T$ in the honest regime.
Compared to the best swap regret bounds by~\citet{anagnostides22uncoupled}, our algorithm achieves the same bound in the honest regime, a new adaptive bound in the corrupted regime in terms of $\hat{S}$ and $\tilde{S}$, and a worst-case bound that is $\sqrt{m}$-times worse than their bound of $O(n m^{5/2} \log T + \sqrt{T m \log m})$.
It is worth noting that 
the bound on $\SwapReg_{x_i,u_i}^T$ deteriorates linearly with respect to $\set{\tilde{C}_j}_{j \neq i}$, whereas the bound on $\SwapReg_{x_i,\util_i}^T$ is affected by $\tilde{C}_i$ only through a square root dependence and by $\set{\tilde{C}_j}_{j \neq i}$ through a \emph{fourth root} dependence.
Interestingly, 
the fourth root dependence on $\tilde{C}_j$ (for $j \neq i$) is better than the squared root on $\tilde{C}_\yrm$ in the upper bound of $\Reg_{x,\gtil}^T$.
This difference arises because, whereas $\Reg_{x,\gtil}^T + \Reg_{y,\ltil}^T \not\geq 0$  (see \Cref{app:remaining_proof_reg}), we always have $\SwapReg_{x_i,\util_i}^T \geq 0$. %   and investigating the optimality of these is an interesting future work.

\begin{remark}
Using expert-wise learning rates $\eta_{i,a}^\t$ as in \Cref{eq:oftrl_multiplayer}, is crucial for proving an upper bound of
$\SwapReg_{x_i, \util_i}^T \lesssim \nsqrt{\tilde{C}_i} + \prn{\sum_{j \neq i} \tilde{C}_j}^{1/4}$, which has a sublinear dependence on $\set{\tilde{C}_i}_{i \in [n]}$.
If one modifies the algorithm of \citet{anagnostides22uncoupled} by using an adaptive learning rate, one might choose the same learning rate for all $a \in \calA_i$ as 
$
  \eta_i^\t
  \!\!\simeq\!
  \min\set[\Big]{
    m_i
    \sqrt{\log T/\prn{4 \!+\! \sum_{s=1}^{t-1} \nrm{u_i^\s \!\!-\! u_i^\sm}_\infty^2 }}
    \com 
    {1}/\prn{n \sqrt{m_i}}
  }
$.
However, despite our attempts, using this learning rate causes $\SwapReg_{\xhat_i,\util_i}^T$ to depend linearly on $\tilde{C}_i$, which fails to ensure the desired robustness against corruption in utilities.
\end{remark}

The key analysis to prove \Cref{thm:indiv_swapreg}  
lies in analyzing the stability of the Markov chain  
determined by OFTRL with the adaptive learning rate in \cref{eq:oftrl_multiplayer},
which defines the output $y_{i,a}^\t$ of each $a \in \calA_i$.  
In contrast to the analysis by \citet{anagnostides22uncoupled}, which uses a constant learning rate, we use the adaptive learning rate.  
By setting the learning rate sufficiently small and applying an analysis similar to that for two-player zero-sum games in \Cref{sec:two_player}, as well as the analysis in \citet{wei18more},  
we can show a swap regret bound of $O(n m^8 \log T)$ in the honest regime.
However, the dependence on $m$ is significantly worse than the $O(n m^{5/2} \log T)$ bound in \citet{anagnostides22uncoupled}.  
To address this issue, we leverage the property of reducing swap regret minimization to external regret minimization (\Cref{lem:swap2base_regret}) in the stability analysis:
exploiting this property ensures the stability  
of transition probability matrices without making the learning rate too small (\Cref{lem:diff_local_oftrl}),
thereby allowing us to establish a swap regret bound of $O(n m^{5/2} \log T)$ in the honest regime while maintaining robustness in the corrupted regime.

\section{Lower Bounds for Corrupted Games}\label{sec:lower_bound}
This section provides several matching lower bounds of \Cref{thm:indiv_reg_corrupt,thm:indiv_swapreg}.
All the lower bounds are constructed for two-player zero-sum games and proofs can be found in \Cref{app:proof_lower_bound}.
% and for simplicity, we refer to a game in a corrupted regime as a corrupted game. 

\begin{theorem}[Lower bounds in the corrupted regime]\label{thm:lower_bounds}
  For any learning dynamics, 
  \newcounter{myitem}
  \begin{itemize}[topsep=3pt, itemsep=-2pt, partopsep=0pt, leftmargin=20pt] % [topsep=2pt, itemsep=0pt, partopsep=0pt, leftmargin=25pt]
    \item[(i)]\label{itm:lower_1}
    there exist corrupted games with $\sumT \nrm{\mspace{-0.5mu} g^\t \mspace{-2mu} - \mspace{-1.5mu} \gtil^\t \mspace{-1mu}}_\infty \!\leq\! \Ctilx$ 
    and $\sumT \nrm{\mspace{-0.5mu} \ell^\t \mspace{-2mu} - \mspace{-1.5mu} \ltil^\t \mspace{-1mu}}_\infty \!\leq\! \Ctily$
    such that 
    $
    \Reg_{x,\gtil}^T = \Reg_{\xhat,\gtil}^T 
    =
    \Omega\prn[\big]{\nsqrt{\tilde{C}_{\xrm} \log \mx}}
    $
    and
    $
    \Reg_{y,\ltil}^T 
    =
    \Reg_{\yhat,\ltil}^T 
    =
    \Omega\prn[\big]{\nsqrt{\tilde{C}_{\yrm} \log \my}},
    $
    respectively;
    % \end{equation}
    \item[(ii)]\label{itm:reg_Chatxy_linear}
    there exist corrupted games with
    $\sumT \nrm{x^\t - \xhat^\t}_1 \leq \Chatx$
    and 
    $\sumT \nrm{y^\t - \yhat^\t}_1 \leq \Chaty$
    such that
    $
    \Reg_{x,g}^T = \Reg_{x,\gtil}^T = \Omega\prn{\Chatx}
    $
    and
    $
    \Reg_{y,\ell}^T = \Reg_{y,\ltil}^T = \Omega\prn{\Chaty},
    $
    respectively;
    \item[(iii)]\label{itm:reghat_Chatxy_lower}
    there exists a corrupted game with
    $\sumT \nrm{y^\t - \yhat^\t}_1 \leq \Chaty$
    and 
    $\sumT \nrm{x^\t - \xhat^\t}_1 \leq \Chatx$
    such that
    $
    \max \set[\big]{ \Reg_{\xhat,g}^T, \Reg_{\yhat,\ell}^T }
    =
    \Omega\prn[\big]{\nsqrt{\Chaty}}
    $
    and
    $
    \max \set[\big]{ \Reg_{\xhat,g}^T, \Reg_{\yhat,\ell}^T }
    =
    \Omega\prn[\big]{\nsqrt{\Chatx}}
    ,
    $
    respectively.
  \end{itemize}
\end{theorem}
Note that the lower bounds in (iii), unlike those in (ii), are for regrets $\Reg_{\xhat,g}^T$ and $\Reg_{\yhat,\ell}^T$,  
which are defined for the pre-corruption strategies  
$\set{\xhat^\t}_{t=1}^T$ and $\set{\yhat^\t}_{t=1}^T$.
The lower bound in (iii) is stated in a form similar to the lower bounds of \citet[Theorem 24]{syrgkanis15fast} and \citet[Theorem 4.2]{chen20hedging}. 
However, while their lower bounds focus only on a specific algorithm (the Hedge algorithm), our lower bounds hold for any learning dynamics. 
Moreover, their lower bounds do not address the corrupted regime.
The detailed results of their lower bounds can be found in \Cref{sec:related_work}.

The lower bound in (i) is proven by considering the payoff matrix $A = 0$ and reducing the problem to the finite-time lower bound for online linear optimization over the probability simplex (\Cref{lem:olo_simplex_lower_logm}).  
The bound in (ii) (specifically, the first statement) is shown for the payoff matrix $A$ whose first $\mx - 1$ rows consist entirely of $1$s and whose last row consists entirely of $0$s.  
By designing corruption that forces $x$-player to select the $\mx$-th action, the desired bound follows.
For (iii), we consider 
$
A = 
\begin{pmatrix}
  1 & 0 & -1 \\
  0 & 1 & -1 \\
\end{pmatrix}.
$
Observe that the first statement of (iii) is equivalent to the existence  
of a constant $\kappa > 0$ and a corrupted game satisfying  
$\sumT \nrm{y^\t - \yhat^\t}_1 \leq \Chaty$  
such that  
$
\Reg_{\xhat,g}^T < \kappa \nsqrt{\Chaty}$ implies $\Reg_{\yhat,\ell}^T \geq \kappa \nsqrt{\Chaty}.
$
Then, we use the fact that under the matrix $A$, for $y$-player to satisfy $\Reg_{\xhat,g}^T < \kappa \nsqrt{\Chaty}$,  
$y$-player can select actions~1 and~2 at most $\kappa \nsqrt{\Chaty}$ times.

\bibliography{references.bib}
\bibliographystyle{plainnat}

\newpage
\appendix

\section{Additional Related Work}\label{sec:related_work}

This section discusses additional related work that could not be included above.

\paragraph{External regret minimization for computing Nash equilibrium}
In two-player zero-sum games, it is known that when each player uses an (external) regret minimization algorithm, the average iterate of the chosen strategies converges to a Nash equilibrium (see \Cref{thm:reg2nasheq}).
For example, when each player employs FTRL with the negative Shannon entropy corresponding to Hedge~\citep{littlestone94weighted,freund97decision}, which has an upper regret bound of $\tilde{O}(\sqrt{T})$, an approximate Nash equilibrium is obtained at a rate of $\tilde{O}(1/\sqrt{T})$. 
A significant improvement in this convergence rate can be achieved by employing the optimistic frameworks such as OFTRL and the optimistic online mirror descent~\citep{rakhlin13online,rakhlin13optimization,syrgkanis15fast}.
The best-known upper bound for the individual regret in the honest regime is $O(\log(\mx \my))$, achieved with OFTRL with the negative Shannon entropy regularizer and a constant learning rate~\citep{syrgkanis15fast}.
Importantly, the authors ensure robustness against adversarial opponents by monitoring the cumulative variation of the actions chosen by the player and switching to an algorithm with a learning rate of $\Theta(1/\sqrt{T})$ once it exceeds a certain threshold.
Since the work of~\cite{syrgkanis15fast}, the optimistic prediction has been widely used as a tool for achieving an $o(\sqrt{T})$ regret in learning in games \citep{foster16learning,wei18more,chen20hedging,anagnostides22uncoupled}.

Although not the focus of this paper, it is known that in multi-player general-sum games, if each player uses a no-external-regret algorithm, their time-averaged distribution of joint play converges to a coarse correlated equilibrium (CCE), which is a less favorable notion of equilibrium compared to the correlated equilibrium. 
\citet{syrgkanis15fast} showed that external regret minimization with OFTRL using the negative Shannon entropy achieves an individual regret of $\tilde{O}(T^{1/4})$ in the honest regime. 
This result was later improved to $\tilde{O}(T^{1/6})$ for two-player games by~\citet{chen20hedging}.
Subsequently, significant advances in algorithms and analytical techniques have shown that (poly)logarithmic external regret can be achieved~\citep{daskalakis21near,farina22near}. 
Although this paper does not address it, we believe that a similar extension of the work by~\citet{farina22near} to the corrupted regime is possible.

\paragraph{Swap regret minimization for computing correlated equilibrium}
In multi-player general-sum games, the correlated equilibrium, which is considered a more desirable solution concept compared to the coarse correlated equilibrium, can be obtained when each player employs a no-swap-regret algorithm (see \Cref{thm:swapreg2coreq}).
As discussed in \Cref{sec:multi_player}, to minimize the swap regret, we prepare a set of base experts for each action and determine the player's strategies using a Markov chain defined by these outputs~\citep{blum07external}.
The challenge in analyzing swap regret minimization lies in the need for precise analysis of the stability of the stationary distribution of the Markov chain. 
Notably, \citet{chen20hedging} analyzed this stability using the Markov chain tree theorem and demonstrated that individual swap regret can be bounded by $O(T^{1/4})$ by employing OFTRL with the negative Shannon entropy.
Later, this swap regret bound was significantly improved by~\citet{anagnostides22near,anagnostides22uncoupled}. 
In particular, \citet{anagnostides22uncoupled} demonstrated that using OFTRL with the log-barrier regularizer and a constant learning rate makes the stability of the Markov chain stronger, leading to the first $O(\log T)$ individual swap upper bound.
To ensure the adversarial robustness of the algorithm, they also considered a procedure that switches to an algorithm with a learning rate of $\Theta(1/\sqrt{T})$ whenever the cumulative variation of the observed utilities exceeds $\log T$, which is similar to~\citet{syrgkanis15fast}.
Although not addressed in this study, there are lines of research that consider variations of the swap regret in Bayesian games~\citep{fujii23bayes} and that focus on dependencies on the number of actions $m$ instead of the number of rounds $T$ when analyzing the swap regret~\citep{dagan24from,peng24fast}.

\paragraph{Adaptive learning rate}
The adaptive learning rate is a critical technical element in our corrupted learning dynamics. 
The use of adaptive learning rate in learning in games is not new.
Adaptive learning rate has been used in two-player zero-sum matrix games~\citep{rakhlin13optimization}, cocoercive games~\citep{lin20finite}, variational inequalities~\citep{antonakopoulos19adaptive,antonakopoulos21adaptive}, and variationally stable games~\citep{hsieh21adaptive}. 
Of these, \citet{rakhlin13optimization,hsieh21adaptive} are the most relevant to our study. 
Some of our results partially subsume theirs as a special case.
However, unlike their work, 
the use of adaptive learning rate for the corrupted regime has not been studied before, especially in terms of the swap regret.

\paragraph{Time-varying games}
\citet{zhang22noregret} introduced the concept of time-varying games in two-player zero-sum settings,  
demonstrating that regret upper bounds can be achieved depending on the variation of the payoff matrix.  
Following their framework, time-varying games have gained increasing interest~\citep{harris23meta,yan23fast,feng23last,anagnostides23convergence}.  
Our setting, where corrupted expected utilities are observed,   can be seen as a generalization of the time-varying setting,  
where the underlying payoff matrix $A$ or utility function $u_i$ are time-varying.  
(When considering such variations, recall from \Cref{rem:setting_AB} that defining regret in terms of $\Reg_{x_i,\util_i}^T$ or $\SwapReg_{x_i,\util_i}^T$ is the natural choice.)  
However, a key distinction is that our primary motivation lies  
in corruption in strategies.  
Moreover, we consider a more general corruption model for utility than \citet{zhang22noregret} and its subsequent works, making a direct comparison with their regret upper bounds impossible.

\paragraph{Lower bounds}
While there has been extensive research on regret upper bounds,  
there has been relatively little work on regret lower bounds.  
Initially, \citet[Theorem 24]{syrgkanis15fast} considered a setting where $x$-player uses (vanilla) Hedge (recall that this corresponds to FTRL with the negative Shannon entropy) with any learning rate, while $y$-player follows a (pure) best response  (i.e., minimizing the expected utility of the current round based on the $x$-player's choice). 
They showed that for Hedge with any learning rate, there exists a two-player zero-sum game where the $x$-player must suffer a $\sqrt{T}$ regret.  
Later, \citet[Theorem 4.2]{chen20hedging} extended this result to a setting where both $x$- and $y$-players use Hedge.
They demonstrated that for Hedge with any learning rate, there exists a two-player general-sum game where at least one player suffers a $\sqrt{T}$ regret.  
Note that these regret lower bounds do not apply to the corrupted regime, and these results rely on regret upper bounds that depend on the update rule of Hedge.  
In contrast, our lower bounds for corrupted games in \Cref{thm:lower_bounds} hold for any learning dynamics.

\section{Regret Analysis of Optimistic Follow-the-Regularized-Leader}\label{app:proof_oftrl}

In this section, we analyze the regret of the Optimistic Follow-the-Regularized-Leader (OFTRL) for online linear optimization.
First, we provide a general analysis of OFTRL. 
Then, we present RVU bounds for OFTRL with the negative Shannon entropy and the log-barrier regularizer, 
which are used respectively in two-player zero-sum games in \Cref{sec:two_player} and multi-player general-sum games in \Cref{sec:multi_player}.
In the context of learning in games, it is common to consider a \emph{utility} vector $u^\t$ instead of a loss vector $\ell^\t$.
Even in such cases, our statements for online linear optimization provided in this paper can be applied by letting $\ell^\t = - u^\t$. 

\subsection{Common analysis}

The following lemma provides a regret  bound for OFTRL that holds for a general regularizer $\psi^\t$.
\begin{lemma}[Regret bound for OFTRL]\label[lemma]{lem:oftrl_bound}
  Let $\calK \subseteq \R^d$ be a nonempty closed convex set.
  Let
  $x^\t \in \argmin_{x \in \calK} \set{\inpr{x, m^\t + \sum_{s=1}^{t-1} \ell^\s} + \psi^\t(x)}$ be the output of OFTRL (with linearized losses) at round $t$.
  Then, for any $x^* \in \calK$,
  \begin{align}
    &\sumT \inpr{x^\t - x^*, \ell^\t}
    \leq
    \psi^{(T+1)}(x^*) - \psi^{(1)}(x^{(1)})
    +
    \sumT \prn*{ \psi^\t(x^\tp) - \psi^\tp(x^\tp)}
    \nn 
    &\qquad+
    \sumT 
    \prn*{
      \inpr{x^\t - x^\tp, \ell^\t - m^\t}
      -
      D_{\psi^\t}(x^\tp, x^\t)
    }
    +
    \inpr{x^* - x^{(T+1)}, m^{(T+1)}}
    \per 
    \label{eq:oftrl_bound}
  \end{align}
\end{lemma}
This lemma follows from a standard analysis of OFTRL and can be seen as a variant of \citet[Theorem 7.36]{orabona2019modern}.
We include the proof for completeness.
\begin{proof}
Define $\tilde{\psi}^\t(x) = \psi^\t(x) + \inpr{x, m^\t}$, and let
\begin{equation}
  F^\t(x) 
  = 
  \inpr*{x, m^\t + \sum_{s=1}^{t-1} \ell^\s} + \psi^\t(x) 
  = 
  \inpr*{x, \sum_{s=1}^{t-1} \ell^\s} + \tilde{\psi}^\t(x) 
  \n
\end{equation}
be the objective function of OFTRL at round $t$.
From $- \sumT \inpr{x^*, \ell^\t} = \tilde{\psi}^{(T+1)}(x^*) - F^{(T+1)}(x^*)$, we have
\begin{align}
  &
  \sumT \inpr{x^\t - x^*, \ell^\t}
  \nn
  &=
  \tilde{\psi}^{(T+1)}(x^*) \!-\! F^{(T+1)}(x^*)
  \!+\!
  \prn[\big]{- F^{(1)}(x^{(1)}) + F^{(1)}(x^{(1)})}
  \!+\!
  \prn[\big]{- F^{(T+1)}(x^{(T+1)}) + F^{(T+1)}(x^{(T+1)})}
  \nn
  \!&\qquad+\!
  \sumT \inpr{x^\t, \ell^\t}
  \nn
  &=
  \tilde{\psi}^{(T+1)}(x^*) - F^{(T+1)}(x^*)
  - F^{(1)}(x^{(1)}) 
  +
  \sumT \prn[\big]{ F^\t(x^\t) - F^\tp(x^\tp)}
  +
  F^{(T+1)}(x^{(T+1)})
  \nn
  &\qquad+
  \sumT \inpr{x^\t, \ell^\t}
  \nn
  &\leq
  \tilde{\psi}^{(T+1)}(x^*) 
  - 
  \tilde{\psi}^{(1)}(x^{(1)})
  +
  \sumT \prn[\big]{ F^\t(x^\t) - F^\tp(x^\tp) + \inpr{x^\t, \ell^\t}}
  \com 
  \label{eq:oftrl-decompose}
\end{align}
where in the second inequality we considered the telescoping sum and in the last line we used the fact that $x^{(T+1)}$ is the minimizer of $F^{(T+1)}$.
The last term in the last inequality is bounded as
\begin{align}
  &
  F^\t(x^\t) - F^\tp(x^\tp) + \inpr{x^\t, \ell^\t}
  \nn
  &=
  F^\t(x^\t)
  - 
  F^\t(x^\tp)
  % \nn
  % &\qquad
  +
  \inpr{x^\tp, m^\t - m^\tp}
  \nn
  &\qquad+
  \psi^\t(x^\tp) - \psi^\tp(x^\tp)
  + 
  \inpr{x^\t - x^\tp, \ell^\t}
  \nn
  &\leq
  - D_{F^\t} (x^\tp, x^\t) 
  +
  \psi^\t(x^\tp) - \psi^\tp(x^\tp)
  \nn
  &\qquad
  + 
  \inpr{x^\tp, m^\t - m^\tp}
  +
  \inpr{x^\t - x^\tp, \ell^\t}
  \per
  \label{eq:oftrl_stab_general}
\end{align}
Here in the last inequality we used 
\begin{equation}
  D_{F^\t}(x^\tp, x^\t)
  % &
  =
  F^\t(x^\tp)
  - 
  F^\t(x^\t)
  -
  \inpr{x^\tp - x^\t, \nabla F^\t(x^\t)}
  \leq 
  F^\t(x^\tp)
  -
  F^\t(x^\t)
  \com
  \n
\end{equation}
where the inequality follows from the first-order optimality condition at $x^\t$, that is,
$
\inpr{x - x^\t, \nabla F^\t(x^\t)} \geq 0
$ for any $x \in \calK$.
Combining \Cref{eq:oftrl-decompose}, \Cref{eq:oftrl_stab_general}, and the fact that $D_{F^\t}(x^\tp, x^\t) = D_{\psi^\t}(x^\tp, x^\t)$, we obtain
\begin{align}
  &
  \sumT \inpr{x^\t - x^*, \ell^\t}
  \nn
  &\leq 
  \tilde{\psi}^{(T+1)}(x^*) 
  - 
  \tilde{\psi}^{(1)}(x^{(1)})
  +
  \sumT \prn*{ 
    \psi^\t(x^\tp) - \psi^\tp(x^\tp)
    - D_{\psi^\t} (x^\tp, x^\t) 
  }
  \nn 
  &\qquad+
  \sumT 
  \inpr{x^\tp, m^\t - m^\tp}
  +
  \sumT
  \inpr{x^\t - x^\tp, \ell^\t}  
  \nn
  &= 
  {\psi}^{(T+1)}(x^*) 
  - 
  {\psi}^{(1)}(x^{(1)})
  +
  \sumT \prn*{ 
    \psi^\t(x^\tp) - \psi^\tp(x^\tp)
    - D_{\psi^\t} (x^\tp, x^\t) 
  }
  \nn 
  &\qquad+
  \sumT 
  \inpr{x^\tp - x^\t, m^\t} + \inpr{x^* - x^{(T+1)}, m^{(T+1)}}
  +
  \sumT
  \inpr{x^\t - x^\tp, \ell^\t}  
  \com 
  \n
\end{align}
where the last line follows from the telescoping 
$\sumT \inpr{x^\tp, m^\t - m^\tp} = \sumT \inpr{x^\tp - x^\t, m^\t} + \inpr{x^{(1)}, m^{(1)}} - \inpr{x^{(T+1)}, m^{(T+1)}}$.
This completes the proof.
\end{proof}

\subsection{OFTRL with negative Shannon entropy regularizer}
\Cref{lem:oftrl_bound} immediately yields the well-known RVU bound for OFTRL with the negative Shannon entropy, which we use in our analysis for two-player zero-sum games in \Cref{sec:two_player}.
\begin{restatable}[RVU bound for OFTRL with negative Shannon entropy regularizer]{lemma}{lemoftrlshannon}\label[lemma]{lem:oftrl_shannon}
  Let
  $\psi^\t(x) = - \frac{1}{\eta^\t} H(x)$ for
  $H(x) = \sum_{k=1}^d x(k) \log(1/x(k))$ be the negative Shannon entropy regularizer with nonincreasing learning rate $\eta^\t$
  and 
  $x^\t \in \argmin_{x \in \Delta_d} \set{\inpr{x, m^\t + \sum_{s=1}^{t-1} \ell^\s} + \psi^\t(x)}$ be the output of OFTRL at round $t$.
  Then, for any $x^* \in \Delta_d$,
  \begin{equation}
  % \begin{align}
      \sumT \inpr{x^\t - x^*, \ell^\t}
    % &
    % \leq 
    % \frac{\log d}{\eta^{(T+1)}}
    % +
    % \sumT 
    % \inpr{x^\t - x^\tp, \ell^\t - m^\t}
    % - 
    % \sumT 
    % \frac{1}{2 \eta^\t} \nrm{x^\t - x^\tp}_1^2
    % + 2 \nrm{m^{(T+1)}}_\infty 
    % \label{eq:lem_oftrl_1}
    % \\
    % &
    \leq
    \frac{\log d}{\eta^{(T+1)}}
    +
    \sumT
    \eta^\t \nrm{\ell^\t - m^\t}_\infty^2 
    - 
    \sumT 
    \frac{1}{4 \eta^\t} \nrm{x^\t - x^\tp}_1^2  
    + 2 \nrm{m^{(T+1)}}_\infty 
    \per 
    \n
    \end{equation}
\end{restatable}

% \lemoftrlshannon*
% \lemoftrlshannon*
\begin{proof}
We will upper bound the RHS of \Cref{eq:oftrl_bound} in \Cref{lem:oftrl_bound}.
% Since $\eta^\t$ is nonincreasing, 
Since $H(x) \leq \log d$ for all $x \in \Delta_d$,
we have
\begin{align}
  &
  \psi^{(T+1)}(x^*) 
  - \psi^{(1)}(x^{(1)})
  +
  \sumT \prn*{ \psi^\t(x^\tp) - \psi^{(t+1)}(x^\tp)}
  \nn
  &\leq 
  \frac{\log d}{\eta^{(1)}}
  +
  \sumT \prn*{\frac{1}{\eta^\tp} - \frac{1}{\eta^\t}} \log d
  =
  \frac{\log d}{\eta^{(T+1)}}
  \per
  \n
\end{align}
Since $\psi^\t$ is $(1/\eta^\t)$-strongly convex w.r.t.~$\nrm{\cdot}_1$, we also have 
\begin{align}
  &
  \inpr{x^\t - x^\tp, \ell^\t - m^\t}
  -
  D_{\psi^\t}(x^\tp, x^\t)
  \nn
  &\leq 
  \nrm{x^\t - x^\tp}_1 \nrm{\ell^\t - m^\t}_\infty
  -
  \frac{1}{2 \eta^\t} \nrm{x^\t - x^\tp}_1^2
  \nn 
  &=
  \nrm{x^\t - x^\tp}_1 \nrm{\ell^\t - m^\t}_\infty
  -
  \frac{1}{4\eta^\t} \nrm{x^\t - x^\tp}_1^2
  -
  \frac{1}{4\eta^\t} \nrm{x^\t - x^\tp}_1^2
  \nn 
  &\leq 
 \eta^\t \nrm{\ell^\t - m^\t}_\infty^2 
  - \frac{1}{4 \eta^\t} \nrm{x^\t - x^\tp}_1^2
  \com 
  \n
\end{align}
where the first inequality follows from H\"older's inequality and the $(1/\eta^\t)$-strong convexity of $\psi^\t$ with respect to $\nrm{\cdot}_1$, and the last inequality follows from $b \sqrt{z} - a z \leq b^2 / (4a)$ for $a > 0, b \geq 0$, and $z \geq 0$.
Combining \Cref{lem:oftrl_bound} with the last two inequalities gives the desired bound.
\end{proof}

\subsection{OFTRL with self-concordant barrier and adaptive learning rate}

Here, we present the RVU bound for OFTRL with a self-concordant barrier, a class of functions that includes the log-barrier regularizer. 
We first introduce the concepts and key properties of self-concordant functions and self-concordant barriers.
\paragraph{Self-concordant function and self-concordant barrier}
We first define the self-concordant function and self-concordant barrier (see \eg~\citealt{nesterov94interior,nemirovski08interior} for detailed background).
Recall that
$\nrm{h}_{x,f} = \sqrt{h^\top \nabla^2 f(x) h}$
and 
$\nrm{h}_{*,x,f} = \sqrt{h^\top \prn{\nabla^2 f(x)}^{-1} h}$.
\begin{definition}\label{def:sc_func}
  Let $\calK \subseteq \R^d$ be a closed convex set.
  Then a function $f \colon \interior\prn{\calK} \to \R$ is \emph{self-concordant} if 
  $f$ is three times continuously differentiable convex function with $f(x_k) \to \infty$ if $x_k \to x_\infty \in \partial \calK$,
  and
  \begin{equation}
    \abs{ D^3 f(x)[h, h, h] }
    \leq 
    2 \prn*{ D^2 f(x)[h, h] }^{3/2}
    \n
  \end{equation}
  for all $x \in \interior\prn{\calK}$ and $h \in \R^d$.
\end{definition}

\begin{definition}\label{def:sc_barrier}
  Let $\calK \subseteq \R^d$ be a closed convex set and $\vartheta \geq 1$.
  Then a function $g$ is \emph{$\vartheta$-self-concordant barrier} for $\calK$ if
  $g$ is self-concordant and
  \begin{equation}
    \abs{ D f(x) [h] }
    \leq 
    \prn{\vartheta D^2 f(x)[h, h]}^{1/2}
    \n
  \end{equation}
  for all $x \in \interior\prn{\calK}$ and $h \in \R^d$.
\end{definition}
We use the fact that the log-barrier function $\phi(x) = - \sum_{k=1}^d \log(x(k))$ is a $d$-self-concordant barrier over the positive orthant.
For self-concordant functions, we use the following lemma.
\begin{lemma}[{\citealt[Theorem 2.1.1]{nesterov94interior}}]\label[lemma]{lem:sc_hessian_stab}
  Let $\calS$ be an open convex subset of a finite-dimensional real vector space.
  Let $f$ be a self-concordant function on $\calS$.
  Then, for any $y \in \calS$ such that $\nrm{x - y}_{y,f} < 1$, 
  \begin{equation}
    \prn{1 - \nrm{x - y}_{y, f}}^2 
    \nabla^2 f(y) 
    \preceq
    \nabla^2 f(x) 
    \preceq 
    \frac{1}{\prn{1 - \nrm{x - y}_{y,f}}^2}
    \nabla^2 f(y) 
    \per 
    \n
  \end{equation}
\end{lemma}
We will use the following properties of the self-concordant barrier.
\begin{lemma}[{\citealt[Proposition 2.3.2]{nesterov94interior}}]\label[lemma]{lem:sc_barrier_diff}
Let $g$ be a $\vartheta$-self-concordant barrier for $\calK$.
Then, for any $x, y \in \interior(\calK)$,
\begin{equation}
  g(y) - g(x) \leq \vartheta \log \prn*{ \frac{1}{1 - \pi(y;x)}}
  \com 
  \n
\end{equation}
where $\pi(y, x) = \inf\set{s \geq 0 \colon x + s^{-1} (y - x) \in \calK}$ is the Minkowski function on $\calK$ with pole at $x$.
\end{lemma}

The following lemma will be used in the proof of \Cref{thm:indiv_swapreg}.
\begin{lemma}\label[lemma]{lem:scbarrier_properties}
  Let $f$ be a $\vartheta$-self-concordant barrier on $\calK$.
  Then, for any $x \in \interior(\calK)$ such that $\nabla^2 f(x)$ is invertible,
  it holds that $\nrm{\nabla f(x)}_{*,x,f}^2 \leq \vartheta$.
\end{lemma}
\begin{proof}
  This follows from the definition of the self-concordant barrier in \Cref{def:sc_barrier} with $h = (\nabla^2 f(x))^{-1} f(x)$.
\end{proof}

\paragraph{OFTRL with self-concordant barrier}
Using these properties, we can establish an upper bound on the regret of OFTRL when employing a self-concordant barrier regularizer and an adaptive learning rate.
To our knowledge, few thorough analyses derive the negative term when using an adaptive learning rate in OFTRL with self-concordant barriers.
\begin{lemma}[RVU bound for OFTRL with self-concordant barrier and adaptive learning rate]\label[lemma]{lem:oftrl_selfconcordant}
  Let $\calK \subseteq \R^d$ be a nonempty closed convex set with a diameter $D = \max_{x, y \in \calK} \nrm{x - y}$.
  Let $\phi$ be a $\vartheta$-self-concordant barrier for $\calK$
  and 
  $\psi^\t(x) = \frac{1}{\eta^\t} \phi(x)$ be a regularizer with nonincreasing and nonnegative learning rate $\set{\eta^\t}_{t=1}^T$.
  For this $\psi^\t$, consider the OFTRL update
  $x^\t \in \argmin_{x \in \calK} \set{\inpr{x, m^\t + \sum_{s=1}^{t-1} \ell^\s} + \psi^\t(x)}$.
  Suppose that the sequence of iterates $\set{x^\t}_{t=1}^T$ satisfies
  % $\eta^\t \nrm{\ell^\t - m^\t}_{*,x^\t,\phi} \leq 1/3$.
  $\nrm{x^\tp - x^\t}_{x^\t,\phi} \leq 1/2$ for all $t \in [T]$.
  Then, for any $x^* \in \calK$,
  \begin{equation}
    \sumT \inpr{x^\t - x^*, \ell^\t}
    \leq
    \frac{\vartheta \log T}{\eta^{(T+1)}}
    +
    \sumT 
    4 \eta^\t \nrm{\ell^\t - m^\t}_{*,x^\t,\phi}^2 
    - 
    \sumT 
    \frac{1}{16 \eta^\t} \nrm{x^\tp - x^\t}_{x^\t, \phi}^2
    + 
    3 D L
    \com 
    \n
  \end{equation}
  where 
  $L = \max \set{ \max_{t \in [T]} \nrm{\ell^\t}_*,  \max_{t \in [T+1]} \nrm{m^\t}_*  }$.
\end{lemma}

\begin{proof}
We will upper bound the RHS of \Cref{eq:oftrl_bound} in \Cref{lem:oftrl_bound}.
We consider OFTRL with the regularizer $\bar{\psi}^\t(\cdot) = \frac{1}{\eta^\t} \bar{\phi}(\cdot)$ for $\bar{\phi}(x) = \phi(x) - \min_{x' \in \Delta_d} \phi(x') \geq 0$, where we note that the output of this OFTRL at each round $t$ is same as $x^\t$.
Fix $x^* \in \calK$ and define $v = (1 - 1/T) x^* + (1/T) x^{(1)} \in \calK$, which satisfies $(1 - 1/T) (x^* - x^{(1)}) = v - x^{(1)}$.
Then, from \Cref{lem:oftrl_bound} with the fact that $\bar{\psi}^\t$ is nonnegative and nondecreasing, we have
\begin{align}
  &
  \sumT \inpr{x^\t - x^*, \ell^\t}
  =
  \sumT \inpr{x^\t - v, \ell^\t}
  +
  \sumT \inpr{v - x^*, \ell^\t}
  \nn
  &\leq
  \frac{\bar{\phi}(v)}{\eta^{(T+1)}}
  +
  \sumT 
  \prn*{
    \inpr{x^\t - x^\tp, \ell^\t - m^\t}
    -
    D_{\psi^\t}(x^\tp, x^\t)
  }
  +
  3 D L
  \com 
  \label{eq:oftrl_scbarrier_p_1}
\end{align}
where we used 
$\inpr{v - x^*, \ell^\t} 
\leq 
\nrm{v - x^*} \nrm{\ell^\t}_* 
=
\frac{1}{T}
\nrm{x^{(1)} - x^*} \nrm{\ell^\t}_* 
\leq 
D L / T
$.

We first upper bound the first term in~\Cref{eq:oftrl_scbarrier_p_1}.
From 
\Cref{lem:sc_barrier_diff}, $\bar{\phi}(v)$ in the first term in \Cref{eq:oftrl_scbarrier_p_1} is upper bounded by
\begin{equation}
  \bar{\phi}(v)
  % &
  =
  \bar{\phi}(v)
  -
  \bar{\phi}(x^{(1)})
  % \nn 
  % &
  \leq 
  \vartheta \log \prn*{\frac{1}{1 - \pi(v; x^{(1)})}}
  \leq 
  \vartheta \log T
  \com 
  \label{eq:oftrl_scbarrier_penalty}
\end{equation}
where the equality follows from $\bar{\phi}(x^{(1)}) = 0$ and the last inequality follows from $\pi(v; x^{(1)}) \leq 1 - 1/T$ since
% \begin{equation}
$
  x^{(1)} + \prn*{1 - 1/T}^{-1} (v - x^{(1)})
  =
  x^{(1)} + \prn{x^* - x^{(1)}}
  =
  x^{(1)} 
  \in
  \calK
$.

We next upper bound the second term in \Cref{eq:oftrl_scbarrier_p_1}.
From Taylor's theorem, 
there exists a point $\xi^\t = \gamma x^\tp + (1-\gamma)x^\t$ with some $\gamma \in [0,1]$ such that $D_\phi(x^\tp, x^\t) \geq \frac12 \nrm{x^\tp - x^\t}_{\xi^\t, \phi}^2$.
Hence from \Cref{lem:sc_hessian_stab},
\begin{align}
  D_\phi(x^\tp, x^\t)
  &\geq 
  \frac12 \nrm{x^\tp - x^\t}_{\xi^\t, \phi}^2
  \nn 
  &\geq 
  \frac12 
  \prn*{ 1 - \nrm{\xi^\t - x^\t}_{x^\t, \phi}}^2
  \nrm{x^\tp - x^\t}_{x^\t, \phi}^2
  \nn
  &=
  \frac12 
  \prn*{ 1 - \gamma \nrm{x^\tp - x^\t}_{x^\t, \phi}}^2
  \nrm{x^\tp - x^\t}_{x^\t, \phi}^2
  \nn
  &\geq
  \frac18
  \nrm{x^\tp - x^\t}_{x^\t, \phi}^2
  \com 
  % \label{eq:stab_lb_1}
  \n
\end{align}
where the last inequality follows from the assumption that
$\nrm{x^\tp - x^\t}_{x^\t, \phi} \leq 1/2$.
Hence,
$
D_{\psi^\t}(x^\tp, x^\t) 
=
\frac{1}{\eta^\t} D_\phi(x^\tp, x^\t) 
\geq 
\frac{1}{8\eta^\t} \nrm{x^\tp - x^\t}_{x^\t,\phi}^2
$.
Therefore, using this inequality, we have
\begin{align}
  &
  \inpr{x^\t - x^\tp, \ell^\t - m^\t}
  -
  D_{\psi^\t}(x^\tp, x^\t)
  \nn 
  &\leq 
  % \prn*{
  \nrm{x^\t - y^*}_{x^\t,\phi} \cdot \nrm{\ell^\t - m^\t}_{*,x^\t,\phi}
  -
  \frac{1}{8 \eta^\t} \nrm{x^\tp - x^\t}_{x^\t,\phi}^2
  % }
  % -
  % \frac{1}{16 \eta^\t} \nrm{x^\tp - x^\t}_{x^\t,\phi}^2
  \nn 
  &\leq 
  4 \eta^\t \nrm{ \ell^\t - m^\t }_{*,x^\t,\phi}^2
  -
  \frac{1}{16 \eta^\t} \nrm{x^\tp - x^\t}_{x^\t,\phi}^2
  \com 
  \label{eq:oftrl_scbarrier_stab}
\end{align}
where the first inequality follows from H\"older's inequality and 
$b \sqrt{z} - a z \leq b^2 / (4a)$ for $a > 0, b \geq 0$, and $z \geq 0$.
Combining \Cref{eq:oftrl_scbarrier_p_1} with \Cref{eq:oftrl_scbarrier_penalty,eq:oftrl_scbarrier_stab} completes the proof.
\end{proof}

\paragraph{OFTRL with log-barrier regularizer}
Now we are ready to prove the RVU bound for OFTRL with the log-barrier regularizer and adaptive learning rate,
which serves as a foundation for our analysis of multi-player general-sum games in \Cref{sec:multi_player}.
\begin{restatable}[RVU bound for OFTRL with log-barrier regularizer and adaptive learning rate]{lemma}{lemoftrllogbarrier}\label[lemma]{lem:oftrl_logbarrier}
  Let
  $\psi^\t(x) = \frac{1}{\eta^\t} \phi(x)$ for $\phi(x) = - \sum_{k=1}^d \log(x(k))$ be the logarithmic barrier regularizer with nonincreasing learning rate $\eta^\t$
  and 
  $x^\t \in \argmin_{x \in \Delta_d} \set{\inpr{x, m^\t + \sum_{s=1}^{t-1} \ell^\s} + \psi^\t(x)}$ be the output of OFTRL at round $t$.
  Suppose that 
  $\nrm{x^\tp - x^\t}_{x^\t,\phi} \leq 1/2$.
  Then, for any $x^* \in \Delta_d$,
  \begin{align}
    \sumT \inpr{x^\t \!-\! x^*, \ell^\t}
    \!\leq\!
    \frac{d \log T}{\eta^{(T+1)}}
    \!+\!
    \sumT 
    4 \eta^\t \nrm{\ell^\t \!-\! m^\t}_{*,x^\t,\phi}^2 
    \!-\! 
    \sumT 
    \frac{1}{16 \eta^\t} 
    \nrm{x^\tp \!-\! x^\t}_{x^\t, \phi}^2
    % D_\phi(x^\tp, x^\t)
    \!+\!
    6 L
    \com 
    \label{eq:oftrl_logbarrier}
  \end{align}
  where $L = \max \set{ \max_{t \in [T]} \nrm{\ell^\t}_*,  \max_{t \in [T+1]} \nrm{m^\t}_*  }$.
\end{restatable}
% \lemoftrllogbarrier*

\begin{proof}
Combining \Cref{lem:oftrl_selfconcordant} with the fact that $\phi$ is a $d$-self-concordant barrier and $\max_{x, y \in \Delta_d} \nrm{x - y}_1 = 2$ yields the desired bound.
\end{proof}

% The log-barrier function is a self-concordant barrier and 
% we provide the RVU bound for OFTRL with a general self-concordant barrier regularizer and an adaptive learning rate in \Cref{app:proof_oftrl}. 
\begin{comment}
For a constant learning rate,
the RVU bound for OFTRL has been investigated~\citep{anagnostides22uncoupled}.
However, to our knowledge, no RVU bounds have been established for the case with adaptive learning rate, which could be of independent interest.
One can see that the key to proving \Cref{lem:oftrl_logbarrier} (and more generally \Cref{lem:oftrl_selfconcordant}) lies in the stability assumption for $x^\t$, namely $\nrm{x^\tp - x^\t}_{x^\t,\phi} \leq 1/2$, which can be verified in swap regret minimization by analyzing the stability of a Markov chain in \Cref{sec:multi_player} and \Cref{app:proof_multi_player}.
\end{comment}
For a constant learning rate, the RVU bound for OFTRL has been studied~\citep{anagnostides22uncoupled}. However, to our knowledge, no RVU bounds have been established for the case with an adaptive learning rate, which is of independent interest. One can observe that the key to proving \Cref{lem:oftrl_logbarrier} (and more generally \Cref{lem:oftrl_selfconcordant}) is the stability assumption for $x^\t$, namely, $\nrm{x^\tp - x^\t}_{x^\t,\phi} \leq 1/2$. This assumption can be verified in swap regret minimization by analyzing the stability of a Markov chain, as detailed in \Cref{sec:multi_player} and \Cref{app:proof_multi_player}.

\section{Deferred Details from \Cref{sec:three_regimes}}\label{app:proof_regimes}
Here we provide the omitted details from \Cref{sec:three_regimes}.
We first provide the proof of \Cref{prop:reg_relation}.
\begin{proof}[Proof of \Cref{prop:reg_relation}]
  From the triangle inequality and the Cauchy--Schwarz inequality,
  we have 
  % $
  \begin{equation}    
    \abs*{\Reg_{x_i,u_i}^T(x^*) - \Reg_{\xhat_i,u_i}^T(x^*)}
    =
    \abs*{\sumT \inpr{\xhat_i^\t - x_i^\t, u_i^\t}}
    \leq 
    \sumT \nrm{\xhat_i^\t - x_i^\t}_1 \nrm{u_i^\t}_\infty 
    \leq 
    \hat{C}_i
    \per
    \n
  \end{equation}
  % $ 
  Similarly, we also have
  % $
  \begin{equation}
    \abs*{\Reg_{x_i,u_i}^T(x^*) \!-\! \Reg_{x_i,\tilde{u}_i}^T(x^*)}
    =
    \abs*{\sumT \inpr{x^* \!-\! x_i^\t, u_i^\t \!-\! \tilde{u}_i^\t}}
    \leq 
    \sumT \nrm{x^* - x_i^\t}_1 \nrm{u_i^\t - \tilde{u}_i^\t}_\infty 
    \leq 
    2 \tilde{C}_i
    \per
    \n
  \end{equation}
  % .
  % $
  The other inequalities can be proven in the same manner.
\end{proof}

A similar proposition to \Cref{prop:reg_relation} also holds for the four types of swap regret in \Cref{sec:three_regimes}.
\begin{proposition}\label[proposition]{prop:swapreg_relation}
  For any $i \in [n]$ and $M \in \calM_{m_i}$, it holds that 
  $\abs{\SwapReg_{x_i,u_i}^T(M) - \SwapReg_{\xhat_i,u_i}^T(M)} \leq 2 \hat{C}_i$,
  $\abs{\SwapReg_{x_i,\tilde{u}_i}^T(M) - \SwapReg_{\xhat_i,\tilde{u}_i}^T(M)} \leq 2 \hat{C}_i$,
  $\abs{\SwapReg_{x_i,u_i}^T(M) - \SwapReg_{x_i,\tilde{u}_i}^T(M)} \leq 2 \tilde{C}_i$, and 
  $\abs{\SwapReg_{\xhat_i,u_i}^T(v) - \SwapReg_{\xhat_i,\tilde{u}_i}^T(M)} \leq 2 \tilde{C}_i$.
\end{proposition}

\begin{proof}[Proof of \Cref{prop:swapreg_relation}]
  From the triangle inequality and the Cauchy--Schwarz inequality,
  we have 
  % $
  \begin{align}
    &
    \abs*{\SwapReg_{x_i,u_i}^T(M) - \SwapReg_{\xhat_i,u_i}^T(M)}
    =
    \abs*{\sumT \inpr{\xhat_i^\t - x_i^\t, M u_i^\t - u_i^\t}}
    \nn
    &\qquad\leq 
    \sumT \nrm{\xhat_i^\t - x_i^\t}_1 \nrm{M u_i^\t - u_i^\t}_\infty 
    \leq 
    2 \hat{C}_i
    \per
    \n
  \end{align}
  % $ 
  Similarly, we also have
  % $
  \begin{align}
    &
    \abs*{\SwapReg_{x_i,u_i}^T(M) - \SwapReg_{x_i,\tilde{u}_i}^T(M)}
    =
    \abs*{\sumT \inpr*{x_i^\t, M \prn{u_i^\t - \tilde{u}_i^\t} - \prn{u_i^\t - \tilde{u}_i^\t}}}
    \nn
    &\leq 
    \sumT \nrm{x_i^\t}_1 
    \prn*{
      \nrm{M \prn{u_i^\t - \tilde{u}_i^\t}}_\infty
      +
      \nrm{u_i^\t - \tilde{u}_i^\t}_\infty 
    }
    \leq 
    2
    \sumT
    \nrm{u_i^\t - \tilde{u}_i^\t}_\infty 
    \leq 
    2 \tilde{C}_i
    \com 
    \n
  \end{align}
  where we used $\nrm{M x}_\infty = \max_{k \in [m]} \abs{\inpr{M(k, \cdot), x}} \leq \max_{k \in [m]} \nrm{M(k, \cdot)}_1 \nrm{x}_\infty \leq \nrm{x}_\infty$ for a row stochastic matrix $M \in \calM_m$ and a vector $x \in \R^m$.
  The other inequalities can be proven in the same manner.
\end{proof}

\section{Deferred Proofs for Two-player Zero-sum Games from \Cref{sec:two_player}}\label{app:proof_two_player}
This section provides the details and deferred proofs from \Cref{sec:two_player}.
\subsection{Corrupted procedure in two-player zero-sum games}\label{subsec:corrupted_game_two_player}
For clarity, we summarize the corrupted learning procedure for a two-player zero-sum game with payoff matrix $A$.
This procedure corresponds to the one used for multi-player general-sum games described in \Cref{sec:three_regimes}.
\begin{mdframed}
At each round $t = 1, \dots, T$: %  [T] = \set{1,\dots,T}$:
\begin{enumerate}[topsep=2pt, itemsep=0pt, partopsep=0pt, leftmargin=25pt]
  \item A prescribed algorithm suggests strategies $\xhat^\t \in \Delta_{\mx}$ and $\yhat^\t \in \Delta_{\my}$;
  \item $x$-player selects a strategy $x^\t \leftarrow \xhat^\t + c_\xrm^\t$ and $y$-player selects $y^\t \leftarrow \yhat^\t + c_\yrm^\t$; %  \hspace{\fill}  \# corruption of strategies
  \item $x$-player observes corrupted expected reward vector $\gtil^\t = g^\t + \tilde{c}_\xrm^\t$ for $g^\t = A y^\t$ and $y$-player observes corrupted expected loss vector $\ltil^\t = \ell^\t + \tilde{c}_\yrm^\t$ for $\ell^\t = A^\top x^\t$;
  \item $x$-player gains a payoff of $\inpr{x^\t, g^\t}$ in Setting (I) and $\inpr{x^\t, \gtil^\t}$ in Setting (II), and $y$-player incurs a loss of $\inpr{y^\t, \ell^\t}$ in Setting (I) and $\inpr{y^\t, \ltil^\t}$ in Setting (II);
\end{enumerate}
\end{mdframed}
\begin{comment}
Here, $c_\xrm^\t$ and $c_\yrm^\t$ are corruption for round $t$ such that 
$\sumT \nrm{c_\xrm^\t}_1 = \sumT \nrm{ x^\t - \xhat^\t }_1 \leq \Cx$ 
and
$\sumT \nrm{c_\yrm^\t}_1 = \sumT \nrm{ y^\t - \yhat^\t }_1 \leq \Cy$,
where we use $\Cx$ and $\Cy$ to denote the upper bounds on the cumulative amount of corruptions in terms of the $\ell_1$-norm.
\end{comment}
Here, $\hat{c}_\xrm^\t$ and $\tilde{c}_\xrm^\t$ are corruption vectors for strategies and utility of $x$-player at round $t$, respectively, such that 
$\sumT \nrm{\hat{c}_\xrm^\t}_1 = \sumT \nrm{ x^\t - \xhat^\t }_1 \leq \hat{C}_\xrm$, 
$\sumT \nrm{\tilde{c}_\xrm^\t}_\infty = \sumT \nrm{ g^\t - \gtil^\t }_\infty \leq \tilde{C}_\xrm$, and $C_\xrm = \hat{C}_\xrm + 2 \tilde{C}_\xrm$. 
Similarly, $\hat{c}_\yrm^\t$ and $\tilde{c}_\yrm^\t$ are corruption levels of strategies and utility of $y$-player, respectively, such that 
$\sumT \nrm{\hat{c}_\yrm^\t}_1 = \sumT \nrm{ y^\t - \yhat^\t }_1 \leq \hat{C}_\yrm$, 
$\sumT \nrm{\tilde{c}_\yrm^\t}_\infty = \sumT \nrm{ \ell^\t - \ltil^\t }_\infty \leq \tilde{C}_\yrm$, and $C_\yrm = \hat{C}_\yrm + 2 \tilde{C}_\yrm$.

\subsection{Preliminary lemmas}

\begin{lemma}\label[lemma]{lem:regx_regy_nn}
  Let $(x^*, y^*)$ be a Nash equilibrium.
  Then, $\Reg_{x,g}^T(x^*) + \Reg_{y,\ell}^T(y^*) \geq 0$.
\end{lemma}
\begin{proof}
We have 
\begin{align}  
&
\Reg_{x,g}^T(x^*) + \Reg_{y,\ell}^T(y^*)
=
\sumT \inpr{x^* - x^\t, A y^\t}
+
\sumT \inpr{y^\t - y^*, A x^\t}
\nn
&=
\sumT \inpr{x^*, A y^\t}
-
\sumT \inpr{ y^*, A x^\t}
=
\sumT \inpr{x^*, A y^\t - A y^*}
-
\sumT \inpr{y^*, A x^\t - A x^*}
\geq 
0,
\n
\end{align}
where the last inequality follows from the definition of the Nash equilibrium.
\end{proof}

\subsection{Proof of \Cref{lem:reg_xhat_bound}}

\begin{proof}[Proof of \Cref{lem:reg_xhat_bound}]
From \Cref{lem:oftrl_shannon},
the regret of $x$-player is bounded by
\begin{equation}
  % \Reg_\xrm^T 
  \Reg_{\xhat, \gtil}^T
  % &=
  % \max_{u \in \Delta_{\mx}} \set*{
  %   \sumT \inpr{u, g^\t}
  % }
  % -
  % \sumT \inpr{\xhat^\t, g^\t}
  % +
  % \sumT \inpr{\xhat^\t - x^\t, g^\t}
  % \nn 
  % &
  \leq 
  \frac{\log \mx}{\eta_\xrm^{(T+1)}}
  +
  \sumT \eta_\xrm^\t \nrm{\gtil^\t - \gtil^\tm}_\infty^2
  -
  \sum_{t=1}^T
  \frac{1}{4 \eta_\xrm^\t} \nrm{\xhat^\tp - \xhat^\t}_1^2
  % +
  % \Cx
  +
  2
  \com 
  \label{eq:reg_x_bound_1}
\end{equation}
From the definition of $\eta_\xrm^\t$, the third term in~\eqref{eq:reg_x_bound_1} is upper bounded by
\begin{equation}\label{eq:intsqrt}
  \sumT \eta_\xrm^\t \nrm{\gtil^\t - \gtil^\tm}_\infty^2
  \leq
  \sqrt{\frac{\logp(\mx)}{2}}
  \sumT \frac{\nrm{\gtil^\t - \gtil^\tm}_\infty^2}{ \sqrt{4 + P_\infty^{t-1}(\gtil) } }
  \leq 
  % \sqrt{2 \sumT \nrm{\gtil^\t - \gtil^\tm}_\infty^2 \logp(\mx)}
  \sqrt{2 P_\infty^T(\gtil) \logp(\mx)}
  \com 
\end{equation}
where the first inequality follows from $\logp(\mx) \geq 4$ and the last inequality from
$\sumT z_t / \sqrt{4 + \sum_{s=1}^{t-1} z_s} \leq 2 \sqrt{\sumT z_t}$ for $z_t \in [0,4]$.
Hence, from \Cref{eq:reg_x_bound_1}, \Cref{eq:intsqrt}, and the definition of $\eta_\xrm^{(T+1)}$, we obtain
\begin{equation}\label{eq:reg_x_bound_2}
  % \Reg_\xrm^T 
  \Reg_{\xhat, \gtil}^T
  \leq 
  2 \sqrt{2 \logp(\mx) \prn*{\logp(\mx) + P_\infty^T(\gtil)} } 
  -
  \sum_{t=1}^T
  \frac{1}{4 \eta_\xrm^\t} \nrm{\xhat^\t - \xhat^\tm}_1^2
  +
  2
  \per 
\end{equation}
Now we have
$
P_\infty^T(\gtil)
=
\sumT \nrm{\gtil^\t - \gtil^\tm}_\infty^2 
\leq 
2 \sumT \nrm{g^\t - g^\tm}_\infty^2
+
8 \tilde{C}_\xrm 
$
and
\begin{align}
  &
  \sumT \nrm{g^\t - g^\tm} _\infty^2 
  =
  \sumT \nrm{A (y^\t - y^\tm)}_\infty^2 
  \leq 
  \sum_{t=1}^T \nrm{y^\t - y^\tm}_1^2 % + 1
  \nn 
  &\leq 
  \sum_{t=1}^T
  \prn*{
    2 \nrm{y^\t - \yhat^\t}_1^2 
    +
    4 \nrm{\yhat^\t - \yhat^\tm}_1^2 
    +
    2 \nrm{\yhat^\tm - y^\tm}_1^2 
  }
  % + 1
  \nn
  &\leq 
  4
  \sumT
  \nrm{y^\t - \yhat^\t}_1^2 
  +
  4 \sum_{t=1}^T 
  \nrm{\yhat^\t - \yhat^\tm}_1^2 
  % + 1
  \leq 
  8 \hat{C}_\yrm
  +
  % 4 \sum_{t=2}^T 
  % \nrm{\yhat^\t - \yhat^\tm}_1^2
  4 P_1^T(\yhat)
  % + 1
  \per
  \n
\end{align}
Hence, plugging the last two inequalities in \Cref{eq:reg_x_bound_2} and using $\eta_\xrm^\t \leq 1/\sqrt{2}$ for all $t \in [T]$, we obtain
\begin{equation}%\label{eq:reg_x_bound_3}
  \Reg_{\xhat, \gtil}^T
  \leq
  2 \sqrt{ 2 \logp(\mx) 
  \prn*{
    % 1
    % + 
    \logp(\mx)
    +
    8 (\hat{C}_\yrm + \tilde{C}_\xrm)
    +
    % 4 \sum_{t=2}^T 
    % \nrm{\yhat^\t - \yhat^\tm}_1^2
    4 P_1^T(\yhat)
  }
  }
  - 
  \frac{1}{\sqrt{8}}
  % \sum_{t=2}^T \nrm{\xhat^\t - \xhat^\tm}_1^2
  P_1^T(\xhat)
  +
  2
  \per 
  \n
\end{equation}

We next consider the regret of $y$-player.
Using \Cref{lem:oftrl_shannon} and following the similar analysis as above, we can upper bound the regret of $y$-player as
\begin{equation}\label{eq:reg_y_bound_1}
  % \Reg_\yrm^T 
  \Reg_{\yhat, \tilde{\ell}}^T
  \leq 
  2 \sqrt{ 2 \logp(\my) \prn*{\logp(\my) + \sum_{t=1}^{T} \nrm{\ltil^\t - \ltil^\tm}_\infty^2} } 
  -
  \sum_{t=2}^T
  \frac{1}{4 \eta_\yrm^\t} \nrm{\yhat^\t - \yhat^\tm}_1^2
  % \green{+
  % \Cy}
  +
  2
  \per
\end{equation}
% The summation in the first term is further bounded by
Now we have
$
% P_\infty^T(\ltil)
% =
\sumT \nrm{\ltil^\t - \ltil^\tm}_\infty^2 
\leq 
2 \sumT \nrm{\ell^\t - \ell^\tm}_\infty^2
+
8 \tilde{C}_\yrm
$
and
\begin{align}
  &
  \sumT \nrm{\ell^\t - \ell^\tm} _\infty^2
  \leq 
  \sum_{t=1}^T \nrm{x^\t - x^\tm}_1^2 
  % + 1
  \nn 
  &\leq 
  \sum_{t=1}^T
  \prn*{
    2 \nrm{x^\t - \xhat^\t}_1^2 
    +
    4 \nrm{\xhat^\t - \xhat^\tm}_1^2 
    +
    2 \nrm{\xhat^\tm - x^\tm}_1^2 
  }
  % + 1
  \nn
  &\leq 
  4
  \sumT
  \nrm{x^\t - \xhat^\t}_1^2 
  +
  4 \sum_{t=1}^T 
  \nrm{\xhat^\t - \xhat^\tm}_1^2 
  % + 1
  \leq 
  8 \hat{C}_\xrm
  +
  % 4 \sum_{t=2}^T 
  % \nrm{\xhat^\t - \xhat^\tm}_1^2
  4 P_1^T(\xhat)
  % + 1
  \per 
  \n
\end{align}
Combining \Cref{eq:reg_y_bound_1} with the last two inequalities and using $\eta_\yrm^\t \leq 1/\sqrt{2}$, we obtain
\begin{equation}%\label{eq:reg_y_bound_2}
  \Reg_{\yhat, \tilde{\ell}}^T
  \leq
  2 \sqrt{ 2 \logp(\my) 
  \prn*{
    % 1
    % + 
    \logp(\my)
    +
    8 \prn{ \hat{C}_\xrm + \tilde{C}_\yrm }
    +
    % 4 \sum_{t=2}^T 
    % \nrm{\xhat^\t - \xhat^\tm}_1^2
    4 P_1^T(\xhat)
  }
  }
  - 
  \frac{1}{\sqrt{8}} 
  % \sum_{t=2}^T \nrm{\yhat^\t - \yhat^\tm}_1^2
  P_1^T(\yhat)
  +
  2
  \com
  \n
\end{equation}
which completes the proof.
\end{proof}

\subsection{Remaining proof of \Cref{thm:indiv_reg_corrupt}}\label{app:remaining_proof_reg}

The following bounds are omitted bounds in \Cref{thm:indiv_reg_corrupt} for $y$-player:
\begin{align}
  % \Reg_\yrm^T 
  % =
  \Reg_{y,\ell}^T
  &\lesssim
    \min\set[\Big]{
    \sqrt{
      \prn*{
        \log(\mx \my)
        +
        \Cx 
        +
        \Cy
      }
      \log \my
    }
    ,
    \sqrt{
      \prn[\big]{
        % \sumT \nrm{\ell^\t \!-\! \ell^\tm}_\infty^2 \!
        P_\infty^T(\ltil)
        + \log \my
      } \log \my
    }
    }
    +
    \Cy
    % &
    % \!\!\mbox{corrupted regime}
    \com
  % \end{cases}
  \nn
% \end{equation}
% \begin{equation}
  % \Reg_\yrm^T 
  % =
  \Reg_{y,\ltil}^T
  &\lesssim
  % \widetilde{O}\prn{\sqrt{1 + \Cy} + \Cx}
  % O\prn[\Big]{
  % \begin{cases}
    % \sqrt{
      % \log(\mx \my)
      % \log \my
    % } & \!\! \mbox{honest regime} 
    % \\ 
    \min\set[\Big]{
    \sqrt{
      \prn*{
        \log(\mx \my)
        +
        \Cx 
        +
        \Cy
      }
      \log \my
    }
    ,
    \sqrt{
      \prn[\big]{
        % \sumT \nrm{\ell^\t \!-\! \ell^\tm}_\infty^2 \!
        P_\infty^T(\ltil)
        + \log \my
      } \log \my
    }
    }
    +
    \Chaty
    % &
    % \!\!\mbox{corrupted regime}
    \per 
  % \end{cases}
  \n
\end{align}

Below, we include the omitted proofs of \Cref{thm:indiv_reg_corrupt}:
\begin{proof}[Remaining proof of \Cref{thm:indiv_reg_corrupt}]
  % \green{
  Here we provide the remaining proof to prove the upper bounds on $\Reg_{x,\gtil}^T$ and $\Reg_{y,\ltil}^T$.
  From \Cref{lem:reg_xhat_bound} and \Cref{prop:reg_relation},
  we have
  \begin{align}
    \Reg_{x, \gtil}^T
    &\leq
    \hat{C}_{\xrm}
    \!+\!
    \sqrt{ 8 \logp(\mx) 
    \prn*{
      1
      + 
      \logp(\mx)
      +
      8 (\hat{C}_\yrm + \tilde{C}_\xrm)
      +
      4 P_1^T(\yhat)
    }
    }
    \!-\! 
    \frac{1}{\sqrt{8}} P_1^T(\xhat)
    \!+\!
    2
    \com
    \label{eq:reg_x_bound_4_til}
    \\
  % \end{equation}
  % \begin{equation}
    \Reg_{y, \tilde{\ell}}^T
    &\leq
    % \green{
    \hat{C}_{\yrm}
    \!+\! 
    % }
    \sqrt{ 8 \logp(\my) 
    \prn*{
      1
      + 
      \logp(\my)
      +
      8 \prn{ \hat{C}_\xrm + \tilde{C}_\yrm }
      +
      P_1^T(\xhat)
    }
    }
    \!-\! 
    \frac{1}{\sqrt{8}} P_1^T(\yhat)
    \!+\!
    2
    \per 
    \label{eq:reg_y_bound_3_til}
  \end{align}
  Summing up the last two upper bounds,
  we have 
  \begin{align}
    &
    \Reg_{x,\gtil}^T + \Reg_{y,\ltil}^T
    \nn
    &\leq
    \sqrt{ 8 \logp(\mx) 
    \prn*{
      1
      +
      \logp(\mx)
      + 
      8 \prn{ \hat{C}_\yrm + \tilde{C}_\xrm }
      +
      4 P_1^T(\yhat)
    }
    }
    - 
    \frac{1}{\sqrt{8}} P_1^T(\xhat)
    +
    \Chatx
    +
    4
    \nn 
    &\quad+
    \sqrt{ 8 \logp(\my) 
    \prn*{
      1
      +
      \logp(\my)
      +
      8 \prn{ \hat{C}_\xrm + \tilde{C}_\yrm }
      +
      4
      P_1^T(\xhat)
    }
    }
    -
    \frac{1}{\sqrt{8}}
    P_1^T(\yhat)
    +
    \Chaty
    \nn 
    &=
    O\prn[\Big]{
      \!
      \sqrt{\!\prn{\hat{C}_\yrm \!+\! \tilde{C}_\xrm} \log \mx}
      \!+\!
      \sqrt{\!\prn{\hat{C}_\xrm \!+\! \tilde{C}_\yrm} \log \my}
      \!+\!
      \log(\mx \my)
      \!+\!
      \Chatx 
      \!+\!
      \Chaty
    }
    \!-\!
    \frac{1}{4\sqrt{2}} 
    \prn*{
      P_1^T(\xhat) \!+\! P_1^T(\yhat)
    }
    \com 
    % \n
    \label{eq:regtil_sum_upper}
  \end{align}
  where we used $b\sqrt{z} - az \leq b^2 / (4 a)$ for $a > 0$, $b \geq 0$ and $z \geq 0$.
  Now from \Cref{prop:reg_relation} and \Cref{lem:regx_regy_nn}, we have 
  \begin{equation}
    \Reg_{x,\gtil}^T + \Reg_{x,\ltil}^T 
    \geq 
    \Reg_{x,g}^T + \Reg_{y,\ell}^T - 2 \prn{\tilde{C}_\xrm + \tilde{C}_\yrm}
    \geq 
    - 2 \prn{\tilde{C}_\xrm + \tilde{C}_\yrm}
    \per
    \n
  \end{equation}
  Combining this inequality and~\Cref{eq:regtil_sum_upper},
  we have 
  \begin{align}
    % &
    P_1^T(\xhat) + P_1^T(\yhat)
    % =
    % O\prn*{
    % \nn
    &\lesssim
    \sqrt{\prn{\hat{C}_\yrm + \tilde{C}_\xrm}\log \mx}
    +
    \sqrt{\prn{\hat{C}_\xrm + \tilde{C}_\yrm} \log \my}
    +
    \log(\mx \my)
    +
    \Cx 
    +
    \Cy
    \nn
    &\lesssim
    \log(\mx \my)
    +
    \Cx 
    +
    \Cy
    % }
    \com 
    \label{eq:second_order_bounds_til}
  \end{align}
  where the last line follows from the AM--GM inequality and the definitions of $\Cx$ and $\Cy$.
  Finally, plugging \Cref{eq:second_order_bounds_til} in \Cref{eq:reg_x_bound_4_til,eq:reg_y_bound_3_til} gives the desired bounds on $\Reg_{x,\gtil}^T$ and $\Reg_{y,\ltil}^T$.
  The upper bound of $\Reg_{x,\gtil}^T \lesssim \sqrt{\prn{P_\infty^T(\gtil) + \log \mx} \log \mx} + \Cx$ follows from \Cref{eq:reg_x_bound_2} in the proof of \Cref{lem:reg_xhat_bound}.
\end{proof}

\section{Deferred Proofs for Multi-player General-sum Games from \Cref{sec:multi_player}}\label{app:proof_multi_player}

\subsection{Proof of \Cref{lem:swap2base_regret}}
\begin{comment}
\begin{proof}
  The swap regret of the sequence of strategies $\set{\xhat_i^\t}_{t=1}^T$ is can be decomposed as
  \begin{align}
    \widehat{\SwapReg}_i^T(M)
    &=
    \sumT \inpr{\xhat_i^\t, M u_i^\t- u_i^\t}
    \nn
    &=
    \sumT \inpr*{M^\top \xhat_i^\t - \prn{Q_i^\t}^\top \xhat_i^\t, u_i^\t}
    +
    \sumT \inpr*{\prn{Q_i^\t}^\top \xhat_i^\t - \xhat_i^\t, u_i^\t}
    \nn
    &=
    \sumT \inpr*{\prn{M - Q_i^\t}^\top \xhat_i^\t, u_i^\t}
    \tag{since $\prn{Q_i^\t}^\top \xhat_i^\t = \xhat_i^\t$}
    \\ 
    &=
    \sumT \inpr*{M - Q_i^\t, u_i^\t \prn{\xhat_i^\t}^\top}
    \nn 
    &=
    \sum_{a \in \calA_i} \sumT \inpr*{M(a,\cdot) - Q_i^\t(a,\cdot), \xhat_i^\t(a) u_i^\t}
    \nn
    &=
    \sum_{a \in \calA_i} \sumT \inpr*{M(a,\cdot) - y_{i,a}^\t, u_{i,a}^\t}
    \tag{$u_{i,a}^\t = \xhat_i^\t(a) u_i^\t$ }
    =
    \sum_{a \in \calA_i} \Reg_{i,a}^T(M(a,\cdot))
    \com 
  \end{align}
  which completes the proof.
\end{proof}
\end{comment}
\begin{proof}
  % \red{(new)}
  The swap regret $\SwapReg_{\xhat_i,\tilde{u}_i}^T(M)$ can be decomposed as
  \begin{align}
    \SwapReg_{\xhat_i,\tilde{u}_i}^T(M)
    &=
    \sumT \inpr{\xhat_i^\t, M \util_i^\t- \util_i^\t}
    \nn
    &=
    \sumT \inpr*{M^\top \xhat_i^\t - \prn{Q_i^\t}^\top \xhat_i^\t, \util_i^\t}
    +
    \sumT \inpr*{\prn{Q_i^\t}^\top \xhat_i^\t - \xhat_i^\t, \util_i^\t}
    \nn
    &=
    \sumT \inpr*{\prn{M - Q_i^\t}^\top \xhat_i^\t, \util_i^\t}
    \tag{since $\prn{Q_i^\t}^\top \xhat_i^\t = \xhat_i^\t$}
    \\ 
    &=
    \sumT \tr\prn*{ \prn*{M - Q_i^\t}  \util_i^\t \prn{\xhat_i^\t}^\top}
    % =
    % \sumT 
    % \inpr*{ \prn*{ M - Q_i^\t }^\top, \util_i^\t \prn{\xhat_i^\t}^\top}
    \nn 
    &=
    \sum_{a \in \calA_i} \sumT 
    \inpr*{ M(a,\cdot) - Q_i^\t(a,\cdot), \xhat_i^\t(a) \util_i^\t}
    \nn
    &=
    \sum_{a \in \calA_i} \sumT \inpr*{M(a,\cdot) - y_{i,a}^\t, \util_{i,a}^\t}
    \tag{$\util_{i,a}^\t = \xhat_i^\t(a) \util_i^\t$ }
    \nn
    &=
    \sum_{a \in \calA_i} \widetilde{\Reg}_{i,a}^T(M(a,\cdot))
    \com 
    \n
  \end{align}
  which completes the proof.
\end{proof}

\subsection{Preliminary lemmas for the proof of \Cref{thm:indiv_swapreg}}\label{subsec:pre_lemmas_swap}
Here we prepare several lemmas to prove~\Cref{thm:indiv_swapreg}.
We begin with the following lemma, which is useful to evaluate the stability of the Markov chain in the proof of \Cref{lem:diff_local_oftrl}.
\begin{lemma}\label[lemma]{lem:localnorm_suffcond}
  Let $\calK$ be a bounded nonempty convex set and $y \in \calK$.
  Let $\delta > 0$ and $f$ be a real-valued strictly convex function over $\calK$ and $x^* = \argmin_{x' \in \calK} f(x')$ be the unique minimizer of $f$.
  Suppose that for any $z \in \calK$ such that $\nrm{z - y} = \delta$ for a norm $\nrm{\cdot}$, it holds that $f(z) \geq f(y)$.
  Then, $\nrm{x^* - y} < \delta$.
\end{lemma}
\begin{proof}
We begin by proving that 
for any $z \in \R^d$ such that $\nrm{z - y} \geq \delta$, it holds that $f(z) \geq f(y)$.
This follows from the convexity of $f$.
To formally prove this,
we take $z \in \calK$ satisfying $\nrm{z - y} > \delta$ arbitrarily
and
define $w(\gamma) = (1 - \gamma) y + \gamma z = y + \gamma (z - y)$ for $\gamma \in [0,1]$.
Define $\gamma^\circ = \delta/\nrm{z - y}$,
which satisfies 
$\nrm{w(\gamma^\circ) - y} = \gamma^\circ \nrm{z - y} = \delta$.
Then, we have
\begin{equation}\label{eq:p_localnorm_suffcond_1}
  f(y)
  \leq 
  f(w(\gamma^\circ))
  \leq 
  (1 - \gamma^\circ) f(y) + \gamma^\circ f(z)
  \com 
\end{equation}
where the first inequality follows from 
$\nrm{w(\gamma^\circ) - y} = \delta$ and the assumption that for any $z \in \calK$ satisfying $\nrm{z - y} = \delta$, $f(z) \geq f(y)$,
and the second inequality from the convexity of $f$.
Rearranging the terms in \Cref{eq:p_localnorm_suffcond_1} gives $f(y) \leq f(z)$.

Therefore, we have 
$
  \set{z \in \calK \colon \nrm{z - y} \geq \delta}
  \subseteq
  \set{z \in \calK \colon f(z) \geq f(y)}
$
and thus from the assumption that $x^*$ is the minimizer of the strictly convex function $f$,
\begin{equation}
  x^*
  \in 
  \set{z \in \calK \colon f(z) < f(y)}
  \subseteq
  \set{z \in \calK \colon \nrm{z - y} < \delta}
  \com 
  \n
\end{equation}
which implies $\nrm{x^* - y} < \delta$.
This completes the proof.
\end{proof}

The following lemma upper bounds the increase of the reciprocal of the learning rate.
\begin{lemma}\label[lemma]{lem:lr_diff}
  The learning rate $\set{\eta_{i,a}^\t}_{t=1}^T$ 
  in \Cref{eq:oftrl_multiplayer}
  % in \Cref{eq:lr_i_corrupt} 
  satisfies
  \begin{equation}
    \frac{1}{\eta_{i,a}^\t} - \frac{1}{\eta_{i,a}^\tm}
    \leq 
    % \green{\frac{4}{m_i \sqrt{\log T}}}
    \frac{\sqrt{8} \prn{ \xhat_i^\tm(a) + \xhat_i^{(t-2)}(a)}}{\sqrt{m_i \log T}}
    \per 
    \n
  \end{equation}
\end{lemma}
\begin{proof}  
From the definition of the learning rate $\eta_{i,a}^\t$, we have
\begin{align}
  \frac{1}{\eta_{i,a}^\t} \!-\! \frac{1}{\eta_{i,a}^\tm}
  &=
  \frac{1}{
    \min\set*{
      \!\!
      \sqrt{\!\frac{m_i \log T / 8}{4 + \sum_{s=1}^{t-1} \nrm{\util_{i,a}^\s - \util_{i,a}^\sm}_\infty^2 }}
      ,
      \frac{1}{256 n \sqrt{m_i}}
      \!
    }
  }
  \!-\!
  \frac{1}{
    \min\set*{
      \!\!
      \sqrt{\!\frac{m_i \log T / 8}{4 + \sum_{s=1}^{t-2} \nrm{\util_{i,a}^\s - \util_{i,a}^\sm}_\infty^2 }}
      ,
      \frac{1}{256 n \sqrt{m_i}}
      \!
    }
  }
  \nn 
  &\leq 
  \sqrt{\frac{4 + \sum_{s=1}^{t-1} \nrm{\util_{i,a}^\s - \util_{i,a}^\sm}_\infty^2 }{m_i \log T / 8}}
  -
  \sqrt{\frac{4 + \sum_{s=1}^{t-2} \nrm{\util_{i,a}^\s - \util_{i,a}^\sm}_\infty^2 }{m_i \log T / 8}
  }
  \nn 
  &\leq 
  \sqrt{\frac{8}{m_i \log T}}
  \sqrt{\nrm{\util_i^\tm - \util_i^{(t-2)}}_\infty^2}
  \nn
  &\leq 
  \frac{
    \sqrt{8} 
    \prn{ \nrm{\util_{i,a}^\tm}_\infty + \nrm{\util_{i,a}^{(t-2)}}_\infty }
  }
  {
    \sqrt{m_i \log T}
  }
  \leq 
  \frac{
    \sqrt{8} 
    \prn{ \xhat_i^\tm(a) + \xhat_i^{(t-2)}(a) }
  }
  {
    \sqrt{m_i \log T}
  }
  \com 
  \n
\end{align}
where 
in the first inequality we used the fact that 
$z \mapsto 1/\min\set{a, z} - 1/\min\set{b, z}$ is nondecreasing in $z$ when $a \leq b$,
in the second inequality we used the subadditivity of $z \mapsto \sqrt{z}$ for $z \geq 0$,
and in the last inequality we used 
$\util_{i,a}^\t = \xhat_i^\t(a) \util_i^\t$ and $\nrm{\util_i^\t}_\infty \leq 1$.
\end{proof}

From \Cref{lem:scbarrier_properties,lem:localnorm_suffcond,lem:lr_diff}, we can prove the following key lemma, which guarantees the stability of the Markov chain under the adaptive learning rate $\eta_{i,a}^\t$ in \Cref{eq:oftrl_multiplayer}. % ~\Cref{eq:lr_i_corrupt}.
In this lemma and its proof, we ignore the player index $i \in [n]$ for notational simplicity; for example, we abbreviate $\xhat_i^\t$ as $\xhat^\t$, $y_{i,a}^\t$ as $y_a^\t$, $u_{i,a}^\t$ as $u_a^\t$, $u_i^\t$ as $u^\t$, and $\eta_{i,a}^\t$ as $\eta_a^\t$, where we recall that we use $a$ for the index for actions and $i$ for the index of players.
\begin{lemma}\label[lemma]{lem:diff_local_oftrl}
  Suppose that $T \geq 3$ and 
  consider the following OFTRL update in~\Cref{eq:oftrl_multiplayer}:
  % with 
  % the non-increasing learning rate $\eta_a^\t$ in~\Cref{eq:oftrl_multiplayer}: % ~\Cref{eq:lr_i_cor unrupt}:
  \begin{equation}
    y_{a}^\t
    =
    \argmax_{y \in \Delta_{m}}
    \set*{
      - F_a^\t(y) 
    }
    \com \quad 
    F_a^\t(y)  
    \coloneqq
    -
    \prn*{
    \inpr*{y, \util_{a}^\tm + \sum_{s=1}^{t-1} \util_{a}^\s}
    - 
    \frac{1}{\eta_a^\t}
    \phi(y)
    }
    \per
    \n
  \end{equation}
  Then, 
  \begin{equation}\label{eq:diff_local_oftrl}
    \sum_{a \in \calA}
    \nrm{y_a^\tp - y_a^\t}_{y_a^\t,F_a^\tp} 
    =
    \frac{1}{\sqrt{\eta_a^\tp}}
    \sum_{a \in \calA}
    \nrm{y_a^\tp - y_a^\t}_{y_a^\t,\phi} 
    \leq 
    \frac12 \per
    % \n
  \end{equation}
\end{lemma}
To ensure the sufficient condition for the stability of the stationary distribution, $\sum_{a \in \calA} \mu_a^\t \leq 1/2$ in \Cref{lem:anagnostides}, under our learning rate in \Cref{eq:oftrl_multiplayer}, % in \Cref{eq:lr_i_corrupt}, 
we will see in \Cref{lem:sum_mu_onehalf} that it suffcies to show $\sum_{a \in \calA} \nrm{y_a^\t - y_a^\tm}_{y_a^\tm,F_a^\t} \leq 1/2$, which is the claim of \Cref{lem:diff_local_oftrl}.
To prove \Cref{lem:diff_local_oftrl}, we will analyze the stability of the outputs $y_1^\t, \dots, y_m^\t$ of the $m (= \abs{\calA})$ experts simultaneously, which is used for swap regret minimization.
Without this new analysis, we need to choose a rather smaller learning rate to ensure the stability of the Markov chain, resulting in a swap regret bound of $O(n m^8 \log T)$ in the honest regime.
\begin{proof}[Proof of \Cref{lem:diff_local_oftrl}]
We recall that $m = \abs{\calA}$ 
and 
use $\calM_m = (\Delta_m)^m = \Delta_m \times \cdots \times \Delta_m$ to denote the Cartesian product of $m$ probability simplices.\footnote{We use $\calM_m$ to denote the Cartesian product of $m$ probability simplices since this is equivalent to the set of all row stochastic matrices.}
Define a strictly convex function $G^\tp \colon \calM_m \to \R$ by
\begin{equation}
  G^\tp(\bm{w})
  =
  G^\tp(w_1, \dots, w_m) 
  = 
  \sum_{a \in \calA} F_a^\tp (w_a)
  \per
  \nonumber
\end{equation}
Note that
for any $\hbm = (h_1, \dots, h_m) \in \calM_m$, 
the local norm $\nrm{\hbm}_{\ybm^\t, G^\tp}$ is given by
\begin{align}
  \nrm{\hbm}_{\ybm^\t, G^\tp} 
  &=
  \sqrt{
    \hbm^\top
    \diag\prn*{
      \set[\big]{
        \nabla^2 F_a^\tp(y_a^\t)
      }_{a \in \calA}
    }
    \hbm
  }
  \nn
  &=
  \sqrt{
    \sum_{a \in \calA} 
    h_a^\top \nabla^2 F_a^\tp(y_a^\t) h_a
  }
  =
  \sqrt{
    \sum_{a \in \calA} 
    \nrm{h_a}_{y_a^\t, F_a^\tp}^2
  }
  \per
  \label{eq:prod_local_norm}
\end{align}
Now from the fact that $y_a^\tp$ is the minimizer of the strongly convex function $F_a^\tp$ for each $a \in \calA$,
the point $\ybm^\t \coloneqq (y_1^\t, \dots, y_m^\t) \in \calM_m$ is the unique minimizer of $G^\tp$.
Hence from \Cref{lem:localnorm_suffcond}, to prove the claim of the lemma,
it suffices to prove that 
for any $\zbm = (z_1, \dots, z_m) \in \calM_m$ satisfying $\nrm{\zbm - \ybm^\t}_{\ybm^\t, G^\tp} = 1/(2 \sqrt{m})$, it holds that $G^\tp(\zbm) \geq G^\tp(\ybm^\t)$.
In fact, if this is proven, then \Cref{lem:localnorm_suffcond} implies $\nrm{\ybm^\tp - \ybm^\t}_{\ybm^\t, G^\tp} \leq 1/(2 \sqrt{m})$, and thus the LHS of \Cref{eq:diff_local_oftrl} is upper bounded by
\begin{equation}
  \sum_{a \in \calA} \nrm{y_a^\tp - y_a^\t}_{y_a^\t, F_a^\tp}
  \!\leq\!
  \sqrt{
    m
    \sum_{a \in \calA} \nrm{y_a^\tp \!- y_a^\t}_{y_a^\t, F_a^\tp}^2
  }
  \!=\!
  \sqrt{m} \nrm{\ybm^\tp \!- \ybm^\t}_{y^\t, G^\tp}
  \!\leq\!
  \frac12
  \com 
  \n
\end{equation}
where the first inequality follows from the Cauchy--Schwarz inequality
and 
the equality follows from~\Cref{eq:prod_local_norm}.
This is the claim of \Cref{lem:diff_local_oftrl}.

In the following, we will prove that 
for any $\zbm = (z_1, \dots, z_m) \in \calM_m$ satisfying $\nrm{\zbm - \ybm^\t}_{\ybm^\t, G^\tp} = 1/(2\sqrt{m})$, it holds that $G^\tp(\zbm) \geq G^\tp(\ybm^\t)$.
Let $h_a = z_a - y_a^\t \in \R^m$ for each $a \in \calA$.
Let us fix $a \in \calA$ and we will lower bound $F_a^\tp(z_a)$.
From Taylor's theorem,
there exists a point $\xi_a^\t = \gamma z_a + (1 - \gamma) y_a^\t$ for some $\gamma \in [0, 1]$ such that
\begin{equation}\label{eq:taylor_suma}
  F_a^\tp(z_a) 
  =
  F_a^\tp(y_a^\t)
  +
  \inpr{ \nabla F_a^\tp (y_a^\t), h_a}
  +
  \frac12 h_a^\top \nabla^2 F_a^\tp(\xi_a^\t) h_a
  \per
\end{equation}

We will lower bound the second term in the RHS of \Cref{eq:taylor_suma} below.
From the first-order optimality condition at $y_a^\t$, we have
\begin{align}
  &
  \nabla F_a^\tp (y_a^\t)
  =
  -
  \util_a^\t - \sum_{s=1}^{t} \util_a^\s
  + \frac{1}{\eta_a^\tp} \nabla \phi(y_a^\t)
  \nn 
  &= 
  \prn*{
    - \util_a^\tm - \sum_{s=1}^{t-1} \util_a^\s 
    + \frac{1}{\eta_a^\t} \nabla \phi(y_a^\t) 
  }
  - 2 \util_a^\t + \util_a^\tm + \prn*{\frac{1}{\eta_a^\tp} - \frac{1}{\eta_a^\t}} \nabla \phi(y_a^\t)
  \nn 
  &=
  - 2 \util_a^\t + \util_a^\tm + \prn*{\frac{1}{\eta_a^\tp} - \frac{1}{\eta_a^\t}} \nabla \phi(y_a^\t)
  \com
  \label{eq:firstorder_diff_local_oftrl}
\end{align}
where the last equality follows from 
$\nabla F_a^\t(y_a^\t) = - \util_a^\tm - \sum_{s=1}^{t-1} \util_a^\s + (1/\eta_a^\t) \nabla \phi(y_a^\t) = 0$.
Hence, the second term in the RHS of \Cref{eq:taylor_suma} is lower bounded as 
\begin{align}
  \inpr{ \nabla & F_a^\tp (y_a^\t), h_a}
  \geq 
  -
  \nrm*{ 
    - 2 \util_a^\t + \util_a^\tm 
    + 
    \prn*{\frac{1}{\eta_a^\tp} - \frac{1}{\eta_a^\t}} \nabla \phi(y_a^\t)
  }_{*, y_a^\t, \phi}
  \nrm{h_a}_{y_a^\t, \phi}
  \tag{by \Cref{eq:firstorder_diff_local_oftrl} and H\"{o}lder}
  \nn 
  &\geq 
  -
  \prn*{
    \nrm{ 
      - 2 \util_a^\t + \util_a^\tm 
    }_{*, y_a^\t, \phi}
    +
    \prn*{\frac{1}{\eta_a^\tp} - \frac{1}{\eta_a^\t}}
    \nrm{ 
      \nabla \phi(y_a^\t)
    }_{*, y_a^\t, \phi}
  }
  \cdot
  \frac{1}{2}
  \sqrt{\frac{\eta_a^\tp}{m}}  \nn 
  &\geq 
  -
  \frac{1}{2}
  \prn*{
    \nrm{ 
      - 2 \util_a^\t + \util_a^\tm 
    }_\infty
    +
    \frac{\sqrt{8} \prn{\xhat^\tm(a) + \xhat^{(t-2)}(a)} }{\sqrt{m \log T}}
    \nrm{ 
      \nabla \phi(y_a^\t)
    }_{*, y_a^\t, \phi}
  }
  \sqrt{\frac{\eta_a^\tp}{m}}
  \tag{
    $
  \nrm{w}_{*,y_a^\tm,\phi} 
  \leq 
  \nrm{w}_\infty
  $
  for any $w \in \R^m$ and \Cref{lem:lr_diff}}
  \nn
  &\geq 
  -
  \frac{1}{2}
  \prn*{
    2 \xhat^\t(a)
    +
    \xhat^\tm(a)
    +
    \frac{\sqrt{8} \prn{\xhat^\tm(a) + \xhat^{(t-2)}(a)} }{\sqrt{\log T}}
  }
  \sqrt{\frac{\eta_a^\tp}{m}}
  \per
  \label{eq:taylor_second_suma_mid}
  % \\
  % &\geq 
  % -
  % \prn*{
  %   \frac32
  %   +
  %   \frac{\sqrt{32}}{\sqrt{\log T}}
  % }
  % \sqrt{\frac{\eta_a^\tp}{m}}
  % \per 
  % \n
  % \label{eq:taylor_second_suma}
\end{align}
% Here, the first inequality follows from H\"{o}lder's inequality,
Here,
the second inequality follows from the triangle inequality and
\begin{align}
  \nrm{h_a}_{y_a^\t, \phi} 
  &= 
  \nrm{z_a - y_a^\t}_{y_a^\t, \phi} 
  = 
  \sqrt{\eta_a^\tp} \nrm{z_a - y_a^\t}_{y_a^\t, F_a^\tp} 
  \nn
  &
  \leq 
  \sqrt{\eta_a^\tp} \nrm{\zbm - \ybm^\t}_{\ybm_a^\t, G^\tp} 
  = 
  \frac{1}{2}
  \sqrt{\frac{\eta_a^\tp}{m}}
  \com 
  \n
\end{align}
and 
the fourth inequality from 
$\util_a^\t = \xhat^\t(a) \util^\t$,
$\nrm{\util^\t}_\infty \leq 1$
and
$\nrm{
  \nabla \phi(y_a^\t)
}_{*, y_a^\t, \phi} \leq \sqrt{m}$,
which holds since $\phi$ is $m$-self-concordant barrier and \Cref{lem:scbarrier_properties}.
Combining \Cref{eq:taylor_suma} with \Cref{eq:taylor_second_suma_mid} gives
\begin{align}
  &
  F_a^\tp(z_a)
  \nn
  &
  \geq
  F_a^\tp(y_a^\t)
  \!-\!
  \frac12
  \prn[\bigg]{
    2 \xhat^\t(a) \!+\! \xhat^\tm(a)
    \!+\!
    \frac{\sqrt{8} \prn*{\xhat^\tm(a) \!+\! \xhat^{(t-2)}(a)}}{\sqrt{\log T}}
  }
  \sqrt{\frac{\eta_a^\tp}{m}}
  +
  \frac12
  \nrm{h_a}_{\xi_a^\t, F_a^\tp}^2
  \per 
  \label{eq:Faz_lower}
\end{align}

We next consider the last term in the RHS of \Cref{eq:Faz_lower}.
We have 
$
\nrm{z_a - y_a^\t}_{y_a^\t, \phi}
=
\sqrt{\eta_a^\tp}
\nrm{z_a - y_a^\t}_{y_a^\t, F_a^\tp}
\leq
\sqrt{\eta_a^\tp}
\nrm{\zbm - \ybm^\t}_{\ybm^\t, G^\tp}
=
\frac{1}{2} \sqrt{\eta_a^\tp / m}
\leq 
\frac{1}{32} m^{-3/4}
$,
where the last inequality follows from $\eta_a^\tp \leq 1/(256 n \sqrt{m})$.
Combining this with a property of self-concordant barriers in \Cref{lem:sc_hessian_stab}, we have
\begin{equation}\label{eq:taylor_third_suma_pre}
  \nrm{h_a}_{\xi_a^\t, \phi}^2
  \!\geq\!
  \prn[\Big]{1 - \nrm{y_a^\t - \xi_a^\t}_{y_a^\t, \phi}}^2
  \nrm{h_a}_{y_a^\t, \phi}^2
  \!=\!
  \prn[\Big]{1 - \gamma \nrm{z_a - y_a^\t}_{y_a^\t, \phi}}^2
  \nrm{h_a}_{y_a^\t, \phi}^2
  \!\geq\!
  \frac{15}{16} \nrm{h_a}_{y_a^\t, \phi}^2
  \per 
  % \n
\end{equation}
Using this inequality, we can lower bound the last term in the RHS of \Cref{eq:Faz_lower} as
\begin{align}\label{eq:taylor_third_suma}
  &
  \frac12
  \sum_{a \in \calA}
  \nrm{h_a}_{\xi_a^\t, F_a^\tp}^2
  =
  \frac{1}{2 \eta_a^\tp}
  \sum_{a \in \calA}
  \nrm{h_a}_{\xi_a^\t, \phi}^2
  \nn
  &\geq 
  \frac{15}{32}
  \frac{1}{\eta_a^\tp}
  \sum_{a \in \calA}
  \nrm{h_a}_{y_a^\t, \phi}^2
  =
  \frac{15}{32}
  \sum_{a \in \calA}
  \nrm{h_a}_{y_a^\t, F_a^\tp}^2
  =
  \frac{15}{32}
  \nrm{\hbm}_{\ybm^\t, G^\tp}^2
  =
  \frac{15}{128 m}
  \com 
\end{align}
where the first inequality follows from \Cref{eq:taylor_third_suma_pre},
the third equality from \Cref{eq:prod_local_norm},
and the last equality from $\nrm{\hbm}_{\ybm^\t, G^\tp} = 1/(2\sqrt{m})$.

Therefore, summing up the inequality \Cref{eq:Faz_lower} over $a \in \calA$ and using \Cref{eq:taylor_third_suma}, we obtain
\begin{align}
  G^\tp(\zbm)
  =
  \sum_{a \in \calA}
  F_a^\tp(z_a) 
  &
  \geq
  \sum_{a \in \calA}
  F_a^\tp(y_a^\t)
  -
  \frac12 
  \prn*{
    \frac{3}{\sqrt{m}}
    +
    \sqrt{\frac{8}{\log T}}
  }
  \sqrt{\eta_a^\tp}
  +
  \frac{15}{128}
  \nn 
  &
  \geq 
  \sum_{a \in \calA} F_a^\tp(y_a^\t)
  =
  G^\tp(\ybm^\t)
  \com 
  \n
\end{align}
where in the first inequality we used the fact that $\xhat^\t, \xhat^\tm, \xhat^{(t-2)} \in \Delta_m$ are elements in the probability simplex 
and 
in the last inequality we used $T \geq 3$ and $\eta_a^\tp \leq 1/256$.
% \footnote{\red{NOTE: When $T \geq \exp(m)$, the condition $\eta^\tp \leq 1/256$ can be relaxed to $\eta^\tp \leq m / 256$, which might be useful to improve the swap regret dependence on $m$.}}
This completes the proof.
\end{proof}

Finally we will use the key lemma, \Cref{lem:diff_local_oftrl}, to prove \Cref{lem:stationary_stab},
which relates the stability of the outputs of $m_i$-experts $\set{y_{i,a}^\t}_{a \in \calA_i}$ and the stability of the output $\xhat_i^\t$ of the Markov chain defined by the transition matrix $Q_i^\t$. %  in \Cref{lem:stationary_stab}.
To prove this relation, we define
\begin{equation}
  \mu_a^\t 
  =
  \max_{b \in \calA} \,
  \abs*{
    1
    -
    \frac{y_a^\t(b)}{y_a^\tm(b)}
  }
  \com 
  \n
\end{equation}
in which the player index $i \in [n]$ is again ignored for simplicity and we will also ignore the index $i$ for lemmas involving $\mu_a^\t$.
Note that $\mu_a^\t$ is a lower bound of $\nrm{y_a^\t - y_a^\tm}_{y_a^\tm, \phi}$, which appears in the LHS of \Cref{eq:diff_local_oftrl} in \Cref{lem:diff_local_oftrl}, since
\begin{equation}\label{eq:mua2yadiff}
  \mu_a^\t 
  = 
  \max_{b \in \calA} 
  \abs*{
    1 - \frac{y_a^\t(b)}{y_a^\tm(b)}
  }
  \leq  
  \sqrt{
    \sum_{b \in \calA} 
    \prn*{
      1 - \frac{y_a^\t(b)}{y_a^\tm(b)}
    }^2
  }
  =
  \nrm{y_a^\t - y_a^\tm}_{y_a^\tm, \phi}
  \per
  % \n
\end{equation}
The following lemma is a direct consequence of \Cref{lem:diff_local_oftrl}.
\begin{lemma}\label[lemma]{lem:sum_mu_onehalf}
  Under the same assumptions as \Cref{lem:diff_local_oftrl}, 
  $\sum_{a \in \calA} \mu_a^\t \leq \sum_{a \in \calA} \nrm{y_a^\t - y_a^\tm}_{y_a^\tm, \phi} \leq 1/2$.
\end{lemma}
Thanks to this lemma, we can use the RVU bound for the log-barrier regularizer in \Cref{lem:oftrl_logbarrier} when proving \Cref{thm:indiv_swapreg}.
\begin{proof}
The claim directly follows from \Cref{lem:diff_local_oftrl}.
In fact,
from \Cref{eq:mua2yadiff} and \Cref{lem:diff_local_oftrl},
we have
\begin{equation}
  \sum_{a \in \calA} \mu_a^\t
  \leq 
  \sum_{a \in \calA} \nrm{y_a^\t - y_a^\tm}_{y_a^\tm, \phi}
  \leq 
  \frac{\sqrt{\eta_a^\t}}{2}
  \leq 
  \frac{1}{32 \sqrt{2} n^{1/2} m^{1/4}} 
  % \prn*{
  \leq 
  \frac12
  % }
  \com 
  \n
\end{equation}
where we used $\eta_a^\t \leq 1/(256 n \sqrt{m})$.
% This completes the proof.
\end{proof}

The following lemma, which is proven in \citet{anagnostides22uncoupled} based on the Markov chain tree theorem, relates the stability of $\xhat^\t$ and $\mu_a^\t$.
\begin{lemma}[{\citealt[Eq.~(26) in the proof of Lemma 4.2]{anagnostides22uncoupled}}]\label[lemma]{lem:anagnostides}
  Suppose that $\sum_{a \in \calA} \mu_a^\t \leq 1/2$.
  Then,
  \begin{equation}
    \nrm{\xhat^\t - \xhat^\tm}_1
    \leq 
    8 \sum_{a \in \calA} \mu_a^\t 
    \per
    \n
  \end{equation}
\end{lemma}

Finally from~\Cref{lem:sum_mu_onehalf,lem:anagnostides}, we can prove the following lemma relating the stability of the output of $m_i$-experts and the stability of the Markov chain defined by~$Q_i^\t$.
\begin{lemma}\label[lemma]{lem:stationary_stab}
  We assume the conditions of \Cref{lem:diff_local_oftrl}.
  Then, it holds that
  \begin{equation}%\label{eq:setationary_stab}
    \nrm{\xhat^\t - \xhat^\tm}_1^2 
    \leq 
    64 \abs{\calA}
    \sum_{a \in \calA}
    \nrm{ y_{a}^\t - y_{a}^\tm }_{y_a^\tm, \phi}^2
    \per 
    \n
  \end{equation}
\end{lemma}
This lemma will be used in the proof of \Cref{thm:indiv_swapreg} to upper bound the negative term in the RVU bound in~\Cref{eq:oftrl_logbarrier}.
\begin{proof}
From \Cref{lem:anagnostides} combined with $\sum_{a \in \calA} \mu_a^\t \leq 1/2$ in \Cref{lem:sum_mu_onehalf},
we have 
$
  \nrm{\xhat^\t - \xhat^\tm}_1
  \leq 
  8 \sum_{a \in \calA} \mu_a^\t 
$,
which implies 
\begin{equation}
  \nrm{\xhat^\t - \xhat^\tm}_1^2
  \leq 
  64
  \prn*{
    \sum_{a \in \calA} \mu_a^\t 
  }^2
  \leq 
  64 \abs{\calA} 
  \sum_{a \in \calA} \prn*{\mu_a^\t}^2
  \leq 
  64 \abs{\calA}
  \sum_{a \in \calA}
  \nrm{y_a^\t - y_a^\tm}_{y_a^\tm, \phi}^2
  % \label{eq:relative_ydiff_lower}
  \com 
  \n
\end{equation}
where 
the second inequality follows from the Cauchy--Schwarz inequality
and the last inequality from \Cref{eq:mua2yadiff}.
This completes the proof.
\end{proof}

\subsection{Upper bound on $\SwapReg_{\xhat_i,\tilde{u}_i}^T$}
Now from the preliminary lemmas provided in \Cref{subsec:pre_lemmas_swap}, we are ready to prove the following lemma,
which will lead to \Cref{thm:indiv_swapreg}.
\begin{lemma}\label[lemma]{lem:swapreg_xhat_bound}
  \Cref{alg:multiple_player_swap} achieves
  % $
  \begin{align}
    % \textstyle
    \SwapReg_{\xhat_i,\tilde{u}_i}^T
    % =
    % \sum_{a \in \calA_i} \tilde{\Reg}_{i,a}^T (M(a, \cdot))
    &\leq
    512
    n m_i^{5/2} \log T
    \!+\!
    16
    m_i
    \sqrt{
      \prn[\bigg]{
        2 n \sum_{j \neq i} 
        P_1^T(\xhat_j)
        \!+\!
        4 n \sum_{j \neq i} \hat{C}_j
        \!+\! 
        \tilde{C}_i
      }
      \log T
    }
    % }
    \!-\!
    \frac{2 n}{\sqrt{m_i}} 
    P_1^T(\xhat_i)
    \com
  % $
    \nn
    \SwapReg_{\xhat_i,\tilde{u}_i}^T
    &\leq
    512 m_i^{5/2} \log T
    +
    64 m_i
    \sqrt{
      T \log T
    }
    \per
    \n
  \end{align}
  % and 
  % $
  % \SwapReg_{\xhat_i,\tilde{u}_i}^T
  % \leq
  % 512 m_i^{5/2} \log T
  % +
  % 64 m_i
  % \sqrt{
  %   T \log T
  % }
  % \per
  % $
\end{lemma}

\begin{proof}[Proof of \Cref{lem:swapreg_xhat_bound}]
From the definition of $\set{y_{i,a}^\t}_{t=1}^T$ and \Cref{lem:oftrl_logbarrier} with $\sum_{a \in \calA} \nrm{y_{i,a}^\tp - y_{i,a}^\t}_{y_{i,a}^\t, \phi} \leq 1/2$ in \Cref{lem:sum_mu_onehalf},
for each $a \in \calA_i$ we have
\begin{align}
  &
  \tilde{\Reg}_{i,a}^T  
  = 
  \max_{y \in \Delta(\calA_i)} \sumT \inpr{y - y_{i,a}^\t, \util_{i,a}^\t}
  \nn
  &\leq 
  \frac{m_i \log T}{\eta_{i,a}^{(T+1)}}
  +
  4 \sumT \eta_{i,a}^\t \nrm{\util_{i,a}^\t - \util_{i,a}^\tm}_{*,y_{i,a}^\t,\phi}^2
  -
  \sumT \frac{1}{16 \eta_{i,a}^\t} \nrm{y_{i,a}^\tp - y_{i,a}^\t}_{y_{i,a}^\t, \phi}^2
  +
  6
  \nn 
  &\leq
  \frac{m_i \log T}{\eta_{i,\max}}
  \!+\!
  \sqrt{
    32 m_i \prn[\bigg]{4 \!+\! \sumT \nrm{\util_{i,a}^\t \!-\! \util_{i,a}^\tm}_{*,y_{i,a}^\t,\phi}^2 } \log T
  }
  \!-\!
  \sumT \frac{1}{16 \eta_{i,a}^\t} \nrm{y_{i,a}^\tp \!-\! y_{i,a}^\t}_{y_{i,a}^\t, \phi}^2
  \!+\!
  6
  \com 
  \label{eq:Regia_upper}
\end{align}
where we used the inequality $\sumT z_t / \sqrt{4 + \sum_{s=1}^{t-1} z_s} \leq 2 \sqrt{\sumT z_t}$ for $z_t \in [0,4]$ for all $t \in [T]$.
Hence, using 
\Cref{lem:swap2base_regret} and \Cref{eq:Regia_upper},
we have 
\begin{align}
  % &
  % \green{
  % \max_{M \in \calM_{m_i}} \widehat{\SwapReg}_i^T(M)
  % =
  % \sum_{a \in \calA_i} \Reg_{i,a}^T 
  % }
  % \nn 
  \SwapReg_{\xhat_i,\tilde{u}_i}^T
  &=
  \sum_{a \in \calA_i} \tilde{\Reg}_{i,a}^T
  \nn
  &\leq
  \frac{m_i^2 \log T}{\eta_{i,\max}}
  +
  \sum_{a \in \calA_i}
  \sqrt{
    32 m_i \prn*{4 + \sumT \nrm{\util_{i,a}^\t - \util_{i,a}^\tm}_{*,y_{i,a}^\t,\phi}^2 } \log T 
  }
  \nn
  &\qquad
  -
  \sumT 
  \sum_{a \in \calA_i} \frac{1}{16 \eta_{i,a}^\t} \nrm{y_{i,a}^\tp - y_{i,a}^\t}_{y_{i,a}^\t, \phi}^2
  +
  6 m_i
  \per
  \label{eq:swapreg_adv_pre}
\end{align}
Now we have
$
% \begin{align}
\nrm{\tilde{u}_{i,a}^\t - \tilde{u}_{i,a}^\tm}_{*,y_{i,a}^\t,\phi}^2 
% &
\leq 
2 \nrm{u_{i,a}^\t - u_{i,a}^\tm}_{*,y_{i,a}^\t,\phi}^2
+
4 \xhat_i^\t(a) \nrm{\tilde{c}_i^\t}_{*,y_{i,a}^\t,\phi}^2
+
4 \xhat_i^\tm(a) \nrm{\tilde{c}_i^\tm}_{*,y_{i,a}^\t,\phi}^2
\\
% \nn
% &
\leq 
2 \nrm{u_{i,a}^\t - u_{i,a}^\tm}_{*,y_{i,a}^\t,\phi}^2
+
4 \xhat_i^\t(a) \nrm{\tilde{c}_i^\t}_{\infty}^2
+
4 \xhat_i^\tm(a) \nrm{\tilde{c}_i^\tm}_{\infty}^2
% \com
% \n
% \end{align}
,
$
and thus the summation in the second term of \Cref{eq:swapreg_adv_pre} is upper bounded by
\begin{equation}\label{eq:suma_util_path}
  \sum_{a \in \calA_i} \nrm{\tilde{u}_{i,a}^\t - \tilde{u}_{i,a}^\tm}_{*,y_{i,a}^\t,\phi}^2
  \leq
  2 \sum_{a \in \calA_i} \nrm{u_{i,a}^\t - u_{i,a}^\tm}_{*,y_{i,a}^\t,\phi}^2 + 8 \tilde{C}_i
  \per
\end{equation}
The second term in \Cref{eq:suma_util_path} is further bounded as
\begin{align}
  &
  \sum_{a \in \calA_i}
  \nrm{u_{i,a}^\t - u_{i,a}^\tm}_{*,y_{i,a}^\t,\phi}^2
  =
  \sum_{a \in \calA_i}
  \nrm{u_i^\t \xhat_i^\t(a) - u_i^\tm \xhat_i^\tm(a)}_{*,y_{i,a}^\t,\phi}^2
  \nn 
  &\leq 
  2 \sum_{a \in \calA_i}
  \nrm{
    u_i^\t \xhat_i^\t(a) 
    - 
    u_i^\tm \xhat_i^\t(a)
  }_{*,y_{i,a}^\t,\phi}^2
  +
  2 \sum_{a \in \calA_i}
  \nrm{
    u_i^\tm \xhat_i^\t(a)
    -
    u_i^\tm \xhat_i^\tm(a)
  }_{*,y_{i,a}^\t,\phi}^2
  \nn 
  &=
  2 \sum_{a \in \calA_i}
  \prn{\xhat_i^\t(a)}^2
  \nrm{
    u_i^\t
    - 
    u_i^\tm 
  }_{*,y_{i,a}^\t,\phi}^2
  +
  2 \sum_{a \in \calA_i}
  \prn{\xhat_i^\t(a) - \xhat_i^\tm(a)}^2
  \nrm{
    u_i^\tm 
  }_{*,y_{i,a}^\t,\phi}^2
  \nn 
  &\leq 
  2 
  \nrm{
    u_i^\t
    - 
    u_i^\tm 
  }_\infty^2 
  +
  2 \sum_{a \in \calA_i}
  \prn{\xhat_i^\t(a) - \xhat_i^\tm(a)}^2
  \nn 
  &=
  2 
  \nrm{ u_i^\t - u_i^\tm }_\infty^2
  +
  2 
  \nrm{ \xhat_i^\t - \xhat_i^\tm }_2^2
  \per 
  \label{eq:hat_swapreg_bound}
\end{align}
Hence, from \Cref{eq:suma_util_path,eq:hat_swapreg_bound},
the sum of the second and third terms in \Cref{eq:swapreg_adv_pre} is upper bounded by
\begin{align}
  &
  \sqrt{
    32 m_i^2 \sum_{a \in \calA_i} 
    \prn*{4 + \sumT \nrm{u_{i,a}^\t - u_{i,a}^\tm}_{*,y_{i,a}^\t,\phi}^2} \log T
  }
  -
  \frac{1}{16 \eta_{i,\max}}
  \sumT 
  \sum_{a \in \calA_i} \nrm{y_{i,a}^\tp - y_{i,a}^\t}_{y_{i,a}^\t, \phi}^2
  \tag{Cauchy--Schwarz and $\eta_{i,a}^\t \leq \eta_{i,\max}$}
  \nn
  &\leq
  \sqrt{
    128 m_i^2 \prn*{
      m_i
      + 
      % \sumT
      % \nrm{ u_i^\t - u_i^\tm }_\infty^2
      P_\infty^T(u_i)
      +
      % \sumT
      % \nrm{ \xhat_i^\t - \xhat_i^\tm }_2^2
      P_2^T(\xhat_i)
      +
      4 \tilde{C}_i
    } \log T 
  }
  -
  % 32 n \sqrt{m_i} 
  \frac{1}{64 m_i \eta_{i,\max}}
  % \sumT \nrm{\xhat_i^\t - \xhat_i^\tm}_1^2
  P_1^T(\xhat_i)
  \tag{\Cref{lem:stationary_stab}, \Cref{eq:suma_util_path}, and \Cref{eq:hat_swapreg_bound}}
  \nn
  &\leq
  8 \sqrt{2} m_i^{3/2} \log T
  +
  8 \sqrt{2} m_i
  \sqrt{
    \prn{
    % \sumT
    % \nrm{ u_i^\t - u_i^\tm }_\infty^2
    P_\infty^T(u_i)
    + 
    4 \tilde{C}_i
    }
    \log T
  }
  % \nn
  % &\qquad
  +
  8 \sqrt{2} m_i
  \sqrt{
    % \sumT
    % \nrm{ \xhat_i^\t - \xhat_i^\tm }_1^2
    P_1^T(\xhat_i)
    \log T
  }
  -
  \frac{1}{64 m_i \eta_{i,\max}} 
  % \sumT \nrm{ \xhat_i^\t - \xhat_i^\tm }_1^2
  P_1^T(\xhat_i)
  \tag{subadditivity of $z \mapsto \sqrt{z}$ and $\nrm{\cdot}_2 \leq \nrm{\cdot}_1$}
  \nn
  &\leq
  8 \sqrt{2} m_i^{3/2} \log T
  +
  8 \sqrt{2} m_i
  \sqrt{
    \prn{
      % \sumT \nrm{ u_i^\t - u_i^\tm }_\infty^2
      P_\infty^T(u_i)
      +
      4 \tilde{C}_i
    }
    \log T
  }
  +
  8192 m_i^3 \eta_{i,\max} \log T
  % \nn
  % &\qquad
  -
  \frac{1}{128 m_i \eta_{i,\max}} 
  % \sumT \nrm{ \xhat_i^\t - \xhat_i^\tm }_1^2
  P_1^T(\xhat_i)
  \com 
  \n
\end{align}
where the last line follows from the inequality $b \sqrt{z} - az \leq b^2 / (4a)$ for any $a > 0$, $b \geq 0$ and $z \geq 0$.

Therefore, combining \Cref{eq:swapreg_adv_pre} with the last inequality, we obtain
\begin{align}
  % &
  \SwapReg_{\xhat_i,\tilde{u}_i}^T
  % =
  % \sum_{a \in \calA_i} \tilde{\Reg}_{i,a}^T
  % \nn
  &\leq
  \frac{m_i^2 \log T}{\eta_{i,\max}}
  +
  8 \sqrt{2} m_i^{3/2} \log T
  +
  8 \sqrt{2} m_i
  \sqrt{
    \prn{
      % \sumT \nrm{ u_i^\t - u_i^\tm }_\infty^2
      P_\infty^T(u_i)
      + 4 \tilde{C}_i
    }
    \log T
  }
  \nn
  &\qquad
  +
  8192 m_i^2 \mmax \eta_{i,\max} \log T
  -
  \frac{1}{128 m_i \eta_{i,\max}}
  % \sumT \nrm{ \xhat_i^\t - \xhat_i^\tm }_1^2
  P_1^T(\xhat_i)
  +
  6 m_i
  \label{eq:swapreg_xhat_util_upper_2}
  \\
  &\leq
  512 m_i^{5/2} \log T
  +
  64 m_i
  \sqrt{
    T \log T
  }
  \com 
  \label{eq:swapreg_xhat_util_upper_adv}
\end{align}
where the last inequality follows from $\eta_{i,\max} \leq 1/(256 n \sqrt{m_i})$.
The last inequality \Cref{eq:swapreg_xhat_util_upper_adv} is the second upper bound on $\SwapReg_{\xhat_i,\tilde{u}_i}^T$ in \Cref{lem:swapreg_xhat_bound}.

Now we will upper bound $P_\infty^T(u_i) = \sumT \nrm{u_i^\t - u_i^\tm}_\infty^2$ in \Cref{eq:swapreg_xhat_util_upper_2} to prove the first upper bound on $\SwapReg_{\xhat_i,\tilde{u}_i}^T$ in \Cref{lem:swapreg_xhat_bound}.
Let $\calA_{-i} = \times_{j \neq i} \calA_j$.
Then,
\begin{align}
  \nrm{u_i^\t - u_i^\tm}_\infty
  &=
  \max_{a_i \in \calA_i}
  \abs*{
    \sum_{a_{-i} \in \calA_{-i}} u_i(a_i, a_{-i}) \prod_{j \neq i} x_j^\t(a_j)
    -
    \sum_{a_{-i} \in \calA_{-i}} u_i(a_i, a_{-i}) \prod_{j \neq i} x_j^\tm(a_j)
  }
  \nn 
  &\leq 
  \sum_{a_{-i} \in \calA_{-i}}
  \abs*{
    \prod_{j \neq i} x_j^\t(a_j)
    -
    \prod_{j \neq i} x_j^\tm(a_j)
  }
  % \nn 
  % &
  \leq
  \sum_{j \neq i} \nrm{x_j^\t - x_j^\tm}_1
  \com 
  \label{eq:swap_2}
\end{align}
where the first inequality follows from $u_i^\t(a_i, a_{-i}) \in [-1,1] $ and the last inequality follows from the fact that the total variation of two product distributions is bounded by the sum of the total variations of each marginal distribution.
Now, for any $j \in [n]$, we have 
\begin{align}\label{eq:swap_mid_1}
  \sum_{t=1}^T \nrm{x_j^\t - x_j^\tm}_1^2
  &\leq 
  \sum_{t=1}^T
  \prn*{
    2 \nrm{x_j^\t - \xhat_j^\t}_1^2
    +
    4 \nrm{\xhat_j^\t - \xhat_j^\tm}_1^2
    +
    2 \nrm{x_j^\tm - \xhat_j^\tm}_1^2
  }
  \nn 
  &\leq 
  \sumT
  \prn*{
    4 \nrm{x_j^\t - \xhat_j^\t}_1^2
    +
    4 \nrm{\xhat_j^\t - \xhat_j^\tm}_1^2
  }
  % +
  % \nrm{x_j^{(1)} - \hat{x}_j^{(1)}}_1^2
  \nn 
  &\leq 
  8 \hat{C}_j
  +
  \sumT \nrm{\xhat_j^\t - \xhat_j^\tm}_1^2
  =
  8 \hat{C}_j
  +
  P_1^T(\xhat_j)
  % +
  % 4
  % \per 
  \com
\end{align}
where we define $x_j^{(0)} = \xhat_j^{(0)} = (1 / m) \ones$ for simplicity.
Hence, 
from \Cref{eq:swap_2}, the Cauchy--Schwarz inequality, and \Cref{eq:swap_mid_1},
 we have
% \red{$\nrm{x - y}_1^2 \leq C \nrm{x - y}_1$ for all $x, y \in \Delta_d$ なる $C > 0$ を忘れずにかけるようにする．}
\begin{align}
  &
  P_\infty^T(u_i)
  =
  \sum_{t=1}^T \nrm{u_i^\t - u_i^\tm}_\infty^2
  \leq   
  \sumT 
  \prn[\Bigg]{
    \sum_{j \neq i} \nrm{x_j^\t - x_j^\tm}_1
  }^2
  % \tag{by \Cref{eq:swap_2}}
  % \nn 
  % &
  \leq 
  (n \!-\! 1) \sum_{j \neq i} \sumT \nrm{x_j^\t - x_j^\tm}_1^2 
  % \tag{by Cauchy--Schwarz}
  \nn 
  &\leq 
  8 (n\!-\!1) \sum_{j \neq i} \hat{C}_j
  +
  4 (n\!-\!1) \sum_{j \neq i} \sumT \nrm{\xhat_j^\t - \xhat_j^\tm}_1^2
  =
  8 (n\!-\!1) \sum_{j \neq i} \hat{C}_j
  +
  4 (n\!-\!1) \sum_{j \neq i} P_1^T(\xhat_j)
  \per 
  \label{eq:swap_3}
\end{align}
Finally, combining \Cref{eq:swapreg_xhat_util_upper_2} with \Cref{eq:swap_3},
we obtain
\begin{align}
  &
  \SwapReg_{\xhat_i,\tilde{u}_i}^T
  % =
  % \sum_{a \in \calA_i} \tilde{\Reg}_{i,a}^T
  \nn
  &\leq
  \frac{m_i^2 \log T}{\eta_{i,\max}}
  \!+\!
  8 \sqrt{2} m_i^{3/2} \log T
  \!+\!
  8 \sqrt{2} m_i
  \sqrt{
    \prn*{
      4 (n\!-\!1) \sum_{j \neq i} 
      % \sumT \nrm{\xhat_j^\t - \xhat_j^\tm}_1^2
      P_1^T(\xhat_j)
      \!+\!
      8 (n\!-\!1) \sum_{j \neq i} \hat{C}_j
      \!+\! 
      4 \tilde{C}_i
    }
    \log T
  }
  \nn
  &\qquad
  +
  8192 m_i^2 \mmax \eta_{i,\max} \log T
  -
  \frac{1}{128 m_i \eta_{i,\max}} 
  % \sumT
  % \nrm{ \xhat_i^\t - \xhat_i^\tm }_1^2
  P_1^T(\xhat_i)
  +
  6 m_i
  \nn
  &\leq
  % O\prn*{
  512
  n m_i^{5/2} \log T
  +
  16 \sqrt{2}
  m_i
  \sqrt{
    \prn*{
      n \sum_{j \neq i} 
      P_1^T(\xhat_j)
      +
      2 n \sum_{j \neq i} \hat{C}_j
      + 
      \tilde{C}_i
    }
    \log T
  }
  % }
  -
  \frac{2 n}{\sqrt{m_i}} 
  P_1^T(\xhat_i)
  \com 
  \n
  % \per
  % \label{eq:swapreg_xhat_util_upper_3}
\end{align}
where we used $\eta_{i,\max} \leq 1/(256 n \sqrt{m_i})$.
This is the desired second upper bound of $\SwapReg_{\xhat_i,\tilde{u}_i}^T$ in \Cref{lem:swapreg_xhat_bound}.
\end{proof}

\subsection{Proof of \Cref{thm:indiv_swapreg}}
Now, from the preliminary lemmas provided in \Cref{subsec:pre_lemmas_swap} and \Cref{lem:swapreg_xhat_bound}, we are ready to prove \Cref{thm:indiv_swapreg}.

\begin{proof}[Proof of \Cref{thm:indiv_swapreg}]
We first prove the upper bound on $\SwapReg_{x_i,u_i}^T$.
From \Cref{lem:swapreg_xhat_bound} and \Cref{prop:swapreg_relation},
we have
\begin{equation}
  % &
  \SwapReg_{x_i,u_i}^T
  % \nn
  % &
  \!\leq\!
  512
  n m_i^{5/2} \log T
  +
  16 \sqrt{2}
  m_i
  \sqrt{
    \!
    \prn[\bigg]{
      \!
      n \sum_{j \neq i} 
      P_1^T(\xhat_j)
      \!+\!
      2 n \hat{S}_{-i}
      \!+\! 
      \tilde{C}_i
    }
    \!
    \log T
  }
  +
  C_i
  -
  \frac{2 n}{\sqrt{m_i}} 
  P_1^T(\xhat_i)
  \com 
  \label{eq:swapreg_i}
\end{equation}
where we recall that $\hat{S}_{-i} = \sum_{j \neq i} \hat{C}_j$ and $C_i = 2 \hat{C}_i + 2 \tilde{C}_i$.
Taking the summation of the inequality~\Cref{eq:swapreg_i} over the set of players $[n]$
and using the same argument as in \Cref{eq:swapreg_i_minus},
we can upper bound the sum of the swap regret, $\sum_{i \in [n]} \SwapReg_{x_i,u_i}^T$, as follows:
\begin{align}
  &
  \sum_{i \in [n]} \SwapReg_{x_i,u_i}^T
  -
  \sum_{i \in [n]}
  \prn*{
    512 n m_i^{5/2} \log T 
    +
    C_i
  }
  \nn 
  &\leq 
  16 \sqrt{2}
  \sum_{i \in [n]}
  m_i
  \sqrt{\prn*{ n \sum_{j \neq i} P_1^T(\xhat_j) + 2 n \hat{S}_{-i} + \tilde{C}_i} \log T }
  -
  \sum_{i\in [n]}
  \frac{2 n}{\sqrt{m_i}}
  P_1^T(\xhat_i)
  \nn 
  &\leq
  16 \sqrt{2}
  \sqrt{\sum_{i \in [n]} m_i^2 \log T}
  \sqrt{ n \sum_{i \in [n]} \sum_{j \neq i} P_1^T(\xhat_j) + 2 n^2 \hat{S}_{-i} + \sum_{i \in [n]} \tilde{C}_i }
  -
  \frac{2 n}{\sqrt{\mmax}}
  \sum_{i\in [n]}
  P_1^T(\xhat_i)
  \tag{Cauchy--Schwarz}
  \nn 
  &\leq
  16 \sqrt{2}
  \sqrt{ \sum_{i \in [n]} m_i^2 \log T }
  \sqrt{ n^2 \sum_{i \in [n]} P_1^T(\xhat_i) + 2 n^2 \hat{S}_{-i} + \tilde{S} }
  -
  \frac{2 n}{\sqrt{\mmax}}
  \sum_{i\in [n]}
  P_1^T(\xhat_i)
  \nn  
  &\leq
  512 n \sqrt{\mmax} \sum_{i \in [n]} m_i^2 \log T
  +
  16 \sqrt{ 
    2
    \sum_{i \in [n]} m_i^2
    \prn*{n^2 \hat{S}_{-i} + 2 \tilde{S}} \log T
  }
  -
  \frac{n}{\sqrt{m}}
  \sum_{i\in [n]}
  P_1^T(\xhat_i)
  \com 
  \label{eq:swapreg_i_minus}
\end{align}
where 
% we used $f \lesssim g$ to denote $f = O(g)$ for brevity and 
we let $\hat{S}_{-i} = \sum_{j \neq i} \hat{C}_j$ and recall that $\hat{S} = \sum_{i \in [n]} \hat{C}_i$, $\tilde{S} = \sum_{i \in [n]} \tilde{C}_i$, and $\mmax = \max_{i \in [n]} m_i$.
Here,
in the last inequality we used the subadditivity of $z \mapsto \sqrt{z}$ for $z \geq 0$ and the inequality $b \sqrt{z} - a z \leq b^2 / (4a)$ for $a > 0$, $b \geq 0$ and $z \geq 0$ with $z = \sum_{i \in [n]} P_1^T(\xhat_i)$.
Combining \Cref{eq:swapreg_i_minus} with the fact that $\sum_{i \in [n]} \SwapReg_{x_i,u_i}^T \geq 0$, we obtain
\begin{equation}
  n \sum_{i \in [n]} P_1^T(\xhat_i)
  \leq
  768 n \mmax \sum_{i \in [n]} m_i^2 \log T
  + 
  16 \sqrt{ \sum_{i \in [n]} m_i^2 \prn*{n^2 \hat{S}_{-i} + \tilde{S}} \log T}
  +
  \sqrt{\mmax} S
  \com 
  \n
\end{equation}
where we recall $S = \sum_{i \in [n]} C_i$.
Plugging the last inequality in \Cref{eq:swapreg_i}, we have 
\begin{align}
  &
  \SwapReg_{x_i,u_i}^T - 512 n m_i^{5/2} \log T
  \nn 
  &\leq 
  m_i
  \sqrt{
    \!
    \prn[\bigg]{ 
      768 n \mmax \sum_{i \in [n]} m_i^2 \log T
      \!+\! 
      16 \sqrt{ \sum_{i \in [n]} m_i^2 
      \prn*{n^2 \hat{S}_{-i} \!+\! \tilde{S}} \log T}
      \!+\!
      \sqrt{\mmax} S
      \!+\!
      2 n \hat{S}_{-i} 
      \!+\!
      \tilde{C}_i
    }
    \!
    \log T
  }
  \!+\!
  C_i
  \nn 
  &\lesssim
  m_i
  \sqrt{
    \prn[\bigg]{ 
      n \mmax \sum_{i \in [n]} m_i^2 \log T
      +
      % \frac{n}{m \log T} 
      n \hat{S}
      + 
      \sqrt{\mmax} S
    }
    \log T
  }
  +
  C_i
  \com 
  % \label{eq:swap_honest_1}
  \n
\end{align}
where in the last line we used
the AM--GM inequality and $\hat{S}_{-i} \leq \hat{S}$.
To simplify the last inequality, 
from $m_i \leq \mmax$ we get 
\begin{align}% \label{eq:swapreg_i_proof}
  \SwapReg_{x_i,u_i}^T
  % &=
  &
  \lesssim
  % O\prn*{
    \mmax
    \sqrt{
      \prn*{
        n^2 \mmax^3 \log T 
        +
        \prn*{ 
          % \frac{n}{m \log T} 
          n + \sqrt{\mmax}
        }
        \hat{S}
        +
        \sqrt{m} \tilde{S}
      }
      \log T
    }
    +
    C_i
  % }
  \nn 
  &
  % =
  % O\prn*{
  \lesssim
  n \mmax^{5/2} \log T 
  +
  m
  \sqrt{
    \prn*{
      \hat{S} \prn*{n + \sqrt{\mmax} }
      +
      \tilde{S} \sqrt{\mmax}
    }
    \log T
  }
  +
  C_i 
  \per 
  \n
\end{align}
Finally, from \Cref{lem:swapreg_xhat_bound} and \Cref{prop:swapreg_relation}, for any opponents we have 
\begin{equation}
  \SwapReg_{x_i,u_i}^T
  \lesssim
  n m_i^{5/2} \log T
  +
  m_i
  \sqrt{
    T \log T
  }
  +
  C_i 
  % }
  \per 
  \n
\end{equation}
Taking the minimum of the last two upper bounds on $\SwapReg_{x_i,u_i}^T$ gives the desired upper bound on $\SwapReg_{x_i,u_i}^T$.

We next prove the upper bound on $\SwapReg_{x_i,\util_i}^T$.
Since it holds that $\SwapReg_{x_i,\util_i}^T \geq 0$ by the definition of the swap regret, we can apply a similar argument as in the case of deriving the regret upper bound for $\SwapReg_{x_i,u_i}^T$.
From \Cref{lem:swapreg_xhat_bound} and \Cref{prop:swapreg_relation},
we have
\begin{equation}\label{eq:swapreg_i_util}
  \SwapReg_{x_i,\util_i}^T 
  \!\leq\!
  512
  n m_i^{5/2} \log T
  +
  16 \sqrt{2}
  m_i
  \!
  \sqrt{
    \!
    \prn[\bigg]{
      n \sum_{j \neq i} 
      P_1^T(\xhat_j)
      \!+\!
      2 n \hat{S}_{-i}
      \!+\! 
      \tilde{C}_i
    }
    \!
    \log T
  }
  +
  \hat{C}_i
  -
  \frac{2 n}{\sqrt{m_i}} 
  P_1^T(\xhat_i)
  \per
\end{equation}

Taking the summation of the inequality \Cref{eq:swapreg_i_util} over the set of players $[n]$, we can upper bound the sum of the swap regret, $\sum_{i \in [n]} \SwapReg_{x_i,\util_i}^T$, by
\begin{align}
  &
  \sum_{i \in [n]} \SwapReg_{x_i,\util_i}^T
  -
  \sum_{i \in [n]}
  \prn*{
    512 n m_i^{5/2} \log T 
    +
    \hat{C}_i
  }
  \nn 
  &\leq
  512 n \sqrt{\mmax} \sum_{i \in [n]} m_i^2 \log T
  +
  16 \sqrt{ 
    2
    \sum_{i \in [n]} m_i^2
    \prn*{n^2 \hat{S}_{-i} + 2 \tilde{S}} \log T
  }
  -
  \frac{n}{\sqrt{m}}
  \sum_{i\in [n]}
  P_1^T(\xhat_i)
  \com 
  \label{eq:swapreg_i_minus_util}
\end{align}
where 
% we used $f \lesssim g$ to denote $f = O(g)$ for brevity and 
we recall that $\hat{S}_{-i} = \sum_{j \neq i} \hat{C}_j$, $\hat{S} = \sum_{i \in [n]} \hat{C}_i$, and $\mmax = \max_{i \in [n]} m_i$.
Combining \Cref{eq:swapreg_i_minus_util} with the fact that $\sum_{i \in [n]} \SwapReg_{x_i,\util_i}^T \geq 0$, we obtain
\begin{equation}
  n \sum_{i \in [n]} P_1^T(\xhat_i)
  \leq
  768 n \mmax \sum_{i \in [n]} m_i^2 \log T
  + 
  16 \sqrt{ \sum_{i \in [n]} m_i^2 \prn*{n^2 \hat{S}_{-i} + \tilde{S}} \log T}
  +
  \sqrt{\mmax} \hat{S}
  \per
  \n
\end{equation}
Plugging the last inequality in \Cref{eq:swapreg_i_util}, we have 
\begin{align}
  &
  \SwapReg_{x_i,\util_i}^T - 512 n m_i^{5/2} \log T
  \nn 
  &\leq\!
  16 \sqrt{2} m_i\!
  \sqrt{
    \!
    \prn[\bigg]{ 
      768 n \mmax \sum_{i \in [n]} m_i^2 \log T
      \!+\! 
      16 \! \sqrt{ \sum_{i \in [n]} m_i^2 \prn*{n^2 \hat{S}_{-i} \!+\! \tilde{S}} \log T}
      \!+\!
      \sqrt{\mmax} \hat{S}
      \!+\!
      2 n \hat{S}_{-i}
      \!+\!
      \tilde{C}_i
    }
    \log T
  }
  \!+\!
  \hat{C}_i
  \nn 
  &\lesssim
  m_i
  \sqrt{
    \prn[\bigg]{ 
      n \mmax \sum_{i \in [n]} m_i^2 \log T
      +
      \prn*{ 
        % \frac{n}{\mmax \log T} 
        n
        + 
        \sqrt{\mmax} 
      } 
      \hat{S}
      +
      \tilde{C}_i
    }
    \log T
  }
  +
  m_i \prn*{\sum_{i \in [n]} m_i^2 \tilde{S} \log T}^{1/4}
  +
  \hat{C}_i
  \com 
  % \label{eq:swap_honest_1_util}
  \n
\end{align}
where the last line follows from the AM--GM inequality.
To simplify this inequality, 
from $m_i \leq \mmax$ we get 
\begin{align}
  \SwapReg_{x_i,\util_i}^T
  % &=
  &
  \lesssim
  % O\prn*{
    \mmax
    \sqrt{
      \prn*{
        n^2 \mmax^3 \log T 
        +
        \prn*{ 
          n + \sqrt{\mmax}
        }
        \hat{S}
      }
      \log T
    }
    +
    \mmax \prn*{n \mmax^2 \tilde{S} \log T}^{1/4}
    +
    \hat{C}_i
  % }
  \nn 
  &
  % =
  % O\prn*{
  \lesssim
    n \mmax^{5/2} \log T 
    +
    \mmax
    \sqrt{
      \prn*{
        \hat{S}
        \prn*{
          n
          +
          \sqrt{\mmax}
        }
        +
        \tilde{C}_i
      }
      \log T
    }
    +
    \prn{\tilde{S} n m^6 \log T}^{1/4}
    +
    \hat{C}_i 
  % }
  \per 
  \n
\end{align}
Finally, from \Cref{lem:swapreg_xhat_bound} and \Cref{prop:swapreg_relation}, for any opponents we have 
\begin{equation}
  \SwapReg_{x_i,\util_i}^T
  \lesssim
  n m_i^{5/2} \log T
  +
  m_i
  \sqrt{
    T \log T
  }
  +
  \hat{C}_i 
  % }
  \per 
  \n
\end{equation}
Taking the minimum of the last two upper bounds on $\SwapReg_{x_i,\util_i}^T$ completes the proof. 
\end{proof}

\section{Deferred Proofs of Lower Bounds from \Cref{sec:lower_bound}}\label{app:proof_lower_bound}

This section provides the proof of \Cref{thm:lower_bounds}.

\subsection{Lower bounds for online linear optimization over the simplex}
Here we provide bounds for online linear optimization over the simplex, which will be used in the proof of \Cref{thm:lower_bounds}.

\begin{lemma}[{\citealt[Theorem 8]{orabona15optimal}}]\label[lemma]{lem:olo_simplex_lower_logm}
  Consider online linear optimization for $T \geq 7$ rounds over the $(d-1)$-dimensional probability simplex $\Delta_m$ for $d \in [2, \exp(T/3)]$.
  Then, there exists a sequence of loss vectors $\ell^{(1)}, \dots, \ell^{(T)} \in [-1,1]^d$ with~$\nrm{\ell^\t}_\infty \leq 1$ and $u \in \Delta_d$ such that the regret of any algorithm that selects $x^\t \in \Delta_d$ at each round $t = 1, \dots, T$ is lower bounded by 
  \begin{equation}
    \Reg^T(u)
    =
    \sumT \inpr{x^\t - u, \ell^\t}
    \geq
    0.09 \sqrt{T \log d} - 2 \sqrt{T}
    \per
    \n
  \end{equation}
\end{lemma}
  
\begin{lemma}\label[lemma]{lem:olo_simplex_lower}
  Consider online linear optimization for $T \geq 1$ rounds over the $(d-1)$-dimensional probability simplex $\Delta_d$.
  Then, there exists a sequence of loss vectors $\ell^{(1)}, \dots, \ell^{(T)} \in [-1,1]^d$ with~{$\nrm{\ell^\t}_\infty \leq 1$} and $u \in \Delta_d$ such that the regret of any algorithm that selects $x^\t \in \Delta_d$ at each round $t = 1, \dots, T$ is lower bounded by 
  \begin{equation}
    \Reg^T(u)
    =
    \sumT \inpr{x^\t - u, \ell^\t}
    \geq
    \sqrt{\frac{T}{2}}
    \per
    \n
  \end{equation}
\end{lemma}
This is a very minor variant of the well-known lower bound in online linear optimization (see \eg~\citealt{hazan16introduction,orabona2019modern}).
We include the proof for completeness.
\begin{proof}
Let $\sigma^{(1)}, \dots, \sigma^{(T)}$ be i.i.d.~Rademacher random variables, \ie~$\P\prn{\sigma^\t = 1} = \P\prn{\sigma^\t = -1} = 1/2$.
Then define $z = e_1 - e_2$ and $\lsto^\t = \sigma^\t z$ for each $t \in [T]$.
Note that this $\lsto^\t$ satisfies $\nrm{\lsto^\t}_\infty = \nrm{z}_\infty = 1$.
Then, we have 
\begin{align}
  \sup_{\ell^{(1)}, \dots, \ell^{(T)}} \max_{u \in \Delta_d} \; & \Reg^T(u)
  \geq 
  \E_{\lsto^{(1)},\dots,\lsto^{(T)}} \brk*{ 
    \sumT \inpr{x^\t, \lsto^\t} 
    - 
    \min_{u \in \Delta_{d}} 
    \sumT \inpr{u, \lsto^\t} 
  }
  \nn
  &= 
  \E_{\sigma^{(1)},\dots,\sigma^{(T)}} \brk*{ 
    \sumT \sigma^\t \inpr{x^\t, z} 
    - 
    \min_{u \in \Delta_{d}} 
    \sumT \sigma^\t \inpr{u, z} 
  }
  \nn
  &= 
  \E_{\sigma^{(1)},\dots,\sigma^{(T)}} \brk*{ 
    \max_{u \in \Delta_{d}} 
    \sumT \sigma^\t \inpr{u, z} 
  }
  % &
  \geq
  \E_{\sigma^{(1)},\dots,\sigma^{(T)}} \brk*{ 
    \max_{u \in \set{e_1, e_2}} 
    \sumT \sigma^\t \inpr{u, z} 
  }
  \com 
  \n
\end{align}
where the last equality follows from the fact that $\sigma^\t$ and $- \sigma^\t$ follow the same distribution.
From $\max\set{a,b} \geq (a+b)/2 + \abs{a-b}/2$ for $a, b \in \R$, this is further lower bounded by
\begin{align}
  &
  \E_{\sigma^{(1)},\dots,\sigma^{(T)}} \brk*{ 
    \frac12
    \sumT \sigma^\t \inpr{e_1 + e_2, z} 
    +
    \frac12 \abs*{
      \sumT \sigma^\t \inpr{e_1 - e_2, z} 
    }
  }
  \nn
  &
  =
  \frac12 
  \E_{\sigma^{(1)},\dots,\sigma^{(T)}} \brk*{ 
    \abs*{
      \sumT \sigma^\t \inpr{e_1 - e_2, z} 
    }
  }
  =
  \E_{\sigma^{(1)},\dots,\sigma^{(T)}} \brk*{ 
    \abs*{
      \sumT \sigma^\t
    }
  }
  \geq 
  \sqrt{\frac{T}{2}}
  \com 
  \n
\end{align}
where the first equality follows from $\E\brk{\sigma_t} = 0$,
the last equality from $\inpr{e_1 - e_2, z} = \nrm{e_1 - e_2}_2^2 = 2$,
and the last inequality from the Khintchine's inequality.
This completes the proof.
\end{proof}

\subsection{Proof of \Cref{thm:lower_bounds}}

\begin{proof}[Proof of \Cref{thm:lower_bounds}~(i)]
We will prove $\Reg_{x,\gtil}^T = \Omega(\sqrt{\tilde{C}_{\xrm} \log \mx})$.
To simplify the analysis, we focus on the case where $\tilde{C}_{\xrm} / 2$ is an integer.
We will consider a two-player zero-sum game
where a payoff matrix is $A = 0$,
and for rounds $t = 1, \dots, \tilde{C}_{\xrm} / 2$,
the expected reward vectors are corrupted so that $\sum_{t=1}^{\tilde{C}_{\xrm} / 2} \nrm{g^\t - \gtil^\t}_\infty \leq \Ctilx$,
while no corruption occurs beyond this;
that is, $g^\t = \gtil^\t$ for $t = \tilde{C}_\xrm/2 + 1, \dots, T$, and $x^\t = \xhat^\t$, $y^\t = \yhat^\t$, and $\ltil^\t = \ell^\t$ for $t = 1, \dots, T$.
Note that in this case we have $\sumT \nrm{g^\t - \gtil^\t}_\infty \leq \tilde{C}_{\xrm}$.

From \Cref{lem:olo_simplex_lower_logm},
there exists a sequence of 
$\set{\gtil^\t}_{t=1}^{\tilde{C}_{\xrm} / 2}$ 
% and $\set{\ltil^\t}_{t=1}^{\floor{\tilde{C}_{\yrm} / 2}}$ 
such that 
\begin{equation}
  \max_{x^* \in \Delta_{\mx}}
  \sum_{t=1}^{\tilde{C}_{\xrm} / 2} \inpr{x^* - x^\t, \gtil^\t}
  \geq 
  0.09 \sqrt{ \prn{\tilde{C}_{\xrm} / 2} \log \mx } - 2 \sqrt{ \tilde{C}_{\xrm} / 2 }
  \geq 
  0.06 \sqrt{ \tilde{C}_{\xrm} \log \mx } - \sqrt{ 2 \tilde{C}_{\xrm}  }
  \per
  \n
\end{equation}
For this $\set{\gtil^\t}_{t=1}^{\tilde{C}_{\xrm} / 2}$,
since $\gtil^\t = A y^\t + \tilde{c}_{\xrm}^\t$, $A = 0$, and $\tilde{c}_{\xrm}^\t = 0$ for all $t \geq \tilde{C}_{\xrm} / 2 + 1$, we can lower bound $\Reg_{x, \gtil}^T$ as follows:
\begin{align}
  \Reg_{x, \gtil}^T
  &=
  \max_{x^* \in \Delta_{\mx}}
  \sumT \inpr{x^* \!-\! x^\t, \gtil^\t}
  =
  \max_{x^* \in \Delta_{\mx}}
  \set*{
    \sum_{t=1}^{{\tilde{C}_{\xrm} / 2}} \inpr{x^* \!-\! x^\t, \gtil^\t}
    +
    \sum_{t=\tilde{C}_{\xrm} / 2+1}^{T} \!\! \inpr{x^* \!-\! x^\t, A y^\t}
  }
  \nn
  &=
  \max_{x^* \in \Delta_{\mx}}
  \set*{
    \sum_{t=1}^{\tilde{C}_{\xrm}/2} \inpr{x^* - x^\t, \gtil^\t}
  }
  \geq 
  0.06 \sqrt{\tilde{C}_{\xrm} \log \mx} - \sqrt{2 \tilde{C}_{\xrm}}
  \com
  \n
\end{align}
which is the desired lower bonnd on $\Reg_{x, \gtil}^T$.
The lower bound for $\Reg_{y,\ltil}^T$ can be proven in a similar manner.
\end{proof}

\begin{proof}[Proof of \Cref{thm:lower_bounds}~(ii)]
We will prove $\Reg_{x,g}^T  = \Omega \prn{ \hat{C}_{\xrm} }$.
To simplify the discussion, we consider the case where $\hat{C}_\xrm / 2$ is an integer.
We will consider a two-player zero-sum game in the corrupted regime with the following property.
Consider the payoff matrix $A$ such that all elements in the first $\mx - 1$ rows are $1$, and all elements in the $\mx$-th row are $0$.
The strategies is corrupted so that 
$x^\t = \xhat^\t + \hat{c}_\xrm^\t = e_{\mx}$
for each round $t = 1, \dots, \hat{C}_{\xrm}/2$,
and no corruption occurs beyond this.
Then, the corruption level in the strategies of $x$-player is upper bounded by $\sumT \nrm{\hat{c}_\xrm^\t}_1 = \sum_{t=1}^{\hat{C}_\xrm/2} \nrm{e_{\mx} - \xhat^\t}_1 \leq \hat{C}_\xrm$.
From the construction of the payoff matrix $A$, 
for each $t = 1, \dots, \hat{C}_{\xrm} / 2$,
we also have $A y^\t = \ones - e_{\mx}$ and $\inpr{x^\t, A y^\t} = 0$ since $A^\top x^\t = 0$.

From this construction of the corrupted game, for any $x^* \in \Delta_{\mx}$, we have
\begin{equation}
  \sum_{t=1}^{\hat{C}_\xrm/2} \inpr{x^* - x^\t, g^\t}
  =
  \sum_{t=1}^{\hat{C}_\xrm/2} \inpr{x^*, A y^\t}
  = 
  \sum_{t=1}^{\hat{C}_\xrm/2} \inpr{x^*, \ones - e_{\mx}}
  =
  \frac{\hat{C}_\xrm}{2} \prn{1 - x^*(\mx)}
  \com 
  % \nonumber
  \label{eq:lower_bound_2_reg_Cover2}
\end{equation}
where we used $\inpr{x^\t, A y^\t} = 0$, $A y^\t = \ones - e_{\mx}$, and $x^* \in \Delta_{\mx}$.
Therefore, 
\begin{align}
  \Reg_{x,g}^T
  &=
  \max_{x^* \in \Delta_{\mx}}
  \sumT \inpr{x^* - x^\t, g^\t}
  \nn
  &=
  \max_{x^* \in \Delta_{\mx}}
  \set*{
    \sum_{t=1}^{\hat{C}_\xrm/2} \inpr{x^* - x^\t, g^\t}
    +
    \sum_{t=\hat{C}_\xrm/2+1}^T \inpr{x^* - x^\t, A y^\t}
  }
  \nn
  &=
  \max_{x^* \in \Delta_{\mx}}
  \set*{
    \frac{\hat{C}_\xrm}{2} \prn{1 - x^*(\mx)}
    +
    \sum_{t=\hat{C}_\xrm/2+1}^T \inpr{x^* - x^\t, \ones - e_{\mx}}
  }
  \tag{by \Cref{eq:lower_bound_2_reg_Cover2}}
  \nn
  &=
  \max_{x^* \in \Delta_{\mx}}
  \set*{
    \frac{\hat{C}_\xrm}{2} \prn{1 - x^*(\mx)}
    +
    \sum_{t=\hat{C}_\xrm/2+1}^T \prn{x^\t(\mx) - x^*(\mx)}
  }
  \geq 
  \frac{\hat{C}_\xrm}{2}
  \com
  \n
\end{align}
where the last inequality follows by choosing $x^*$ with $x^*(\mx) = 0$.
This completes the proof of $\Reg_{x,g}^T = \Omega(\hat{C}_\xrm)$.
The lower bound of $\Reg_{y,\ell}^T \geq \hat{C}_\yrm / 2$ can be proven in a similar manner.
\end{proof}

\begin{proof}[Proof of \Cref{thm:lower_bounds}~(iii)]
Let $\kappa = 1/4$.
Then, it suffices to prove that there exists a two-player zero-sum game with 
$\sumT \nrm{y^\t - \yhat^\t}_1 \leq \Chaty$
such that
$\Reg_{\xhat,g}^T < \kappa \sqrt{\Chaty}$ implies $\Reg_{\yhat,\ell}^T \geq \kappa \sqrt{\Chaty}$.
To simplify the discussion, we consider only the case where $\Chaty / 2$ is an integer.

We will consider a two-player zero-sum game in the corrupted regime with the following property.
Consider the payoff matrix 
$
A = 
\begin{pmatrix}
  1 & 0 & -1 \\
  0 & 1 & -1 
\end{pmatrix}
.
$

The suggested strategies $\hat{y}^\t$ is corrupted to ${y}^\t$ for rounds $t = 1, \dots, \hat{C}_{\yrm}/2$,
and no corruption occurs beyond this.
Then, the corruption level in the strategies of $y$-player is upper bounded by $\sumT \nrm{\hat{c}_\yrm^\t}_1 = \sum_{t=1}^{\hat{C}_\yrm/2} \nrm{\hat{c}_\yrm^\t}_1 \leq \hat{C}_\yrm$.
In particular, we take the strategies $\hat{y}^{(1)}, \dots, \hat{y}^{(\hat{C}_y / 2)}$ of $y$-player as follows.
From \Cref{lem:olo_simplex_lower},
there exists a sequence of reward vectors $\set{\check{g}^\t}_{t=1}^{\Chaty / 2}$ such that $\nrm{\check{g}^\t}_\infty \leq 1$ and
\begin{equation}\label{eq:reg_to_chaty_lower}
  \max_{x \in \Delta_{\mx}}
  \sum_{t=1}^{\Chaty/2} \inpr{x - \xhat^\t, \check{g}^\t}
  \geq 
  \frac12
  \sqrt{\Chaty}
  .
\end{equation}
Then for this $\set{\check{g}^\t}_{t=1}^{\Chaty / 2}$, 
we take $y^\t$ satisfying $\check{g}^\t = A y^\t$.
This is indeed possible since
for any $g \in [-1, 1]^2$ there exists $y \in \Delta_{\my}$ such that
$
g
=
A y
= 
\begin{pmatrix}
  y(1) - y(3)  \\
  y(2) - y(3)
\end{pmatrix}
.
$

Finally, we are ready to prove $\Reg_{\tilde{y},\ell}^T \geq \kappa \sqrt{\hat{C}_y}$.
Now, from $\Reg_{\yhat,\ell}^T < \kappa \sqrt{\Chaty}$ and the choice of the payoff matrix $A$, we have
\begin{equation}\label{eq:y_suboptimal_upper}
  \kappa \sqrt{\Chaty} 
  > 
  \Reg_{\yhat,\ell}^T 
  \geq 
  \sumT \prn{1 - \yhat^\t(3)}
  \geq 
  \sum_{t=\Chaty/2+1}^{T} \prn{1 - \yhat^\t(3)}
  =
  \sum_{t=\Chaty/2+1}^{T} \prn{\yhat^\t(1) + \yhat^\t(2)}
  \per
\end{equation}
We also observe that the instantaneous regret against any $x \in \Delta_{\mx}$ is lower bounded by
\begin{equation}\label{eq:inst_reg_lower}
  \inpr{x - \xhat^\t, A y^\t}
  =
  (x(1) - \xhat^\t(1)) y^\t(1) + (x(2) - \xhat^\t(2)) y^\t(2)
  \geq
  - \prn{ y^\t(1) + y^\t(2) }
  \per
\end{equation}
By combining \Cref{eq:y_suboptimal_upper} and \Cref{eq:inst_reg_lower},
for any $x \in \Delta_{\mx}$ the regret compared to a point $x$ after round $t = \Chaty/2 + 1$ is lower bounded by
\begin{equation}\label{eq:reg_after_chaty_lower}
  \sum_{t=\Chaty/2+1}^T \inpr{x - \xhat^\t, A y^\t}
  \geq
  - \sum_{t=\Chaty/2+1}^T \prn{ y^\t(1) + y^\t(2) }
  \geq 
  - \kappa \sqrt{\Chaty}
  \per
\end{equation}
Therefore, from \Cref{eq:reg_to_chaty_lower,eq:reg_after_chaty_lower}, we obtain 
\begin{equation}
  \Reg_{\ytil,\ell}^T
  \geq 
  \frac12
  \sqrt{\Chaty} - \kappa \sqrt{\Chaty}
  =
  \prn*{\frac12 - \kappa} \sqrt{\Chaty}
  =
  \kappa \sqrt{\Chaty}
  \com 
  \n
\end{equation}
which completes the proof.
\end{proof}

\end{document}